\RequirePackage{amsmath}
\def\draft{0} % When set to 0, author notes are removed.
\def\anonymous{0}
\def\toc{0} %Table of contents. Shouldn't be used in proceedings.
\def\iacr{1}

\documentclass[runningheads]{llncs}
\newcommand{\norm}[1]{}
\usepackage[utf8]{inputenc}
\usepackage[mathscr]{eucal}
\usepackage{graphicx}
\usepackage{listings}
\usepackage{amssymb}
\usepackage{dsfont}
\usepackage{mathtools}
\usepackage{algorithm}% http://ctan.org/pkg/algorithm
\usepackage{algpseudocode}
\usepackage{xcolor}
\usepackage{multirow}
\usepackage{amsmath}
\usepackage[english]{babel}
\usepackage{wasysym}
\usepackage{qip}
\usepackage{qm}
\usepackage{csquotes}
\usepackage{environ}

\usepackage[clock]{ifsym}

\usepackage[
	lambda,
	operators,
	advantage, 
	sets,
	adversary,
	landau,
	probability,
	notions,
	logic,
	ff,
	mm,
        primitives,
        events,
        complexity,
        asymptotics, 
        keys]{cryptocode}

\usepackage[tikz]{bclogo}
\usepackage[draft=false]{hyperref}
\usepackage[normalem]{ulem}
\usepackage[capitalise]{cleveref}
\crefname{claim}{Claim}{Claims}
\crefname{construct}{Construction}{Constructions}

\usepackage{comment}
\usepackage{paralist}
\usepackage{wrapfig}
\usepackage[normalem]{ulem}
\newcommand{\hildd}[1]{\ensuremath{\mathscr{#1}}}

\newcommand{\botOWF}{\bot\text{-}\mathsf{OWF}}
\newcommand{\botUOWHF}{\bot\text{-}\mathsf{UOWHF}}
\ifnum\draft=1
    \newcommand{\louis}[1]{{\color{red}\textbf{Louis}: #1}}
    \newcommand{\mo}[1]{{\color{blue}\textbf{Mohammed}: #1}}
    \newcommand{\amit}[1]{{\color{green}\textbf{Amit}: #1}}
    \newcommand{\anote}[1]{{\color{green}\textbf{Amit}: #1}}
    \newcommand{\onote}[1]{{\color{orange}\textbf{Or}: #1}}
    \newcommand{\lnote}[1]{{\color{magenta}\textbf{Lior}: #1}}
\else
    \newcommand{\louis}[1]{}
    \newcommand{\mo}[1]{}
    \newcommand{\amit}[1]{}
    \newcommand{\anote}[1]{}
    \newcommand{\onote}[1]{}
    \newcommand{\lnote}[1]{}
\fi

\newcommand{\UOWHF}{\ensuremath{\textsf{UOWHF}}}

\newcommand{\OWF}{\ensuremath{\textsf{OWF}}}
\newcommand{\PDOWF}{\ensuremath{\textsf{PD-OWF}}}

\renewcommand{\qed}{{\hfill{$\blacksquare$}}}
\usepackage{xspace}
\DeclareMathAlphabet{\mathpzc}{OT1}{pzc}{m}{it}

\newcommand{\PDPRG}{\mathsf{PD}\text{-}\mathsf{PRG}}
\newcommand{\pdprg}{{pseudodeterministic}\text{ }{pseudorandom generator}\xspace}
\newcommand{\Klam}{\mathcal{K}_\secpar}
\newcommand{\Gvote}{G^{\textsf{vote}}}

\newcommand{\PRS}{\mathsf{PRS}}
\newcommand{\PRG}{\mathsf{PRG}}
\newcommand{\PRF}{\mathsf{PRF}}
\newcommand{\SPRS}{\mathsf{SPRS}}
\newcommand{\LPRS}{\mathsf{LPRS}}

\newcommand{\botPRG}{\bot\text{-}\mathsf{PRG}}
\newcommand{\botPRF}{\bot\text{-}\mathsf{PRF}}
\newcommand{\botPR}{\bot\text{-pseudorandomness}}

\newcommand{\PRFS}{\mathsf{PRFS}}

\newcommand{\isbot}{\mathsf{Is}\text{-}\bot}
\newcommand{\expt}{\mathsf{Expt}}
\newcommand{\distingexpt}{\mathsf{Distinguish}\text{-}\expt}
\newcommand{\prfdistingexpt}{\distingexpt^\mathcal{F}_D}
\newcommand{\ch}{\mathsf{Ch}}
\newcommand{\hybrid}{\mathsf{Hybrid}}
\newcommand{\treeroot}{\mathsf{Root}}
\newcommand{\treeprf}{\mathsf{Tree}\text{-}\mathsf{PRF}}

\newcommand{\XX}{\mathbb{X}}

\newcommand{\EE}{\mathbb{E}}

\newcommand{\Xx}{\mathbf{X}}
\newcommand{\Yy}{\mathbf{Y}}
\newcommand{\Zz}{\mathbf{Z}}
\newcommand{\Cc}{\mathbf{C}}

\newcommand{\gprob}{\mathsf{guess}\text{-}\mathsf{prob}}

\newcommand{\skgen}{\mathsf{SKGen}}
\newcommand{\pkgen}{\mathsf{PKGen}}
\renewcommand{\enc}{\mathsf{Enc}}
\renewcommand{\dec}{\mathsf{Dec}}

\newcommand{\ct}{\mathsf{ct}}
\newcommand{\qpke}{\text{tamper-proof quantum public key encryption}}
\newcommand{\QPKE}{\mathsf{QPKE}}
\newcommand{\cpatamp}{\mathsf{IND}\text{-}\mathsf{pkT}\text{-}\mathsf{CPA}}
\newcommand{\PDPRF}{\mathsf{PD}\text{-}\mathsf{PRF}}
\newcommand{\BPRG}{\mathrm{BPRG}}
\spnewtheorem{construct}{Construction}{\bfseries}{\itshape}
\spnewtheorem*{thm}{Theorem}{\bfseries}{\itshape}
\spnewtheorem*{lem}{Lemma}{\bfseries}{\itshape}
\usepackage{ifdraft}
\ifdraft{\newcommand{\authnote}[3]{{\color{#3} {\bf  #1:} #2}}}{\newcommand{\authnote}[3]{}}  %%%%%Amit: need to incorporate to the style here. See \louis and \mo
\begin{document}
\title{
Signatures From Pseudorandom States via \texorpdfstring{$\bot$}{bot}-PRFs
%A Pseudo-Deterministic Title
% New applications of pseudodeterministic QPRG via recognizable abort based on pseudorandom states.\\
% Constructions and applications of QPRGs and QPRFs with recognizable abort.\\
% Digital signatures and tamper proof qPKE from pseudorandom states\\
% Quantum cryptography from new flavors of pseudodeterminism\\
% Regconizable abort PRGs and PRFs from pseudorandom states, and their applications:
% ecognizable abort quantum pseudorandomness \\
% PRFs with recognizable abort from pseudorandom states.\\
% PRS->recognizable abort PRF-> digital signatures and tamper proof PKE.\\
}
\ifnum\anonymous=0
    \author{Mohammed Barhoush\inst{1} \and Amit Behera \inst{2} \and Lior Ozer\inst{2} \and Louis Salvail\inst{1} \and Or Sattath\inst{3}}
    \authorrunning{Barhoush et al.}
     \institute{Universit\'e de Montr\'eal (DIRO), Montr\'eal, Canada\\   \email{mohammed.barhoush@umontreal.ca}, \email{salvail@iro.umontreal.ca}  \and
     Department of Computer Science, Ben-Gurion University of the Negev, Beersheba, Israel \\
     \email{behera@post.bgu.ac.il}, \email{lioroz@post.bgu.ac.il}, \email{sattath@bgu.ac.il}}
\else

    \author{}
    \institute{}
\fi
        
\maketitle

\begin{abstract}
%\onote{We introduce PRGs and PRFs with recognizable abort, based on pseudodeterministic PRGs. We use our constructions to show digital signatures and tamper-proof public key encryption. }

Different flavors of quantum pseudorandomness have proven useful for various cryptographic applications, with the compelling feature that these primitives are potentially weaker than post-quantum one-way functions. 
Ananth, Lin, and Yuen (2023) have shown that logarithmic pseudorandom states can be used to construct a pseudo-deterministic $\PRG$ : informally, for a fixed seed, the output is the same with $1-1/poly$ probability.

In this work, we introduce new definitions for $\botPRG$ and $\botPRF$. The correctness guarantees are that, for a fixed seed, except with negligible probability, the output is either the same (with probability $1-1/poly$) or recognizable abort, denoted $\bot$. Our approach admits a natural definition of multi-time $\PRG$  security, as well as the adaptive security of a $\PRF$. We construct a $\botPRG$ from any pseudo-deterministic $\PRG$ and, from that, a $\botPRF$.

Even though most Minicrypt primitives, such as symmetric key encryption, commitments, MAC, and length-restricted one-time digital signatures, have been shown based on various quantum pseudorandomness assumptions, digital signatures remained elusive. Our main application is a (quantum) digital signature scheme with classical public keys and signatures, thereby addressing a previously unresolved question posed in Morimae and Yamakawa's work (Crypto, 2022). Additionally, we construct CPA secure public-key encryption with tamper-resilient quantum public keys.

% Quantum cryptography has shed light on the potential redundancy of one-way functions (OWFs) in certain applications, where it would otherwise be required classically. It has been revealed that many such cryptographic primitives can be constructed using the concept of pseudorandom quantum states (PRSs), which represent a potentially weaker foundational assumption.

% In this work, we aim to investigate and provide a more comprehensive understanding of the relationship between PRSs and OWFs. To begin, we introduce the novel notion of pseudo-deterministic one-way functions (PD-OWFs). These are similar to conventional OWFs with the exception that the output is deterministic only with high probability.

% Our study demonstrates that PD-OWFs can indeed be derived from logarithmic-length PRSs (SPRSs), utilizing recent developments in the construction of pseudo-deterministic pseudorandom functions from SPRSs. As a consequence, we present a (many-time) digital signature scheme for classical messages with classical signatures, thereby addressing a previously unresolved question posed in Morimae and Yamakawa's work (Crypto, 2022).
\end{abstract}
  
\ifnum\toc=0
    \setcounter{secnumdepth}{3}
    \setcounter{tocdepth}{3}
    \newpage
    \tableofcontents 
    \newpage
\fi

\onote{Todo:
}
  
\anote{For next version:
\begin{enumerate}
\item Mo: Remove colon notation and add a remark to clarify that bot-functions are QPT. 
\item Fix definition of F. 
    \item Add the construction of CPA symmetric encryption from $\botPRF$.
    \item \textbf{Generalization lemma for bot primitives that takes care of repetitions as well.} I think it could be possible to formulate a generalization lemma for the $\botPRF$ and $\botPRG$. Something along the form that if there is an experiment $\exp^F$ which remains secure when $F$ is instantiated with a classical pseudorandom function, then there exists a repeated version of the experiment ${\exp'}^F$, that is secure when $F$ is instantiated with a $\botPRF$.
    Currently, we construct digital signatures from a one-time strongly unforgeable digital signature that we build from PD-OWHFs and then use the $\botPRF$ to argue security for the digital signature scheme and then amplify the correctness by repeating several times. I think currently we use the security of the one-time strongly unforgeable digital signature and the $\botPRF$ guarantees together. However, if we can separate these two steps out, then it could be possible to use the generalization lemma to directly do the step involving the $\botPRF$ guarantees and the amplification. 
    \item Equivalence between $\botPRG$ and $\SPRS$. We can use the output of the $\botPRG$ as an inverse to the extractor, and hopefully, have the tomography of the quantum state. Given the classical description of the state, it should be easy to construct the state efficiently (since it has only logarithmic size).  
    \item Length extensions of $\botPRG$. Seems easy: the classical trick where we apply the PRG again and against seems to work.
    \item Length extension of $\SPRS$. This can be done by using the two results above: use the equivalence to construct a $\botPRG$ from a $\SPRS$; apply a length extension; go back to a $\SPRS$ of longer length using the equivalence.
    \item Constructing $\LPRS$ from $\botPRF$. This is still wide open. But this is a great motivation for security with superposition queries. If this was the case, we could have used the JLS construction.
    \item We know PD-PRG implies bot-PRG. But what about the other way? Yes, it does, by simply treating bot as all zeroes and amplifying the resulting weak $\PDPRG$. And what about $\botPRF$ implies $\PDPRF$? Probabyly yes as well, because the definition of $\PDPRF$ is so weak, and therefore, Yao's XOR lemma or its variants could be applies as well, by replacing the bots with all 0's. 
\end{enumerate}

}
\anote{Self reminder:
\begin{enumerate}
    \item \sout{Move the sections around as decided to see if it fits within 30 pages.}
    \item \sout{Remember to cross check that all the notations such as $\prg_b$ used in the bot prf section is properly defined.}
    \item \sout{Define the parallel repetition game for the multi-time pseudorandomness of $\botPRG$.}
    \item \sout{Lior: Check if the name $\botPRG$ and $\botPRF$ is used everywhere.}
    \item \sout{Make sure that all the properties/result I used in the GGM proof have been stated or proven in the $\bot$ $\PRG$  section. }
    \item \sout{Check the anotes for self reminders.}
\end{enumerate}

}
\section{Introduction}
%\anote{Currently, we do not have a construction of $\PRS$ of any output sizes from $\botPRG$ or $\botPRF$. Is it possible to construct $\PRS$ of any size from $\botPRG$ or $\botPRF$? The conversion of log-length state to poly-length strings via tomography seems pretty invertible, for example the bit commitment scheme with classical communication from $\SPRS$. A constant length extension of $\botPRG$ or $\PDPRG$ may be possible. So if we take a $\SPRS$, construct a $\botPRG$ or $\PDPRG$ from it and then double its output size (without changing the input size of course), and then finally use the inverse map to get back a $\SPRS$ with double the length.}
Introduced by Ji, Liu, and Song in \cite{JLS18}, pseudorandom quantum state ($\PRS$) generators capture pseudorandomness in the quantum realm similar to pseudorandom generators ($\PRG$) in the classical setting \cite{BM84}.

Intuitively, a $m$-$\PRS$ generator \cite{JLS18} is a quantum polynomial-time algorithm that takes an input $k\in \{0,1\}^\lambda$ and outputs an $m$-qubit state $\ket{\phi_k}$ that is indistinguishable from a random Haar state, even when provided with 
polynomially many copies of $\ket{\phi_k}$. We consider three $\PRS$  sizes: 
\begin{enumerate}
    \item $m=c\cdot \log(\lambda)$ with $c\ll 1$.
    \item $m=c\cdot \log(\lambda)$ with $c\geq 1$, which we call \emph{short pseudorandom states ($\SPRS$s)}.
    \item $m=\Omega (\lambda)$, which we call \emph{long pseudorandom states ($\LPRS$s)}.
\end{enumerate} 

These different cases behave quite differently and provide different applications. The first type of $\PRS$s can be constructed unconditionally~\cite{ZO20} and are cryptographically inept, similar to how $\PRG$s of output size logarithmic in input size. However, despite the small output size, this is not the case for the second type, i.e., $\SPRS$s because many cryptographic primitives such as symmetric encryption and bit commitments with classical communications can be built from $\SPRS$s~\cite{ALY23}. Finally, $\LPRS$s have shown numerous applications including quantum pseudo-encryption, quantum bit-commitments, length-restricted one-time signatures with quantum public keys, and private-key quantum money \cite{JLS18,AQY22,MY22}.  %Meanwhile, neither $\LPRS$s nor $\SPRS$s have been based on or separated from the other.
 Unlike the classical analogues, the relationship between $\SPRS$ and $\LPRS$ is not well understood since no implication is known in either direction.
%$\SPRS$s and $\LPRS$s seem to provide different applications but $\SPRS$s seems to be more powerful. For instance, Ananth, Lin, and Yuen~\cite{ALY23} showed that $\SPRS$s could be used to generate pseudorandom classical strings, but this approach fails when relying on $\LPRS$s. Indeed, our reliance on this result is the reason why our constructions are based on $\SPRS$s throughout. 

In relation to $\PRG$s, both $\LPRS$s and $\SPRS$s can be constructed from $\PRG$s \cite{JLS18,ZO20}. In the other direction, Kretschmer~\cite{K21} showed a black-box separation between $\PRG$s and $\LPRS$s\footnote{The separation is with respect to post-quantum one-way function, but since any $\PRG$ is also a one-way function, the separation holds for $\PRG$s as well.}, but this separation does not imply a separation between $\PRG$s and $\SPRS$s. 

Despite the plethora of cryptographic application of $\PRS$, a significant gap persists between the cryptographic applications of $\PRG$s and these quantum pseudorandom notions in cryptography. Specifically, (many-time) Digital Signatures (DS), arguably one of the most fundamental primitives in cryptography, can be based on $\PRG$s \cite{CW97,BGH23}, but there are no construction of DS based solely on $\PRS$s (of any size). Concurrent to our work, it was shown in~\cite{CM24} that DS\footnote{The separation in~\cite{CM24} also works for the case of digital signatures with quantum public keys.} cannot be constructed in a black-box manner from $\LPRS$s. This leaves the question of whether we can construct DS  based on $\SPRS$s, which is the main question that we study in this work.

\subsection{Our Results}
 %In classical cryptography, $\PRG$s with longer output size are cryptographically more useful.
 Since DS are known to be separated from $\LPRS$s~\cite{CM24}, it is natural to ask why it might be possible to build it from $\PRS$ from smaller output size such as $\SPRS$. So what makes $\SPRS$ more useful than $\LPRS$ in constructing DS? The main advantage that $\SPRS$ has over $\LPRS$ is that the logarithmic output size of the $\SPRS$ unlocks the possibility to deterministically extract randomness by doing efficient tomography on the random state using only polynomial number of copies of the state. This property of $\SPRS$ was utilized in the work of Ananth, Lin, and Yuen~\cite{ALY23} to generate pseudorandom strings, yielding a cryptographically useful primitive known as pseudodeterministic $\PRG$s ($\PDPRG$) \cite{ALY23}. These functions have classical inputs and outputs yet incorporate quantum computations, which imply that the output is pseudo-deterministic (see \cref{def:aly23}). Similarly, they define pseudodeterministic $\PRF$s ($\PDPRF$) and show a construction based on logarithmic size $\PRFS$. %\cite{ALY23} shows various applications of $\PDPRG$s and $\PDPRF$s were 

Our work further studies the cryptographic application $\PDPRG$ to achieve similar applications as those obtained from $\PRG$s, including DS. To this end, we first address the following question: under what circumstances can a $\PRG$  be replaced with a $\PDPRG$? 
\begin{comment}
We find that as long as the construction only evaluates the $\PRG$  on randomly sampled inputs, then the $\PRG$  can be safely replaced with a $\PDPRG$. This result, which we call the Generalization Lemma (see \cref{lem:gener PRG}), lets us obtain pseudo-deterministic one-way functions from $\PDPRG$, and could be of independent interest. 
\end{comment}
Unfortunately, in many constructions involving $\PRG$s,  %invoke the generator on structured inputs. Here,
the pseudodeterminism is problematic and prevents us from substituting $\PRG$s with $\PDPRG$. While completely eliminating the pseudodeterminism appears difficult, we show that the non-deterministic evaluations, i.e., those that are obtained with small probability, can be recognized. This enables us to replace these evaluations with $\perp$, naturally leading to the novel notion of $\prg$ with recognizable abort\footnote{The term \emph{recognizable abort} was used to describe a similar phenomenon in~\cite{AQY22}, which we follow.}, denoted $\botPRG$. These functions, except with negligible probability, yield fixed outputs on the majority of inputs, and for the rest, the output is either some fixed value or $\bot$. 
We construct $\botPRG$ from $\PDPRG$ with multi-time $\botPR$ guarantees, which ensures security even if the same seed is evaluated multiple times (see the technical overview for further details), giving the following result:
%we show that pseudodeterminstic errors can be made recognizable, and therefore abort in such cases. 
%This enables us to convert $\PDPRG$ into a novel function which we term $\botPRG$. 

%we define a form of $\PRG$ with recognizable abort, denoted $\botPRG$, in which the pseudodeterminstic errors are recognizable. This can be useful in applications where we can abort in such cases. Furthremore, we introduce a novel pseudorandomness guarantee for multi-time evaluations, which is discussed in more detail in the technical overview. We show:
\begin{theorem}[Informal version of \cref{thm:pdprg_implies_botPRG}] $\PDPRG$ implies $\botPRG$.
\label{thm:informal_pd_prg_implies_bot_prg}
\end{theorem}
% Our second contribution is to convert $\botPRG$ into (adaptively-secure) $\botPRF$ following a standard approach for converting traditional $\PRG$s into $\PRF$s \cite{GGM86}.
An important example where the $\PRG$s are invoked on structured inputs is the GGM construction~\cite{GGM86}, which constructs a $\PRF$  from a $\PRG$. Indeed, we use the multi-time $\botPR$ of our $\botPRG$, and show for the first time a similar reduction:
\begin{theorem}[Informal version of \cref{thm:bot_prf_from_bot_prg}]
$\botPRG$ implies $\botPRF$.    \label{thm:informal-bot_prf_from_bot_prg}
\end{theorem}
%$\botPRF$ are natural extension of the $\botPRG$. 
Just like the definition of $\botPRG$ manages to incorporate multi-time $\botPR$, the definition of $\botPRF$ provides security even against adaptive queries. On the contrary, the existing pseudorandomness definition of $\PDPRF$~\cite{ALY23} is nonadaptive, thereby restricting the cryptographic applications of $\PDPRF$. We remove these restrictions to achieve new applications.

Furthermore, $\botPRG$ are, in some sense, better well-behaved than $\PDPRG$. This enables us to use them to build $\botUOWHF$s. Our proof follows a similar avenue to the deterministic construction given in \cite{HHR+10}, however, non-trivial adaptations and new notions of entropy are required to deal with the non-determinism of $\bot$ evaluations. \ifnum\iacr=1 At the end, we prove that $\botPRG$ implies $\botUOWHF$. \else We obtain the following result:

\begin{theorem}[Informal version of \cref{thm:uowhf}]
    $\botPRG$ implies $\botUOWHF$.
\end{theorem}
\fi

With $\botPRF$ and $\botUOWHF$ in hand, we build digital signatures with some desirable properties:
\begin{theorem}[Informal version of \cref{thm:main sig}]
    \label{thm:informal_pd_prg_implies_signatures}
    $\PDPRG$ implies an unforgeable digital signature scheme, with a classical secret key, a classical public key, and classical signatures.
\end{theorem}

The construction follows a similar approach to the standard construction of digital signatures from $\PRF$s and $\UOWHF$s~\cite{L79,NY89,R90} (see \cite[Chapter 6]{Gol04} for a modern presentation). %Note that both these assumptions can be based on $\SPRS$s. 

Interestingly, \cite{CM24} recently demonstrated that digital signatures are black-box separated from $\LPRS$. Along with our result (\Cref{thm:informal_pd_prg_implies_signatures}), this implies that $\SPRS$s are fully black-box separated from $\LPRS$s. 
%Due to the abort errors of the underlying $\botPRF$, the resulting scheme has only $1-\frac{1}{\poly}$ correctness. 
%We use standard error reduction techniques through repetition to achieve negligible error.  

We then turn to encryption. Notably, \cite{KMNY23} have constructed QPKE with tamper-resilient quantum public keys from \emph{deterministic} strongly unforgeable digital signatures. Remarkably, security is guaranteed even if an adversary can tamper with the quantum public keys, provided a classical verification key can be authentically distributed. Our digital signatures are not deterministic due to the potential $\bot$ evaluations of the $\botPRF$. Nevertheless, we show that their approach can be adapted to our signature scheme, hence showing: 
\begin{theorem}[Informal version of \cref{thm:qpke-from-pdprf-restated}]
    \label{thm:informal_pd_prg_implies_QPKE}
    $\PDPRG$ implies CPA-secure QPKE with reusable (tamper-resilient) quantum public keys.
\end{theorem}

These achievements signify that such $\SPRS$s can replace $\PRG$s in the two fundamental tasks of authentication and encryption. 
Security is guaranteed against adversaries making classical queries to the signing or encryption oracle as  
 it was done previously in the works of \cite{AGQ22,AQY22,MY22}.

\subsection{Related Works}

Ever since the introduction of $\PRS$s \cite{JLS18} as a potentially weaker assumption than $\PRG$s, a series of works have attempted to replace $\PRG$s with $\PRS$s. Non-interactive statistically binding commitments~\cite{MY22}, non-interactive statistically hiding~\cite{Yan22,HMY23} commitments as well as quantum MPC were shown (see~\cite{MY22} and references therein).

However, certain applications remained elusive which prompted the introduction of the (potentially) stronger assumption of \emph{pseudorandom function-like states} ($\PRFS$s)\footnote{An $m$-$\PRFS$ is a generator $G$ which takes a seed $k\in \{0,1\}^\lambda$ and an input $x\in \{0,1\}^\ell$ and outputs an $m$-qubit state such that oracle access to $G$ is indistinguishable from oracle access to a truly random Haar state generator. It was shown that $\PRFS$ with $\ell=O(\log(\lambda))$ can be constructed from $\PRS$, but $\PRFS$ with $\ell=\omega(\log(\lambda))$ seems like a stronger assumption and is more useful for applications.} \cite{AQY22}. This primitive was used to construct CPA-secure symmetric encryption, message-authentication codes, and quantum public-key encryption (QPKE) with non-reusable keys (each encryption necessitates a new public-key copy) \cite{AQY22,ALY23,BGH23}. 

Morimae and Yamakawa \cite{MY22} constructed length-restricted one-time signatures with quantum verification keys from $\PRS$s. 
Notably, they left an intriguing question unanswered, which is whether digital signatures---the only major missing protocol in Minicrypt---could be constructed based on $\PRS$s (or variants thereof). We answer this question affirmatively in \cref{thm:informal_pd_prg_implies_signatures}. As mentioned before, a recent work~\cite{CM24} has shown that digital signatures are black-box separated from linear-length $\PRS$, which in conjunction with \Cref{thm:informal_pd_prg_implies_signatures} implies a black-box separation between $\LPRS$ and $\SPRS$. It should be noted that an independent concurrent work~\cite{BM24} also shows black-box separation between $\LPRS$ and $\SPRS$.
%Notably, they left an intriguing question unanswered---the feasibility of a many-time signature scheme rooted in $\PRS$s, which we answer affirmatively in \cref{thm:informal_pd_prg_implies_signatures}.

An even (potentially) stronger pseudorandom notion is $\PRFS$s \emph{with proof of destruction} (\textsf{PRFSPD})~\cite{BBS23}, which requires that the $\PRFS$ can be destructed in a way that generates a classical certificate proving destruction. The classical nature of these certificates enables various applications, such as one-time signatures with classical keys and QPKE with reusable quantum public keys \cite{BGH23}. Note that both QPKE schemes discussed are not resistant to tampering attacks. This is a drawback because a quantum state cannot be signed~\cite{BCG02}, implying that their construction requires secure quantum channels for the distribution of the quantum public keys, which is a very strong setup assumption. This motivated the work~\cite{KMNY23}, which provides a QPKE scheme with quantum public-keys that is resistant to tampering attacks based on \textsf{OWF}s. Our scheme is the first to achieve this based on an assumption that is potentially weaker than \textsf{OWF}s -- thereby answering a question posed in the same work~\cite{KMNY23} regarding the possibility of this result.

Most important for our work is the recent paper by Ananth, Lin, and Yuen~\cite{ALY23}, where they demonstrated how to construct $\PDPRG$ from $\SPRS$s and $\PDPRF$ from logarithmic $\PRFS$s. It should be noted that the definition of $\PDPRF$ considered in~\cite{ALY23} only satisfies selective security with unique queries. Consequently, they provide as applications, commitments with classical communication and CPA-secure symmetric encryption with classical ciphertexts. 
%Our research, in essence, can be seen as an extension of their work. 
We leverage their construction for $\PDPRG$ to construct $\botPRF$. Notably, our $\botPRF$ are based on $\SPRS$s, whereas their $\PDPRF$ are based on logarithmic $\PRFS$s which is a potentially stronger assumption than $\SPRS$s. Moreover, our $\botPRF$ provides an easier and more structured definition to work with, thereby enabling the development of more powerful cryptographic applications.

\subsection{Technical overview}
\label{ssec:tech_overview}

%\onote{Focal point: A new template for defining pseudodeterministic generators, which allows for adaptive security, and deals with pseudo-deterministic errors, both for the bot-PRG and the bot-PRF.

%We demonstrate that it can be useful for applications, where it is quite easy to adapt the proofs for that case as well. 
%}

We will first recall the definition of a $\PDPRG$. Informally, this is an efficient quantum algorithm $G$, that takes a classical seed of length $\secpar$, and outputs a string $y$ of length $\ell>\secpar$. 
Unlike a standard $\PRG$, the $\PDPRG$ output need not be deterministic. 
The correctness guarantee is as follows: There exists a large set of good keys $\mathcal K$ consisting of $1-1/\poly$ fraction of all keys, such that for every good key $k \in K$, there exists a string $y \in \{0,1\}^\ell$ s.t. $\Pr(G(k)=y)\geq 1-1/\poly$. 
The pseudorandomness property is defined as standard $\PRG$s: A QPT adversary cannot distinguish between $G(k)$, where $k\xleftarrow{\$} \{0,1\}^\secpar$ and a random string $r \xleftarrow{\$} \{0,1\}^\ell$. 
The formal details are given in~\cref{def:PD-PRG}. Ref.~\cite{ALY23} proved that $\SPRS$ implies $\PDPRG$, in some parameter regimes (see \cref{thm:short_prs_implies_pdprg}).

Following the classical template to construct digital signatures, the first step would be to construct universal one-way hash functions (\textsf{UOWHF}) from a $\PRG$ in the psuedodeterministic setting in order to lift a one-time signature scheme to a many-time stateful scheme. The next step would be to construct an adaptively secure $\PRF$ from a $\PRG$ to enable a full-fledged stateless signature scheme, as we discuss in further detail later. 

To achieve the goals mentioned above, we need to develop appropriate notions for $\PRF$s and $\UOWHF$s that can be realized from $\PDPRG$.

\paragraph{Need for $\botPRF$.} 
A natural attempt would be to build pseudodeterministic $\PRF$s from $\PDPRG$, following the constructions in the deterministic setting. However, there are obstacles with this approach:
\begin{enumerate}
    \item \emph{Need for stronger notions of security.} \label{item:stronger notions of security} 
    In the classical template to construct digital signatures, the signing algorithm chooses a random leaf in the authentication tree, and the corresponding one-time signing key of the leaf finally signs the message, where the one-time signing and public keys for each of the internal node are obtained by evaluating the $\PRF$ on the node. Therefore, signing multiple messages would involve evaluating the $\PRF$ evaluation of the same string (especially the string representing the root and some of its nearest successors) multiple times. Hence, it is crucial that the $\PRF$ satisfies multi-time security, i.e., the adversary in the distinguishing game is allowed to query on the same input twice. Of course, multiple evaluations do not affect security for classical $\PRF$s since they are deterministic, but it is harder to achieve multi-time security in the pseudodeterministic setting. Indeed, the only definition of $\PDPRF$ studied in~\cite{ALY23} is not multi-time secure since the security game only allows unique queries from the adversary. This is a necessary restriction because querying a $\PDPRF$ on the same input twice may result in two different outcomes, which will trivially make the $\PDPRF$ distinguishable from a random function.
    %\footnote{This problem can be solved at the cost making the signing non-deterministic by instead sampling a random path for every signature independent of the message. However, even then, in the security game, the final verification that the challenger performs on the alleged signature, submitted by the forger, would require that the $\PRF$ be queried along the path described in the alleged signature, which would amount to an adaptive query to the adversary.}. 
    
    %Hence, in order to argue CMA security in the presence of the signing oracle, it is crucial that the underlying $\PRF$ is multi-time secure. %There is no known definition of adaptively secure $\PDPRF$. Indeed, the pseudodeterminism poses a problem with properly defining such a notion. This is because adaptive security allows the adversary to query the same message twice in the Oracle distinguishability game. Clearly, due to pseudorandomness, querying the  $\PDPRF$ on the same message twice will result in different outcomes with non-negligible probability, whereas a truly random function is deterministic and will output the same value twice, thereby leading to a trivial distinguishing attack.

    Secondly, it is also desirable to make the signing algorithm deterministic. In the classical template, the signing algorithm is derandomized by replacing the randomness required to choose the leaf in the authentication tree for signing the message, with the $\PRF$ evaluation of the $\PRF$ on the message. With this modification, it becomes crucial in order to argue CMA security that the $\PRF$ is secure even when the adversary can make adaptive queries.
    Alongside the construction of digital signatures with deterministic signing, adaptively secure $\PRF$ is widely used in simple constructions of other fundamental primitives. For instance, one of the simplest constructions of Message Authentication Codes (MAC) is where signing a message $m$ is simply evaluating the $\PRF$ on $m$ and verification of an alleged message tag pair $(m,\sigma)$ is to evaluate the $\PRF$ on $m$ and checking if it the output is same as $\sigma$. Clearly, the adaptive security of $\PRF$ is crucial to argue the CMA-security of the MAC thus constructed.  However, the (only) $\PDPRF$ definition studied in~\cite{ALY23} considers only selective security. Furthermore, the proof techniques used in~\cite{ALY23} to construct $\PDPRF$ seem inadequate to prove adaptive security.
    \item \emph{Need for structured pseudodeterminism.} \label{item:structured-pseudodeterminism} The psuedodeterminism in $\PDPRF$ prevents us from utilizing them in the construction of digital signatures. Recall that in the classical template, the public key of the digital signature is associated with the public key of the one-time digital signature scheme at the root of the authentication tree. This public key is used to sign the public keys of its two children, which are the $\PRF_k(0)$ and $\PRF_k(1)$.
    %Essentially, this is because the $\PRF$ is used to determine which message to sign under a one-time signature scheme in the authentication tree. 
    However, the pseudodeterministic errors of the $\PDPRF$ may cause the one-time signature scheme to sign two different messages using the same secret key, which may compromise the security of the one-time scheme and, thus, of the entire authentication tree.
    For instance, if the adversary queries the signing oracle twice, and during the signing of the two messages, the two different evaluations of the root yield different outputs, then the adversary would learn two signatures under the same one-time signature secret key at the root.
\end{enumerate}

The main conceptual contribution of our work is solving these issues by introducing the notion of recognizable abort. Recognizable abort for $\PRF$s means that for every input $x$, the evaluation either outputs a fixed string $y$ or aborts by outputting the special symbol $\bot$. This clearly provides a more structured form of pseudodeterminism that solves the issues mentioned in \Cref{item:structured-pseudodeterminism} above. Surprisingly, the structured pseudodeterminism thus obtained also leads us to natural definitions for stronger security guarantees, namely, multi-time security and adaptive security, as mentioned above in \Cref{item:stronger notions of security}. 
%Thus recognizable abort solves all the issues mentioned above. 
More generally, we believe that this notion provides a simpler route toward achieving applications from pseudodeterministic primitives.

Formally, a $\botPRF$ family is a family of quantum algorithms $\{f_k\}_{k\in \{0,1\}^\secpar}$, with $m$ bits of inputs, which always yield a  \emph{classical} output from the set $\{0,1\}^\ell \cup \{\bot\}$ (i.e., the algorithm can abort).  For every key $k$ and input $x$, there exists $y$ such the evaluation $f_k$ is in $\{y,\bot\}$ except with negligible probability. Moreover for every $x\in \{0,1\}^m$, $\Pr(f_k(x) \neq \bot : k\xleftarrow{\$}\{0,1\}^\secpar) \geq 1-\frac{1}{\poly}.$
The recognizable abort property lets us define multi-time adaptive security (i.e., the adversary can adaptively query as well query at the same input more than once) for $\botPRF$ in the following way, which suffices for constructing digital signatures in the classical template.
The challenger samples a key $k\xleftarrow{\$}\{0,1\}^\secpar$, a random function $f:\{0,1\}^m\to \{0,1\}^\ell$, and a bit $b\xleftarrow{\$} \{0,1\}$. 
If $b$ equals $0$, it provides (classical) oracle access to $f_k(\cdot)$. If $b=1$, the oracle that the challenger provides behaves as follows. Given an input $x$, the oracle outputs $\bot$ if $f_k(x)$ evaluates to $\bot$, and $f(x)$ otherwise. Note that the adversary may query the same element $x$ many times. 

%\mo{I moved this paragraph from above... its not that imp that it should be the opening of tech overview.}

Now that we have an multi-time adaptively secure definition of $\PRF$ in the pseudodeterministic setting, the next challenge is how to construct $\botPRF$ from $\SPRS$.
%\anote{Maybe we should postpone the discussion about the UOWHFS from the previous paragraph in a later paragraph as it is repeated again in the digital signatures paragraph}
%\anote{Mohammed, I think this could be a nice place to say the need for bot \textsf{UOWHF} as another ingredient required in the classical template to build digital signatures. }

\paragraph{Construction of $\botPRF$, and the need for $\botPRG$} A na\"ive attempt to construct $\botPRF$ would be to instantiate the famous Goldreich, Goldwasser, and Micali (GGM) transformation~\cite{GGM86} from $\PRG$-to-$\PRF$ with a $\PDPRG$. Note that in the proof of pseudorandomness for the GGM construction, the reduction uses multiple evaluations of the $\PRG$ at the same random seed. Of course, multiple evaluations do not affect security in the classical setting since the $\PRG$ there is deterministic, but in the pseudodeterministic setting, multiple evaluations may output distinct values, and hence security does not hold for $\PDPRG$ if multiple evaluations are allowed. Indeed, even if two evaluations of a $\PDPRG$ on the same seed are allowed, then with non-negligible probability, we will get different outputs for different evaluations, which is trivially distinguishable from giving a uniformly random string twice.
\begin{comment}
One of the challenges with using a $\PDPRG$ is that the pseudorandom definition allows for only a single query.
In other words, the construction is insecure if the adversary is allowed to have even two copies of the output of $G(k)$ with the same seed $k$: even for a good seed $k\in \mathcal K$, by getting two outputs of $G(k)$, there is a $1/\poly$ probability of seeing two different outcomes. %On the other, it is not so clear how to extend the definition for polynomial queries. \lnote{what do we think about the last sentence?}
\end{comment}

As a remedy, we attempt to define an intermediate primitive called $\PRG$ with recognizable abort or $\botPRG$.
Similarly, to $\PDPRG$, a $\botPRG$ is a quantum algorithm $G$ mapping $\{0,1\}^\secpar$ to $\{0,1\}^{\ell}$, except the algorithm can also output a special string $\bot$, which represents aborting. We require $\ell>\secpar$, for non-triviality, as all forms of $\PRG$. As for the output, it is guaranteed that for every $k\in \{0,1\}^\secpar$, there exists a string $y \in \{0,1\}^\ell$ s.t. $\Pr(G(k)\in \{y,\bot\})\geq 1-\negl$, for some negligible function $\negl$. Moreover, just-like $\PDPRG$, there exists a good set $\mathcal{K}$ consisting of $1-1/\poly$ fraction of all keys, such that for every good key $k\in \mathcal{K}$,
$\Pr(G(k)=\bot)$ is negligible. A $\botPRG$ can be easily distinguished from a random string since a random string is never $\bot$. Therefore, standard pseudorandomness cannot hold for $\botPRG$.

A natural alternative would be to demand conditional pseudorandomness, i.e., indistinguishability from a random string conditioned on the evaluation of the $\botPRG$ not being $\bot$. We will refer to it as $\bot$-pseudorandomness. An advantage of this alternate definition is that it can be extended to incorporate $\botPR$, even when the $\botPRG$ is evaluated multiple times on the same random $k$. This notion of pseudorandomness that we call multi-time $\botPR$ can be explained via the following indistinguishability game where we demand negligible distinguishing advantage. For some $q(\secpar)\in \poly$ denoting the number of repeated evaluations, the challenger samples a key $k\xleftarrow{\$}\{0,1\}^\secpar$, and then either provides $q$ evaluations of $G(k)$ to the adversary or sends $y_1,\ldots,y_q$ to the adversary where $y_1,\ldots,y_q$ are generated as follows: The challenger samples a string $y\xleftarrow{\$}\{0,1\}^\ell$, evaluates $G(k)$ $q$-times and if the $i^{th}$ evaluation is successful (i.e., the evaluation did not abort) sets $y_i$ to $y$ else sets $y_i$ to $\bot$. It can be verified that for $q=1$, the above indistinguishability game coincides with the conditional pseudorandomness.
Note that a $\botPRG$ satisfying single-time $\botPR$ (i.e., $q=1$) does not generically satisfy multi-time $\botPR$. 
In particular, the multi-time $\botPR$ game, unlike the single-time $\botPR$ game, allows the distinguisher to learn the probability of outputting $\bot$ for the sampled key, which can leak information regarding the non-$\bot$ output of the key. %using which he may distinguish between the random string and the $\botPRG$ case in the multi-time $\botPR$ game with non-negligible probability.
For example, consider a boolean output $\botPRG$ $G$ and, for simplicity, assume that there are only four keys $k_1,\ldots,k_4$, out of which $k_4$ is the only non-good key which evaluates to $\bot$ with probability $\frac{1}{2}$, and outputs $0$ with the remaining probability. For $i\in [3]$, $k_i$ outputs $0$ with probability $\frac{5}{12}$ and $1$ with the remaining probabibility. It can be easily checked that $G$ satisfies the conditional pseudorandomness or single-time $\botPR$. However, $G$ does not satisfy multi-time $\botPR$ because, using polynomial evaluations, the distinguisher in the multi-time $\botPR$ game can learn the probability of outputting $\bot$ for the key $k$ that was sampled by the challenger and, in particular, distinguish between the case when $k$ is a good key, i.e.,  $k\in \{k_1,k_2,k_3\}$ or $k$ is non-good, i.e., $k=k_4$. Since there is an asymmetry between the output distribution of good keys and the non-good key, the distinguisher can guess the non-$\bot$ output of the key with a constant advantage, using which he can distinguish if the non-$\bot$ evaluations provided to him were from the $\botPRG$ or if they were provided using a uniformly random string. Concretely, in this example, if the distinguisher learns that the key is from the good set and the non-$\bot$ evaluation is $1$ or that the key is non-good and the non-$\bot$ evaluation is $1$, then he outputs $0$ or else outputs $1$. Clearly, the distinguisher outputs $1$ with a significantly larger probability in the random string case than in the $\botPRG$ case.
%the adversar $0$ and $1$ and then answer $0$ () (which happens with non-negligible probability), and similarly, $k_2$ outputs $\bot$ with probability $\frac{3}{4}$ and $1$ with the remaining probability and $\frac{3}{4}$ respectively, and ou.   

Intuitively, the security of $\botPRF$ is an extension of the multi-time $\botPR$ property of $\botPRG$ to $\botPRF$ similar to the relation between the pseudorandomness of $\PRG$ and $\PRF$ in the classical setting. 
This enables the construction of $\botPRF$ from $\botPRG$ via the GGM transform, via the same classical arguments in the security proof up to adaptations required due to the pseudodeterminism. Interestingly, the main technical difficulty in proving the GGM transformation for $\botPRG$-to-$\botPRF$ is the proof of the correctness for the resulting $\botPRF$ (i.e., it satisfies the relevant $\bot$ properties), for which we had to use the security properties, i.e., $\botPR$ along with the correctness of the underlying $\botPRG$, see the proof of \Cref{prop:correctness-bot-PRF-family}. This is in contrast to the classical GGM construction, where correctness holds trivially. 
%\anote{Self: add the complication regarding the correctness guarantees.} \anote{I need to think what are the differences with the proof in the classical setting that are meaningful enough to explain here. If I find the correct argument, I will add a new paragraph named construction of $\botPRF$ from $\botPRG$.}
%Therfore, it may be possible to construct $\botPRG$ which enables the construction and  can be seen as a natural extension of the $\botPR$ property of $\botPRF$ to $\botPRG$. ensures pseudorandomness even when the $\botPRG$ is called on the same random seed multiple times. An interesting example where this is crucially used is the Goldreich, Goldwasser, and Micali (GGM) construction~\cite{GGM86} of a $\PRF$  from a $\PRG$, as we will see later.

%\paragraph{$\botOWF$ from $\botPRG$.}

%We will also need to use $\OWF$s later in the construction of $\UOWHF$s. It is difficult to build these based on 
%As a direct application, we build $\botOWF$s from $\botPRG$ using standard techniques, thereby demonstrating that $\botOWF$ can be based on $\PDPRG$s. 

\paragraph{$\botPRG$ from $\PDPRG$.} We show a construction of $\botPRG$ from any $\PDPRG$, in two steps. %Suppose the pseudodeterministic error of the $\PDPRG$ is $\nu$. 
In the first step, we run the $\PDPRG$ with the same seed $\secpar$ times, output the outcome that appeared in at least $60\%$ of the evaluations, and output $\bot$ if no such outcome occurs. Note that for a good seed $k\in \mathcal K$, the generator will not output $\bot$, except with negligible probability. A standard argument is used to prove that this construction satisfies the standard notion of weak-PRG pseudorandomness (the output can easily be distinguished from random whenever the output is $\bot$, which may occur
when a bad seed is picked with 
probability $1/\poly$).

The goal of the second step is to kill two birds with one stone: to reduce the (weak) pseudorandomness guarantees from $1/\poly$ advantage to a negligible advantage and achieve that in the multi-time $\botPR$ setting, thereby achieving both amplification and multi-time $\botPR$. To achieve that, we use Levin's version of Yao's XOR lemma adapted to the quantum setting. Our construction works as follows. The new $\botPRG$ uses $\secpar$ times longer seeds compared to the previous one, which we use 
as $k_1,\ldots,k_\secpar$. We first evaluate $G(k_1),\ldots, G(k_\secpar)$ and output $\bot$ if any of them aborts, and continue to output the XOR of all of the values otherwise. Note that any seed of the new $\botPRG$ that consists of $\secpar$ original good seeds is a good seed.

\paragraph{$\botUOWHF$ and its construction from $\botPRG$}
Toward building digital signatures, we will need to construct some form of $\UOWHF$s. The works \cite{R90,HHR+10} have showed how to build a $\UOWHF$ from any $\textsf{OWF}$. 
Naturally, we would like to adapt the deterministic constructions~\cite{R90,HHR+10} for $\UOWHF$ from $\OWF$, so that it can be instantiated with a pseudodeterministic form of $\OWF$s ($\PDOWF$)\footnote{Here, we expect that just like in the classical setting, every $\PDPRG$ with sufficient stretch is a $\PDOWF$.}. 
However, it is not at all clear how to do this since the standard constructions rely on structural properties of $\OWF$s that do not hold for pseudodeterministic functions. For instance, the constructions require intricately working with the entropy of a uniformly random preimage of an image element under the $\OWF$. This is problematic for pseudodeterministic functions since the number of preimages of an image element is ambiguous, as many inputs may map to the image with some probability. As a result, many bounds do not carry over. In fact, we do not know how to solve these problems for $\PDOWF$.

To address these obstacles, we introduce the notion of $\botOWF$s. Similar to $\botPRG$s, for a $\botOWF$ $F$ (see \Cref{def:botOWF}), it is guaranteed that for every input $x$, there exists an output $y$ s.t. $\Pr(F(x)\in \{y,\bot\})\geq 1-\negl$, for some negligible function $\negl$. Moreover, there exists a good set $\mathcal{X}$ consisting of $1-1/\poly$ fraction of all inputs, such that for any good input $x\in \mathcal{X}$, $\Pr(F(x)=\bot)$ is negligible. For security, we require that it is difficult for an efficient adversary to find an inverse to a non-$\bot$ image. This is sufficient for our applications as we simply abort when we obtain a $\bot$. Similar to the classical $\PRG$ and $\OWF$ case, we show that every $\botPRG$ (with sufficient stretch) is a $\botOWF$. Our construction uses standard techniques along with some adaptations to deal with the $\bot$ evaluations. 

While $\botOWF$s still exhibit some level of non-determinism, it is far more manageable than $\PDOWF$. Intuitively, when we consider the inverses to an element $y$, we are interested in inputs that evaluate to $y$ with non-negligible probability, as these are the inputs that can potentially pass as valid inverses for $y$. For each input of a $\botOWF$ $F$, there exists a value $y$ such that $F(x)\in \{\bot, y\}$ except with negligible probability. Hence, for each input, there exists a unique non-$\bot$ value in the image which it can evaluate to with non-negligible probability. This allows us to define several notions of entropy for the inverse sets that act in the same way as standard notions of entropy used in \cite{HHR+10} for deterministic functions. 

Consequently, we build $\botUOWHF$ (formally defined in \Cref{def:botUOWHF}) from $\botOWF$ (formally defined in \Cref{def:botOWF}) following a similar avenue to the deterministic construction given in \cite{HHR+10}. However, we note that the proof in the non-deterministic setting requires non-trivial adaptations. For instance, in \cite{HHR+10}, it is noted that if $\beta$ portion of the inputs evaluate to an output that begins with 0 under a $\textsf{OWF}$, then there is $1-\beta$ portion of inputs that evaluate to an output that begins with 1. This trivial argument does not hold in our case, since many inputs may map to $\bot$ with high probability. As a result, some of the consecutive steps in their proof break down, necessitating further analysis. See \ifnum\iacr=1Supplement~\ref{sec:botUOWHF construction}\else \Cref{sec:botUOWHF construction}\fi for our proof. 

\paragraph{How to use pseudo-determinism with recognizable abort.} \label{pg:template-for-applications-intro}
%All the applications of $\botPRF$ in this work are primitives that can be constructed based on vanilla $\PRF$s. 
Suppose primitive A (e.g., MAC) can be constructed from a deterministic primitive $B$ (e.g., $\PRF$).
How can we achieve a similar result using a pseoudeterministic version of $B$ (for example, MAC based on $\botPRF$)? 
We use the following template in our work: 
%We use the following template to achieve applications from pseudo-deterministic primitives with recognizable abort. W
%We start with a secure construction of the primitive from $\PRF$. 
In the first step, we replace the underlying determinstic primitive (e.g., the $\PRF$) in that construction with its recognizable abort counterpart (e.g., a $\botPRF$), which, in most cases can easily be shown to be secure\footnote{
Note that this step may not work if we have pseudodeterminism without recognizable abort because the security proof may require multi-time adaptive security from $B$ (e.g., $\PRF$) that may not hold for the pseudodeterministic version of $B$ (e.g. $\PDPRF$). 
}%
, but only with $(1-\frac{1}{\poly})$-correctness, due to $\bot$ errors. 
%Intuitively, the security of the resulting scheme follows from the security of the original scheme.
%This is because the security experiment in the resulting scheme is identical to that of the original scheme except for the fact that some of the algorithms in the resulting construction may output $\bot$ during the experiment. 
%However, the phenomenon that some of the algorithms output $\bot$ can only harm the prospects of the adversary in the security game, and therefore the resulting scheme is at least as secure as the original scheme.
%Note that this may not work if we use the existing definitions of $\PDPRF$ instead of $\botPRF$ because the security of the resulting construction may not hold because, unlike the vanilla $\PRF$  in the original construction, the $\PDPRF$ of the new construction does not satisfy adaptive security which may have been crucial to argue security for the construction of the primitive.
In the next step of the template, we consider a variant of a $\secpar$-repetition of the $(1-\frac{1}{\poly})$-correct but secure construction at the end of the previous step to amplify the correctness to $(1-\negl)$. Using standard techniques, we then show that the resulting $\secpar$-fold construction is secure if the underlying base construction is secure. For example, to show a MAC based on $\botPRF$, we sample $\secpar$ signing keys, and signing a message $m$ would be done by evaluating the $\botPRF$ on this message using all the keys. The verification procedure applies the $\botPRF$ on the message with all the signing keys, and accepts if at least one of the successful un-aborted evaluations agrees with the signature.

\paragraph{Digital Signatures.}

In the classical setting, basing (stateless) digital signatures (DS) on the minimal assumption of \textsf{OWF}s took a sequence of works \cite{L79,NY89,R90,Gol04}, culminating in a complex final scheme. We first briefly describe this effort and how we adapt it to our work, replacing \textsf{OWF}s with $\SPRS$s.

Lamport \cite{L79} first showed how to construct a simple one-time signature scheme from any \textsf{OWF}. For single-bit messages, the secret key is two randomly sampled inputs $\sk\coloneqq (x_0,x_1)$, while the verification key is the images $\vk\coloneqq (F(x_0),F(x_1))$ under a \textsf{OWF} $F$. The signer can simply reveal $x_b$ to sign message $b$; unforgeability follows easily from the difficulty of inverting $F$. 

Furthermore, as shown in \cite{Gol04}, by using a \textsf{UOWHF} instead of an \textsf{OWF}, we can ensure that the adversary cannot generate a new signature for a message $m$ even given a signature of $m$. This property is called \emph{strong unforgeability} and is used in the application of QPKE, see \ifnum\iacr=1Supplement~\ref{subsec:construct-qpke-proof}\else \Cref{subsec:construct-qpke-proof}\fi. 

Next, Naor and Young \cite{NY89} showed how to sign many-time signature scheme using \textsf{UOWHF}s. However, in their scheme, the signer needs to remember all previous signatures produced. As a result, this scheme is called ``\emph{stateful.}'' They also show how to construct a \textsf{UOWHF} from any injective \textsf{OWF}. Rompel in \cite{R90} then showed that \textsf{UOWHF} can be constructed from any \textsf{OWF} and, therefore, from any $\PRG$. Rompel's proof is quite involved, prompting a follow-up work \cite{HHR+10} that provides a simpler construction. 

%Instead, a weaker notion known as \emph{universal one-way hash functions} (UOWHFs) is sufficient. 
%Fortunately, Rompel \cite{R90} later showed that this primitive can be constructed directly from \textsf{OWF}s which gives a construction from $\PRG$s. 

%In order to adapt Lamport's scheme to SPRSs, we construct pseudodeterministic \textsf{OWF}s (PD-OWFs) from $\PDPRG$. A standard $\PRG$ already acts as a \textsf{OWF} as any inversion algorithm can be used to distinguish the output from random. Informally, this is because the image of a $\PRG$ is much smaller than the codomain, meaning that an algorithm that can tell whether a string is in the image can distinguish if the string is sampled randomly or from the $\PRG$. It is easy to show that the Generalization Lemma \ref{lem:gener PRG} implies that PD-OWFs can be constructed from $\PDPRG$ in the same way, giving a one-time signatures following Lamport's approach.  

Later, Goldreich \cite{Gol04} leveraged $\PRF$s to achieve the final result: a many-time ``\emph{stateless}'' DS scheme, where the signer does not need to store previous signatures.

%In an attempt to adapt this approach to digital signatures but relying on pseudodeterministic functions, we face some critical issues. First of all, the constructions of \textsf{UOWHF}s from \textsf{OWF}s \cite{R90,HHR+10} both rely heavily on structural properties of the underlying \textsf{OWF} which do not hold for pseudodeterministic functions. %For instance, the constructions require intricately working with the entropy of the preimage sizes of the \textsf{OWF}. This is problematic for pseudodeterministic functions since the preimage size of an image element is ambiguous as many inputs may map to the image with some probability. As a result, many bounds do not carry over. In fact, we do not know how to solve these problems for $\PDPRG$s. 

With $\botPRF$s and $\botUOWHF$, we carefully adapt the techniques used in the aforementioned works to our setting. The overall construction of the digital signature follows the same template as mentioned above on \cpageref{pg:template-for-applications-intro}. On the downside, by relying on pseudodeterministic functions, our signatures have a non-negligible probability of failing verification. This issue can be addressed by repetition, however, the resulting scheme is no longer strongly unforgeable and only satisfies standard unforgeability -- it is easy to modify one of the 'bad' signatures in the repeated scheme without being detected. Hence, our work gives the option to choose between statistical correctness and strong unforgeability. 
%Furthermore, another issue is that when adapting Goldreich's approach to our setting, $\PDPRF$ cannot be used instead of $\PRF$s as it is essential that the function is deterministic. Otherwise, the signer might accidentally sign two different messages using the same one-time signature keys. While $\botPRF$ are not fully deterministic, the $\bot$ outputs can be handled by simply aborting the signing procedure. When no $\bot$ appears, then all the computations are deterministic except with negligible probability, allowing us to substitute $\PRF$s in Goldreich's scheme with $\botPRF$. 
In conclusion, we achieve statistical correct unforgeable (many-time) stateless DSs from $\SPRS$s and $(1/4)$-correct strongly unforgeable signatures from $\SPRS$s. %The inverse-polynomial probability of getting $\bot$ implies the scheme does not satisfy correctness, but this issue can easily be addressed by repetition.  % \onote{To myself: Mention that we follow the template above in this application.}\anote{Is this okay?}\mo{no I dont believe it follows the same template. Because bot owf is more problematic to use in constructing uowhf than owfs.}

\paragraph{Tamper-proof QPKE.} A generic construction of CPA secure tamper-proof QPKE from strongly unforgeable digital signatures with unique signatures was shown in~\cite{KMNY23}. We use their generic construction to show that tamper-proof QPKE can be constructed based on $\SPRS$ as follows. First, we show that instantiating the construction in~\cite{KMNY23} with a $(1-\frac{1}{\poly})$-correct digital signature scheme that can be based on $\SPRS$ results in a $(1-\frac{1}{\poly})$-correct $\QPKE$ scheme $\Pi$ that is $\cpatamp$ secure, i.e., CPA security in the tamper-proof setting as defined in~\cite{KMNY23}. Next, we do a variant of $\secpar$-parallel repetition of $\Pi$ to get a $\QPKE$ scheme $\Pi^\lambda$ that satisfies negligible correctness error. In the repeated scheme $\Pi^\secpar$, we use $\secpar$ independently generated secret keys, public keys, and verification keys of $\Pi$ as the new secret key, public key, and verification key, respectively. Encryption is performed by encrypting the message $\secpar$-times using the $\secpar$ public key independently. The decryption interprets the cipher as $\secpar$ ciphers of $\Pi$ runs the decryption of $\Pi$ on each of them with the respective secret keys and outputs $\bot$ if all of them are $\bot$. Else, checks if all the decrypted messages that are non-$\bot$ are the same or not. If they are the same, output that message, else output $\bot$. We show that $\Pi^\secpar$ is $\cpatamp$ secure if $\Pi$ is $\cpatamp$ secure, via a standard reduction similar to that used in the black-box construction of secure multi-bit encryption from secure single-bit encryption via parallel repetition. The full proof is given in Supplement~\ref{subsec:construct-qpke-proof}.

\subsection{Open problems}

An interesting avenue for future work is to upgrade the security of $\bot$-PRF, or variants thereof, in order to permit adversaries quantum oracle access to the function. In our work, we have only demonstrated security against QPT adversaries that are given \emph{classical} access to the various functionalities. This limitation has implications for our applications as well. For instance, in the context of MACs and CCA-secure symmetric encryption, an adversary is only allowed to query the authentication or encryption oracles on classical inputs rather than quantum inputs. There exist security notions allowing adversaries to query on a superposition of inputs (see \cite{BZ13,GSM20}), and hence, upgrading the security of our $\bot$-PRFs would allow these notions to become attainable. 
\anote{Do we need to add anything from the next version to-do list for added motivation? }

\ifnum\anonymous=0
    \subsubsection*{Acknowledgments}
        We wish to thank an anonymous reviewer for helpful suggestions regarding a previous version of this manuscript. 
                \BeforeBeginEnvironment{wrapfigure}{\setlength{\intextsep}{0pt}}
                \begin{wrapfigure}{r}{100px}
                    %\centering
                    \includegraphics[width=100px]{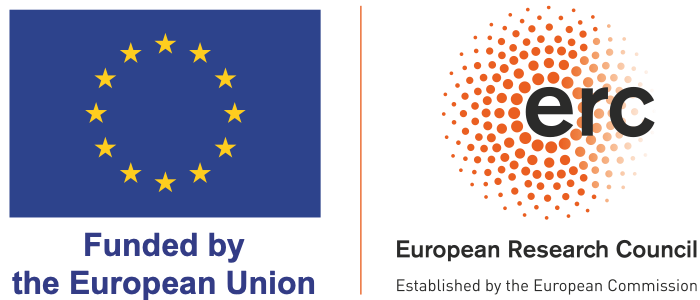}
                \end{wrapfigure}
                The work was funded by the European Union (ERC-2022-COG, ACQUA, 101087742). Views and opinions expressed are however those of the author(s) only and do not necessarily reflect those of the European Union or the European Research Council Executive Agency. Neither the European Union nor the granting authority can be held responsible for them.
\fi
\section{Preliminaries}

\subsection{Notations}
%We denote Hilbert spaces by calligraphic letters, such as $\hildd{H}$, and the set of positive semi-definite operators in $\hildd{H}$ as $\pos{\hildd{H}}$. If a Hilbert space $\hildd{X}$ holds a classical value, then we denote this variable as $\text{X}$.

We denote the density matrix of a quantum state in a register $E$ is denoted as $\rho_E$. The trace distance is given as $\delta(\rho,\sigma)\coloneqq \frac{1}{2}\| \rho-\sigma\|$. 

We say $x\leftarrow X$ if $x$ is chosen from the values in $X$ according to the distribution $X$. If $X$ is a set, then $x\leftarrow X$ denotes an element chosen uniformly at random from the set.

We use the notation and convention in~\cite{ALY23} to define quantum algorithms. We say that a quantum algorithm $A$ is \emph{QPT} if it consists of a family of generalized quantum circuits (i.e., circuits that have an additional subset of input qubits called ancillae qubits that are initialized to $0$) $\{A_\secpar\}_{\secpar}$ such that each circuit $A_\secpar$ is bounded a polynomial $p(\secpar)$. We call a QPT algorithm $A$ non-uniform if it consists of a family $\{A_\secpar,\rho_\secpar\}_\secpar$ where $\{A_\secpar\}_\secpar$ are circuits of polynomial size (not necessarily uniformly generated), such that for each $\secpar$, there is additionally a subset of input qubits of $A_\secpar$ that are designated to be initialized
with the density matrix $\rho_\secpar$ of polynomial length. All the algorithms used in the security definitions are non-uniform unless mentioned otherwise.

We let $A^P$ mean that $A$ is given access to a polynomial number of oracle queries to $P$. We let $[n]\equiv \{1,2,\ldots,n\}$ for every $n\in \NN$ and $\negl[x]$ denote any function that is asymptotically smaller than the inverse of any polynomial. As in~\cite{ALY23}, we work in the non-uniform setting, meaning all the security guarantees in this work are against non-uniform QPT adversaries, i.e., the adversaries are quantum poly time algorithms with quantum advice states.

We let $\textsf{Haar}(\mathbb{C}^d)$ denote the Haar measure over $\mathbb{C}^d$ which is the uniform measure over all $d$-dimensional unit vectors. For simplicity, we let $\textsf{Haar}_m$ denote the Haar measure over $m$-qubit space i.e. $\textsf{Haar}((\mathbb{C}^2)^{\otimes m})$. We let $\mathcal{O}_{\textsf{Haar}_m}$ denote an oracle that generates random $m$-qubit Haar states. 

\subsection{Probability Theorems}
\begin{theorem}[Chernoff bound for $n$ independent Poisson trials cited from Theorems 4.4,4.5~\cite{MU05}]
    Let $X_1,\dots,X_n$ be independent Poisson trials such that $\Pr[X_i]= p_i$. Let $X = \sum_{i=1}^n X_i$ and $\mu = E[X]$. Then, for $0<\delta<1$:
    \begin{enumerate}
        \item $\Pr[X\geq (1+\delta)\mu]\leq e^{-\mu\delta^2/2}$
        \item $\Pr[X\leq (1-\delta)\mu]\leq e^{-\mu\delta^2/2}$ 
    \end{enumerate}  
\label{thm:chernoffbound}
\end{theorem}

%\onote{Make sure all the security definitions are against non-uniform adversaries, that is, with quantum advice. We follow ALY in that regard. Note that we don't need security against quantum advice (we only need classical advice), and we use that simply to be consistent with ALY.}

\subsection{Pseudorandom States and Variants}

We recall the definitions of pseudorandom states (PRSs) \cite{JLS18}. In a follow up work \cite{ALY23}, $\SPRS$s were used to construct $\PDPRG$ which we define as well. 

\begin{definition}[Pseudorandom State Generator]
    A family of QPT algorithms $G=\{G_\secpar\}_{\secpar \in \mathbb{N}}$ is called $m(\lambda)$-\emph{pseudorandom state generator (PRS)}, if the following conditions hold:
    \begin{itemize}
        \item For all $\lambda\in \mathbb{N}$ and ${k}\in \{0,1\}^\lambda$, $G_\lambda(\textsf{k})=\rho_\textsf{k}$ for some $m(\lambda)$-qubit state $\rho_{{k}}$.
        \item For any polynomial $t(\cdot)$ and QPT distinguisher $\adv$:
        \begin{align*}
            \left|  \Pr_{{k}\leftarrow \{0,1\}^\lambda} [\adv_\lambda (G_\lambda ({k})^{\otimes t(\lambda)})=1]-\Pr_{\ket{\phi}\leftarrow \textsf{Haar}_{m(\lambda)}} [\adv_\lambda (\ket{\phi}^{\otimes t(\lambda)})=1]\right| \leq \negl[\lambda].
        \end{align*}
    \end{itemize}
    We divide $\PRS$ into three regimes, based on the state size $m(\secpar)$:
    \begin{enumerate}
        \item $m=c\cdot \log(\lambda)$ with $c\ll 1$.
        \item $m=c\cdot \log(\lambda)$ with $c\geq 1$, which we call \emph{short pseudorandom states ($\SPRS$s)}.
        \item $m=\Omega (\lambda)$, which we call \emph{long pseudorandom states ($\LPRS$s)}.
    \end{enumerate} 
\end{definition}

\begin{definition}[Pseudodeterministic Pseudorandom Generators, adapted from Definition 4.1 in~\cite{ALY23}]\label{def:aly23}
    A family of QPT algorithms $G = \{G_\secpar\}_{\secpar\in\mathbb{N}}$ which on input seed $k\in \{0,1\}^{\secpar}$ $G_\secpar(k)$ outputs a string of length $\ell(\secpar)$ is a \emph{pseudodeterministic pseudorandom generator} ($\PDPRG$) if the following conditions hold:
    \begin{itemize}
        \item \textbf{Pseudodeterminism:} There exists a constant $c>0$ and functions $\mu(\secpar),\nu(\secpar) = O(\secpar^{-c})$ such that for sufficiently large $\secpar\in \mathbb{N}$, there exists a set $\Klam \subset \{0,1\}^\secpar$ such that:
        \begin{enumerate}
            \item $\Pr_{k\xleftarrow{\$}\{0,1\}^\secpar}[k\in \Klam]\geq 1-\mu(\secpar).$ \label{def:pseudodeterminism-part1}
            \item For every $k\in \Klam$:
            \begin{align}
                \max_{y\in\{0,1\}^{\ell(\secpar)}}\Pr[Y\gets G_\secpar(k): Y=y]\geq 1-\nu(\secpar).
                \label{def:pseudodeterminism-part2}
            \end{align}
        \end{enumerate}
        \item \textbf{Stretch:} The output length of $G_\secpar$, namely $\ell(\secpar)$, is strictly greater than $\secpar$. 
    \end{itemize}
    We let $\left(\mu,\nu,\ell\right)-\PDPRG$ denote a generator with stretch $\ell = \ell(\secpar)$  satisfying pseudodeterminism with the functions $\mu = \mu(\secpar),\nu= \nu(\secpar)$.
    \label{def:PD-PRG}
\end{definition}

\begin{definition}[Weak/Strong Pseudorandomness for $\PDPRG$, adapted from Definition 4.1 in~\cite{ALY23}]    
    For every non-uniform QPT distinguisher $\adv$, there exists a function $\epsilon$ such that: 
        \begin{equation}
            \left|  \Pr_{k\gets \{0,1\}^\secpar}[y\gets G_\secpar(k):\adv(y)=1]-\Pr_{y\gets \{0,1\}^{\ell(\secpar)}}[\adv(y)=1]\right|\leq \epsilon(\secpar).
        \end{equation}
        We say that $G$ is a strong $\PDPRG$ if for every $\adv$, $\epsilon(\secpar)$ is a negligible function, and is a weak $\PDPRG$ if for every $\adv$, $\epsilon(\secpar)\leq\frac{1}{v(\secpar)}$ for some polynomial $v$. 
        We say that the $\PDPRG$ is $\epsilon$-pseudorandom for a function $\epsilon$ if for every $\adv$ there is a $\secpar_0 \in \mathbb{N}$ such that the bound holds with the \emph{same} function $\epsilon$ for all $\secpar>\secpar_0$.
\label{def:pdprg_pseudorandomness}
\end{definition}

\section{\texorpdfstring{$\bot$}{bot} Pseudorandom Generators}
\subsection{Definitions}

Next, we introduce a different notion of pseudo-deterministic $\PRG$ , with a recognizable abort property. 

\begin{definition}[Pseudorandom Generator with Recognizable Abort ($\botPRG$)] 
    A family of QPT algorithms $G=\{G_\secpar\}_{\secpar\in\mathbb{N}}$ which on input seed ${k}\in \{0,1\}^{\secpar}$,  $G_\secpar(k)$ outputs a string of length $\ell(\secpar)$ is a \emph{pseudodeterministic pseudorandom generator with recognizable abort} ($\botPRG$) if the following conditions hold:
    \begin{enumerate}
        \item \textbf{Pseudodeterminism:} There exist a constant $c>0$ and a function $\mu(\secpar)= O(\secpar^{-c})$ such that for sufficiently large $\secpar\in \mathbb{N}$ there exists a set $\Klam\subset \{0,1\}^\secpar$ such that the following holds:
        \begin{enumerate}
            \item \[\Pr_{k\xleftarrow{\$}\{0,1\}^\secpar}\left[k\in\Klam \right] \geq 1-\mu(\secpar).\] \label{def:bot-pseudodeterminism-part1}
            \item For every $k\in \Klam$ there exists a $y\in\{0,1\}^{\ell(\secpar)}$ and a negligible function such that: 
            \begin{align}
                \Pr\left[Y\gets G_\secpar(k): Y=y \right] \geq 1 - \negl.
            \end{align} \label{def:bot-pseudodeterminism-part2}
            \item For every $k\in\{0,1\}^\secpar$ there exists a $y\in\{0,1\}^{\ell(\secpar)}$ and a negligible function such that: 
            \begin{align}
                \Pr\left[Y \gets G_\secpar(k) : Y \in \{y,\bot\}\right] \geq 1 - \negl.
            \end{align}  \label{def:bot-pseudodeterminism-part3}
        \end{enumerate}  
        \label{def:bot-pseudodeterminism}
        \item \textbf{Stretch:} The output length of $G_\secpar$, namely $\ell(\secpar)$, is strictly greater than $\secpar$.
        
    \end{enumerate}
    We let $(\mu,\ell)$-$\botPRG$ denote a $\botPRG$ with stretch $\ell=\ell(\secpar)$ and $\mu=\mu(\secpar)$ pseudodeterminism.
    \label{def:bot-PRG}
\end{definition}

A $\botPRG$ differs from a \pdprg in two ways.
First, for every input $x$, there exists a $y$ such that the output is either $y$ or $\bot$, except with negligible probability.
Second, it has better pseudodeterminism for keys in the good set: For those keys, the $\bot$ probability is negligible.

 \lnote{Review}
\anote{Lior and I think the following is not required anymore.}\anote{The advantage of using a $\PRG$  with recognizable abort is that security is easier to show, as will be demonstrated later. The main challenge when using this type of $\PRG$  is proving correctness; one approach that can be often used is parallel repetition: even if we have only $1-\frac{1}{\poly}$ correctness, repeating the construction polynomially many times with independent keys can have $1-\negl$ correctness in some cases.}\onote{Reference concrete examples where this holds.}
\onote{Explain the term recognizable abort, which was borrowed from~\cite{AQY22}, from a slightly different context.}

Next, we address the pseudorandomness property of the $\botPRG$. Since the $\botPRG$ may often output the $\bot$ value, there is no hope of achieving the strong security as defined in \cref{def:pdprg_pseudorandomness}.
We modify the pseudorandomness definition in two important ways. First, we make sure that the adversary cannot obtain any advantage from seeing abort, by providing the abort value both in the pseudorandom, and in the truly random case---see below for details. The main advantage of this approach is the following.
Recall that in the pseudorandomness property of a \pdprg, as given in~\cref{def:pdprg_pseudorandomness}, the adversary receives only a single copy from the output. For this reason, the constructions have to be careful so that the adversary would never receive two outputs sampled using the same key. We use a different definition, that takes into account the recognizable abort, which allows the adversary to receive many samples from the $\botPRG$ output. Our definition uses the following auxiliary operator:

\begin{definition}[$\isbot$]
We define the operator 
\begin{align*}
    \isbot(a,b):=\begin{cases}
    \bot        & \text{if } a = \bot \\
        b        & \text{otherwise}.
    \end{cases}
\end{align*}
\end{definition}
\begin{definition}[Multi-Time $\bot$-Pseudorandomness for $\botPRG$]\label{def:multi-time-pseudorandomness}
    Let $G=\{G_\secpar\}_{\secpar\in\mathbb{N}}$ be a $\botPRG$. $G$ has multi-time $\bot$-pseudorandomness if for any non-uniform QPT distinguisher $\adv$ and for any polynomial $q=q(\secpar)$ there exists a negligible function $\epsilon$ such that:
    \begin{align*}
        \left| \Pr \left[ \begin{matrix}
            k\gets \{0,1\}^\secpar\\
            y_1\gets G_\secpar(k)\\
            \vdots \\
            y_q \gets G_\secpar(k)      
        \end{matrix} : \adv(y_1,...,y_q) = 1 \right] - \Pr \left[ \begin{matrix}
            k\gets \{0,1\}^\secpar \\
            y\gets \{0,1\}^{\ell(\secpar)} \\
            y_1\gets \isbot(G_\secpar(k),y)\\
            \vdots \\
            y_q \gets \isbot(G_\secpar(k),y)      
        \end{matrix}: \adv(y_1,\ldots,y_q) = 1    
        \right] \right| \leq \epsilon(\secpar)
    \end{align*}
    Similarly, $G$ has $t$-fold multi-time $\botPR$ for a polynomial $t(\secpar)$ if for every adversary $\adv$, the distinguishing advantage in the $t$-fold parallel repetition of the above indistinguishability game is negligible.
\end{definition}

\begin{theorem}\label{thm:parallel-repetition-multi-sample-bot-pseudorandomness}
    A $\botPRG$ that satisfies multi-time $\bot$ pseudorandomness, also satisfies the $t$-fold parallel version of multi-sample $\bot$ pseudorandomness for $t\equiv t(\secpar)\in \poly$.
\end{theorem}

\begin{proof}
    The proof follows by standard hybrid arguments for parallel repetition of indistinguishability games. 
    \qed
\end{proof}

\subsection{Multi-Time \texorpdfstring{$\botPRG$}{bot-PRG} From Weak \texorpdfstring{$\PDPRG$}{PD-PRG}}
The main result of this section is as follows.

\begin{theorem}[Multi-Time $\botPRG$ from weak $\PDPRG$]
    Assuming weak $(\mu,\nu,\ell)$-$\PDPRG$ exists for which $\ell(\secpar)>\secpar^2$ and there exist a constant $c>1$ such that $\mu(\secpar)= O(\secpar^{-c})$, then there exists multi-time
    $(\secpar\mu,\ell') - \botPRG$, where $\ell'(\secpar^2)=\ell(\secpar)$.
    \label{thm:pdprg_implies_botPRG}
\end{theorem}
The proof is given in \ifnum\iacr=1Supplement\else Section\fi~\ref{sec:proof_pdprg_implies_botPRG}.

Ananth, Lin and Yuen proved that a $\SPRS$ implies a weak $\PDPRG$:
\begin{theorem}[Theorem 4.2 in~\cite{ALY23}]
    Assuming the existence of $(c\log\secpar)$-PRS for some constant $c>6$, then there exists a $O(\secpar^{-c/6})$-pseudorandom $(\mu = O(\secpar^{-c/12}),\nu = O(\secpar^{-c/12}),\ell = \secpar^{c/6})$-$\PDPRG$. 
    \label{thm:short_prs_implies_pdprg}
\end{theorem}

By combining \cref{thm:pdprg_implies_botPRG,thm:short_prs_implies_pdprg}, we conclude:
\begin{corollary}\label{cor:bot-prg-from-sprs}
     Assuming the existence of $(c\log\secpar)$-PRS for some constant $c>12$, then there exists multi-time
    $(\mu = O(\secpar^{-c/12 +1}),\ell = \secpar^{c/12}) - \botPRG$.
\end{corollary}
Note that $c$ in the corollary above must satisfy $c>12$, so that the stretch $\ell(\secpar)$ of the $\botPRG$ would satisfy $\ell(\secpar)>\secpar$. Even though this would not be used in our work, we mention that Ref.~\cite{ALY23} shows that the exact same assumption as in the above, also implies a strong $\PDPRG$:

\begin{theorem}[Corollary 4.6 in~\cite{ALY23}]
    Assuming the existence of $(c\log\secpar)$-PRS for some constant $c>12$, then there exists a strong pseudorandom $(\mu = O(\secpar^{-(c/12-1)}),\nu = O(\secpar^{-(c/12-1)}),\ell = \secpar^{c/6})$-$\PDPRG$. 
\end{theorem}
\begin{remark}
    One can also show a result in the other direction, that is, multi-time $\botPRG$ with certain parameters implies strong \pdprg (with worse parameters). The proof is essentially the same construction, except we treat $\bot$ as some fixed value, such as $0^{\ell(\secpar)}$. The proof is omitted.
\end{remark}

\section{\texorpdfstring{$\bot$}{bot} Pseudorandom Functions}
\begin{comment}
GGM construction on Bot-PD-PRG to get adaptive secure Bot-PD PRF from multi-sample secure bot-PD PRG.
\anote{
Construction and proof sketch.

\begin{enumerate}
    \item Construction: Same as the GGM construction.
    \item Proof: The idea is to reduce the $q$-query adaptive distinguisher against the $\botPRF$ to a $q$-fold parallel version of $q$-sample pseudorandomness of the $\bot$-$\PDPRG$.
    \item Mimic the GGM construction and construct $n$-hybrids where $n$ is the length of the input, and then show that the adaptive distinguisher can be transformed into a distinguisher between two neighboring hybrids.
    \item Then show that the distinguishing game between any two neighboring hybrids can be simulated using $q^2$ samples $q$-samples for each of the $q$-parallel pseudorandomness games. If the prefix for the query is the same as a previous query then use a new sample from the $q$-samples from the same pseudorandomness game of the $\BPRG$ that was used for the previous query, else use a new parallelly repeated pseudorandomness game from the $q$ parallel pseudorandomness games.
    \item For quantum adaptive queries, we first use the small-range technique of Zhandry to make the distinguishing game classical (but for a larger polynomial $q'$ number of classical queries), and then do steps 3 and 4, to reduce the game to to a $q'$-fold parallel version of $q'$-sample pseudorandomness of the $\bot$-$\PDPRG$.
\end{enumerate}
}
\end{comment}

%%%%Formal definition and proof starts here

\begin{definition}\label{def:bot-prf-correctness}
Let $\bot$ be a special symbol called $\bot$. 
    A family of quantum algorithms (i.e., the output on a fixed input is a random variable) $\mathcal{F}=\{f_k\}_{k\in \{0,1\}^\secpar}$ with domain $\{0,1\}^m$ and co-domain $\{0,1\}^\ell$ where $m=m(\secpar)$ and $\ell=\ell(\secpar)$ is called a $\bot$-function family.
    \paragraph*{Efficient Implementation:} There exists a QPT algorithm $A$ that implements $\mathcal{F}$.
    \paragraph*{Correctness:} There exists an inverse polynomial $\mu(\secpar)$,
    \begin{enumerate}
        \item \textbf{Pointwise far from bot:} For every $x\in \{0,1\}^m$,
        \[\Pr[k\xleftarrow{\$}\{0,1\}^\secpar: f_k(x) \neq \bot] \geq 1-\mu(\secpar).\] \label{it:far-from-bot}
        \item \textbf{Bot-pseudodeterminism:} There exists a negligible function $\negl$ such that for every $x\in \{0,1\}^m$ and $k\in \{0,1\}^\secpar$, there exists a $y\in \{0,1\}^\ell$ and negligible function $\negl$ such that $\Pr[f_k(x)\in \{y,\bot\}]\geq 1-\negl$.\label{it:bot-determinism}
    \end{enumerate}
The inverse polynomial $\mu(\secpar)$ is the $\bot$-error, and the negligible function $\negl$ is the pseudodeterminism error. We use $(\mu,m,\ell)$-$\botPRF$ to denote a $\botPRF$ family with $\bot$-error $\mu(\secpar)$, input length $m(\secpar)$ and output length $\ell(\secpar)$.
\end{definition}

\begin{definition}\label{def:bot-prf-security} %\onote{You didn't mention efficiency: There should be a QPT circuit implementing it.}\anote{Added in the definition of both function families. I plan to mention in the notation what an efficiently implementable map means. Let me know if that is enough.}
    A $\bot$-function family $\mathcal{F}=\{f_k\}_{k\in \{0,1\}^\secpar}$  with domain $\{0,1\}^m$ and co-domain $\{0,1\}^\ell$ is (adaptively) pseudorandom if for every QPT distinguisher $D$, there exists a negligible function $\negl$ such that
    \[\Pr[\prfdistingexpt(\secpar)=1]\leq \frac{1}{2}+\negl,\]
    where $\prfdistingexpt(\secpar)$ is as defined in \cref{fig:botprfadaptexpt}.
\end{definition}

\begin{figure}[!htb]
   \begin{center} 
   \begin{tabular}{|p{16cm}|}
    \hline 
\begin{center}
\underline{$\prfdistingexpt(\secpar)$}: 
\end{center}
\begin{enumerate}
\item Challenger $\ch$ samples a bit $b\xleftarrow{\$}\{0,1\}$ and a key $K\xleftarrow{\$}\{0,1\}^{\secpar}$.
\item $\ch$ samples a random function $F\xleftarrow{\$}\left(\{0,1\}^{\ell(\secpar)}\right)^{\{0,1\}^{m(\secpar)}}$.
\item Let $G$ be the function that on input $x$, runs $f_K(x)$ and if the outcome is $\bot$, outputs $\bot$, else outputs $F(x)$.
\item If $b=0$, $\ch$ gives $D$ oracle access to $f_K(\cdot) $, and if $b=1$, gives oracle access to $G(\cdot)$.
\item $D$ outputs a bit $b'$.
\item Output $1$ (i.e., distinguisher wins) if $b'=b$.
\end{enumerate}
\ \\ 
\hline
\end{tabular}
    \caption{$\botPRF$ pseudorandomness experiment}
    \label{fig:botprfadaptexpt}
    \end{center}
\end{figure}

\subsection{Construction of \texorpdfstring{$\botPRF$}{bot-PRF} From Multi-Time \texorpdfstring{$\botPRG$}{bot-PRG}}
\label{sec:construct-GGM}
Our construction of $\botPRF$ from $\botPRG$ is the same as the classical construction of PRF from $\PRG$  by Goldreich, Goldwasser, and Micali~\cite{GGM86}. In short, we will be referring to this construction as the GGM construction.
%\anote{Add GGM as Goldreich-Goldwasser-Micali}

\begin{theorem}%\anote{Remove all $n$ to some other letter preferably m, and m to ell because $n$ is the security parameter.}
    Assume $\mu(\secpar)\in\frac{1}{\poly}$ and $m(\secpar)\in \poly$ be such that $m(\secpar)\cdot \mu(\secpar)\in \frac{1}{\poly}$.
    Then, there exists a black box construction of a $\left(m\mu+\delta,m,\secpar\right)$-$\botPRF$ family (see \cref{def:bot-prf-correctness,def:bot-prf-security}) from any multi-time $\bot$-pseudorandom $(\mu,2\secpar)$-$\botPRG$ (see \cref{def:bot-prf-correctness}), for some negligible function $\delta(\secpar)$.
    \label{thm:bot_prf_from_bot_prg}
\end{theorem}

\ifnum\iacr=1
    The construction is the same as the~\cite{GGM86} construction. 
 \fi
The proof is given in \ifnum\iacr=1Supplement~\ref{sec:bot_prf_from_bot_prg-proof}\else \Cref{sec:bot_prf_from_bot_prg-proof}\fi. %\anote{add reference}
%\subsection{Proof of Theorem \ref{thm:bot_prf_from_bot_prg}}

Combining \Cref{thm:bot_prf_from_bot_prg} with \Cref{cor:bot-prg-from-sprs}, we get the following result.
\begin{corollary}\label{cor:bot-PRF-from-SPRS}
    Assume $m(\secpar)\in \poly$ and a constant $c$ be such that $m(\secpar)\cdot \secpar^{-c/6 +1}\in \frac{1}{\poly}$. Let $\eta(\secpar)\equiv m(\secpar)\cdot \secpar^{-c/6 +1} $.
    Then, there exists a black box construction of a $\left(\eta,m,\secpar\right)$-$\botPRF$ family (see \cref{def:bot-prf-correctness,def:bot-prf-security}) from any $(c\log\secpar)$-PRS (see \cref{def:bot-prf-correctness}), for some negligible function $\delta(\secpar)$.
\end{corollary}

\ifnum\iacr=0
    % The construction is the same as the~\cite{GGM86} construction. The proof is given in \ifnum\iacr=1Supplement\else Section\fi~\ref{sec:bot_prf_from_bot_prg-proof}. %\anote{add reference}

 \subsection{Proof of Theorem \ref{thm:bot_prf_from_bot_prg}}
\label{sec:bot_prf_from_bot_prg-proof}

%\begin{proof}[Proof of Theorem \ref{thm:bot_prf_from_bot_prg}]
    The proof follows by combining \Cref{prop:correctness-bot-PRF-family,prop:pseudorandomness-adaptive-PRF}.
%\end{proof}
\paragraph*{Notations.}\label{pg:notations-ggm}
In this section, we will use capital letters such as $X,Y$ to denote random variables. Moreover, for $y\in \{0,1\}^*$, $r\leq |y|$ and $b\in \{0,1\}$, $y_{r,b}$ denotes the first $r$ prefix of $y$ if $b=0$, and the last $r$ suffix of $y$ if $b=1$. Similarly, for any algorithm $F$ with input $\{0,1\}^\secpar$ and output $\{0,1\}^{2\secpar}\cup\{\bot\}$, and for $b\in \{0,1\}$ and any input $x\in \{0,1\}^\secpar$, we denote $F_b(x)$ to denote $\{F(x)\}_{\secpar,b}$, i.e., the first half of $F(x)$ if $b=0$ or the last half of $F(x)$ if $b=1$. 
%\anote{Mention that the proof of correctness is a bit weird in the sense that we need the pseudorandomness of the pdprg for the correctness proof. However, we only need single time pseudorandomness of the pdprg.}

\begin{figure}[!htb]
   \begin{center} 
   \begin{tabular}{|p{16cm}|}
    \hline 
% \begin{center} 
% \end{center}
\noindent\textbf{Assumes:} $(\mu(\secpar),2\secpar)$-$\botPRG$ $\BPRG$ (see \cref{def:multi-time-pseudorandomness,def:bot-PRG})\\ 
\ \\
\noindent$\treeprf(k,x)$: 
\begin{compactenum}
    \item Set $k_0=k$.
    \item For $i\in [m]$,
    \begin{enumerate}
        \item Compute $\BPRG(k_{i-1})$. \label{line:k_i}
        \item If the outcome is $\bot$, abort and output $\bot$.
        \item Else set $k_i=(\BPRG(k_{i-1}))_{\secpar,x_i}$, where for $y\in \{0,1\}^{2\secpar}$ and $b\in \{0,1\}$, $y_{\secpar,b}$ denotes the first $\secpar$ prefix of $y$ if $b=0$, and the last $\secpar$ suffix of $y$ if $b=1$.
    \end{enumerate}
    \item Output $k_m$.
\end{compactenum}
\ \\ 
\hline
\end{tabular}
    \caption{A $\botPRF$ family $\mathcal{F}=\{\treeprf(k,\cdot)\}_{k\in \{0,1\}^\secpar}$ with domain $\{0,1\}^{m(\secpar)}$ and co-domain $\{0,1\}^{\secpar}$, for any polynomial $m(\secpar)$. This is a natural extension of the GGM construction~\cite{GGM86} of $\prf$ from $\prg$, in the recognizable abort setting.}
    \label{fig:bot-PRF-GGM-construction}
    \end{center}
\end{figure}

\begin{remark}
    Even though the construction in $\Cref{fig:bot-PRF-GGM-construction}$, only considers as input $x\in\{0,1\}^m$, the $\treeprf$ is well defined for every $x\in \{0,1\}^*$. This fact will be heavily used both in the correctness and security proofs.
\end{remark}

\begin{proposition}[Correctness]\label{prop:correctness-bot-PRF-family}
Assuming $\BPRG$ is a $(\mu(\secpar),2\secpar)$-$\botPRG$, the construction in \cref{fig:bot-PRF-GGM-construction} is a $(m\mu+\delta,m,\secpar)$-$\bot$-function family provided $m(\secpar)\cdot \mu(\secpar)\in \frac{1}{\poly}$, for some negligible function $\delta(\secpar)$.
% Let $\mu(\secpar)\in\frac{1}{\poly}$ and $m(\secpar)\in \poly$ be such that $m(\secpar)\cdot \mu(\secpar)\in \frac{1}{\poly}$. \
% Assuming $\BPRG$ is a $(\mu(\secpar),2\secpar)$-$\botPRG$ (see \cref{def:bot-PRG,def:multi-time-pseudorandomness}) such that $m(\secpar)\cdot \mu(\secpar)\in \frac{1}{\poly}$, the construction in \cref{fig:bot-PRF-GGM-construction} is a $(m\mu+\delta,m,\secpar)$-$\bot$-function family (see \cref{def:bot-prf-correctness}), for some negligible function $\delta(\secpar)$.
\end{proposition}

We remark that the proof of correctness is somewhat non-standard in the sense that it requires the pseudorandomness of the underlying $\botPRG$. However,  we only need single-time pseudorandomness of the $\botPRG$.

\begin{proof}[Proof of \Cref{prop:correctness-bot-PRF-family}]

\textbf{Condition~\ref{it:far-from-bot} of correctness (\cref{def:bot-prf-correctness}) with respect to the parameters in \cref{prop:correctness-bot-PRF-family}.}

%Next, we prove the second condition (\cref{it:bot-determinism} in \cref{def:bot-prf-correctness}) of the correctness of a bot-function family. 
We will use induction on $m$, the height of the GGM tree. 
As defined in the proposition, let $\mu$ be an inverse polynomial function denoting an upper bound on the fraction of non-good or bad keys (see \Cref{def:bot-pseudodeterminism-part1} in \Cref{def:bot-PRG}) for $\BPRG$, the underlying $\botPRG$, and let $\epsilon$ be a negligible function denoting a bound on the pseudorandomness error for $\BPRG$. Let $\negl$ be the negligible function denoting the bound on the probability of seeing $\bot$ under a good key.

Induction statement $P(m)$: For every $x\in \{0,1\}^m$, $\Pr[K\xleftarrow{\$}\{0,1\}^\secpar: \prf(K,x)= \bot]\leq m(\mu+\negl)+(m-1)\epsilon$.%\anote{Make the random variable as large letters everywhere.}

Base Case $(m=1)$: Fix an input $b\in \{0,1\}$. Then,
\begin{align}
    &\Pr[K\xleftarrow{\$}\{0,1\}^\secpar: \prf(K,b)= \bot]\\
    %&=\Pr[Y\gets {\BPRG(K)}, K\xleftarrow{\$}\{0,1\}^\secpar: Y=\bot]\\
    &=\Pr[K\xleftarrow{\$}\{0,1\}^\secpar: \BPRG(K)=\bot]\\
    &\leq \mu + \negl. \label{eq:ind-hypothesis}
\end{align}%\anote{Add a theorem reference in the last equation ask Lior if that is proven.}
The last inequality follows from the fact that either $K$ happens to be a good key in which case it will evaluate to $\bot$ with probability at most $\negl$, or $K$ is a bad key but that happens with probability at most $\mu$.

Induction step: Assuming $P(m)$ holds, we need to prove that $P(m+1)$ holds.
Fix an input  $x'\in \{0,1\}^{m+1}$, and let ${x'}_{m,1}$ be the last $m$ suffix of $x'$. 
%\noindent \underline{$\hybrid_0$}: Same as the original $\treeprf$
Note that we want to bound the probability of the event $Y= \bot$ where $Y$ is the random variable defined by the following process.

\begin{itemize}
    \item Sample $K \xleftarrow{\$}\{0,1\}^\secpar$.
    \item Evaluate $Y\gets  \treeprf(K,x')$.
\end{itemize}

Note that the definition of  $Y$ can be equivalently rewritten as the following hybrid.

%\anote{Switch to $\BPRG_b$ notation and define it wherever possible and in places, where it is not enough use the suffix, prefix notation, currently introduced in the construction}
\noindent \underline{$\hybrid_0$}:
\begin{itemize}
    \item Sample $K\xleftarrow{\$}\{0,1\}^\secpar$.
    \item Evaluate $Y'\gets  \BPRG_{x_1}(K)$.
    \item If the outcome $Y'\neq \bot$, then perform $Y\gets \treeprf(Y',x'_{1,m})$ else set $Y=\bot$.
\end{itemize}
It is enough to show that $\Pr[Y\sim \hybrid_0 : Y= \bot]\leq (m+1)\mu+m\epsilon$.
Next we define the following hybrid definition of $Y$ in $\hybrid_1$.

\noindent \underline{$\hybrid_1$}:
\begin{enumerate}
    \item Sample $K\xleftarrow{\$}\{0,1\}^\secpar$, and $U\xleftarrow{\$}\{0,1\}^\secpar$. 
    \item Evaluate $Y''\gets  \BPRG_{x_1}(K)$. 
    \item If $Y''=\bot$, set $Y'=\bot$ else set $Y'=U$.\label{it:bot-check-first-step.}
    \item If the outcome $Y'\neq \bot$, then perform $Y\gets \treeprf(Y',x'_{1,m})$, else set $Y=\bot$.
\end{enumerate}

Note that if the probabilities of the events $Y=\bot$ in $\hybrid_0$ and $\hybrid_1$ differ by a non-negligible amount $p$, then we get an efficient distinguisher between $\hybrid_0$ and $\hybrid_1$ with advantage $p$. It is easy to see that if the processes $\hybrid_0$ and $\hybrid_1$ can be distinguished efficiently, then there exists a distinguisher against single-time $\botPR$ of $\BPRG$ with the same advantage. (Given a sample $Y'$, just perform the last step that is common in both $\hybrid_0$ and $\hybrid_1$ and output $1$ if the outcome $Y= \bot$.)
 Therefore by the single-time $\botPR$ of $\BPRG$,
 \begin{equation}\label{eq:pseudorandomness-in-correctness}
 \Pr[Y\sim \hybrid_0: Y= \bot]\leq \Pr[Y\sim \hybrid_1: Y= \bot]+\epsilon.
 \end{equation}
 Note that in $\hybrid_1$, the event $Y=\bot$ happens only if either the event $Y''=\bot$ happens or the event $\treeprf(U,x'_{1,m})$ outputs $\bot$, happens.
Therefore,
\begin{align}\label{eq:hybrid_1}
&\Pr[Y\sim \hybrid_1: Y= \bot]\\
&\leq \Pr[K\xleftarrow{\$}\{0,1\}^\secpar, Y''\gets \BPRG(K):Y''=\bot]\\
+&\Pr[Y\gets \treeprf(U,x'_{1,m}), U\xleftarrow{\$}\{0,1\}^\secpar:Y=\bot].
\end{align}
Note that the first term in the last equation is at most $(\mu+\negl)$, by \Cref{eq:ind-hypothesis}. %the fact that at least $(1-\mu)$ fraction of the keys are excellent for $\BPRG$, and that for an excellent key the probability of outputting bot is at most $\negl$\anote{Add reference}, it holds that $\Pr[Y''\gets \BPRG_{x_1}(K):Y''=\bot]\leq \mu + \negl$. 
Since $x'_{1,m}$ is of length $m$ by the induction hypothesis,
\[\Pr[Y\gets \treeprf(K,x'_{1,m}), K\xleftarrow{\$}\{0,1\}^\secpar:Y=\bot]\leq m(\mu + \negl) +(m-1)\epsilon.\]

Combining the last two bounds with \cref{eq:hybrid_1}, we conclude that,
\[\Pr[Y\sim \hybrid_1:Y=\bot]\leq \mu +\negl +  m(\mu +\negl) +(m-1)\epsilon=(m+1)(\mu + \negl) + (m-1)\epsilon.\]

Combining this with \cref{eq:pseudorandomness-in-correctness}, we conclude that
\begin{align*}
&\Pr[Y\sim~\hybrid_0:Y=\bot]\\
&\leq \Pr[Y\sim \hybrid_1:Y=\bot] +\epsilon \\
&\leq  (m+1)(\mu +\negl) + (m-1)\epsilon+\epsilon=(m+1)(\mu+\negl) + m\epsilon.
\end{align*}

Since $x'$ was fixed arbitrarily, this completes the proof for the induction step.
Hence, we conclude that $P(m)$ holds for every $m\in \NN$.

Since $m\in \poly$, and $\negl,\epsilon(\secpar)$ are negligible functions, it further holds that $\delta(\secpar)\equiv m\negl + (m-1)\epsilon$ is a negligible function. Combining this with the fact that $P(m)$ holds for every $m\in\NN$, we conclude that, there exists a negligible function $\delta(\secpar)$ such that for every $x\in \{0,1\}^m$,
\[\Pr[Y\gets\prf(K,x); K\xleftarrow{\$}\{0,1\}^\secpar: Y= \bot]\leq m\mu+\delta(\secpar)\]
which concludes Condition~\ref{it:far-from-bot} with parameter $m\mu+\delta(\secpar)$ of correctness (see \cref{def:bot-prf-correctness}) for the construction given in \cref{fig:bot-PRF-GGM-construction}.
% after sampling $K\xleftarrow\{0,1\}^\secpar$, the evaluation of $\treeprf(K,x')$ is the same as first evaluating $Y'\gets \BPRG_{x_1}(k)$ and if the outcome $Y'\neq \bot$ then performing the tree PRF evaluation with $x'_{1,r}$ (which is of length $r$ so we can use induction hypothesis) as the input and $Y'$ as key to getting an output $Y$. We want to calculate the probability of the event $Y\neq \bot$. 
% Therefore, the probability that we want
% \[\Pr[Y\gets\prf(K,x'); K\xleftarrow{\$}\{0,1\}^\secpar: Y\neq \bot]=\Pr[Y\gets\prf(Y',x'); K\xleftarrow{\$}\{0,1\}^\secpar: Y\neq \bot]\]

% Note that the above process can be replaced with 

\textbf{Condition~\ref{it:bot-determinism} of correctness (\cref{def:bot-prf-correctness})}
The proof idea is to do induction on $m$ just as in the proof of the first condition.

   Induction statement: $P(m)$ states that for every input $x\in \{0,1\}^m$ and key $k\in \{0,1\}^\secpar$ there exists a string $y\in\{0,1\}^\secpar$, such that the tree $\PRF$  evaluation (see \cref{fig:bot-PRF-GGM-construction}) output is not in $\{y,\bot\}$ with probability at most$m\cdot \negl$, where $\negl$ is the negligible pseudo determinism bound on the $\botPRG$ $\BPRG$.

Base step ($m=1$): Fix an input bit $b$ and key $k\in\{0,1\}^\secpar$ arbitrarily. Note that the tree $\PRF$  evaluation is simply to evaluate $\BPRG_b(k)$. By the pseudo determinism of $\BPRG$, there exists $y\in \{0,1\}^{2\secpar}$ such that $\Pr[\BPRG(k)\not\in \{y,\bot\}]\leq \negl$. Since the event $\BPRG(k)\in \{y,\bot\}$ implies $\BPRG_b(k)\in \{y_{\secpar,b}\}$, we conclude that
\[\Pr[\BPRG_b(k)\not \in \{y_{\secpar,b},\bot\}]\leq \Pr[\BPRG(k)\not\in \{y,\bot\}] \leq \negl.\]
%Follows directly from the $\bot$-$\PDPRG$ condition.     

Induction step:    We assume that $P(m)$ holds, i.e., for every input $x\in \{0,1\}^m$ and key $k\in \{0,1\}^\secpar$ there exists a string $y\in\{0,1\}^{\ell(\secpar)}$, such that the tree $\PRF$  evaluation output is not in $\{y,\bot\}$ with probability at most $m\cdot\negl$.%, where $p=p(\secpar)$ is the pseudo determinism bound on the $\bot$-$\PDPRG$ $\BPRG$.
 We need to show $P(m+1)$ holds.
Let $k\in \{0,1\}^\secpar$ and $x'\in \{0,1\}^{m+1}$, be fixed, and as seen before, let $x'_{m,1}$ be the last $m$ suffix of $x'$. The evaluation of $\treeprf(k,x')$, i.e., the tree $\PRF$  evaluation with $x'$ as the input and $k$ as the key is the same as first evaluating $Y'\gets \BPRG_{x_1}(k)$ and if the outcome $Y'\neq \bot$ then performing the tree $\PRF$  evaluation with $x'_{1,m}$ (which is of length $m$ so we can use induction hypothesis) as the input and $Y'$ as key to getting an outcome $Y$.
By the pseudo determinism of $\BPRG$, for every $k$, there exists $\tilde{y}\in \{0,1\}^{2\secpar}$ such that
\[\Pr[\BPRG(k)\not\in\{\tilde{y},\bot\}]\leq \negl.\]
Let $y'=\tilde{y}_{\secpar,x_i}$. Hence, we conclude that 
\begin{equation}\label{eq:bot-pseudo-determinism-prg}
\Pr[Y'\gets \BPRG_{{x'}_i}(k): Y'\not \in\{y',\bot\}]\leq \Pr[\BPRG(k)\not\in\{\tilde{y},\bot\}]\leq \negl.
\end{equation}
Note that the induction hypothesis on $y'$ as key and $x'_{1,m}$ as input, there exists $y$ such that
\begin{equation}\label{eq:by-induction-hypothesis}
\Pr[Y\gets \treeprf(y',x'_{1,m}):  Y\not\in \{y,\bot\}]\leq m\cdot\negl.
\end{equation}
Note that the evaluation of $\treeprf(k,x')\not \in \{y,\bot\}$ only if either the evaluation of the first step  $\BPRG_{x_i}(k)\not\in\{y',\bot\} $ or   $\BPRG_{x_i}(k)=y'$ but $\treeprf(y',x'_{1,m})\not\in \{y,\bot\}$.
Hence,
\begin{align*}
    &\Pr[Y\gets \treeprf(k,x'): Y\not \in \{y,\bot\}]\\
    &\leq \Pr[Y'\gets \BPRG_{{x'}_i}(k): Y'\not \in\{y',\bot\}]\\
    + &\Pr[Y\gets \treeprf(y',x'_{1,m}):  Y\not\in \{y,\bot\}]\\
    &\leq \negl + m\cdot\negl &\text{By \Cref{eq:bot-pseudo-determinism-prg,eq:by-induction-hypothesis} respectively.}\\
    &=(m+1)\cdot\negl.
\end{align*}

% Note that $\Pr[Y\in\{y
% _{y'},\bot\}\mid Y'=\bot]=\Pr[Y=\bot\mid Y'=\bot]=1$. Since $Y'=\bot$ and $Y'=y'$, are disjoint events the last two inequalities imply 

% \begin{equation}\label{eq:conditional-ind-hypothesis}
% \Pr[Y\in \{y,\bot\}\mid Y'\in \{y',\bot\}]\geq \min\left[(1-\negl)^r,1\right]=(1-\negl)^r.
% \end{equation}
% Combining \cref{eq:bot-pseudo-determinism-prg,eq:conditional-ind-hypothesis}, we conclude that
% \begin{align}
% &\Pr[Y\in \{y,\bot\}]\\
% &\geq \Pr[Y\in \{y,\bot\}\wedge Y'\in \{y',\bot\}]\\
% &=\Pr[Y\in \{y,\bot\}\mid Y'\in \{y',\bot\}]\Pr[Y'\in\{y',\bot\}]\\
% &\geq (1-\negl)^r\cdot(1-\negl)=(1-\negl)^{r+1}.
% \end{align}
Since $x'$ was arbitrary, this concludes the proof for the induction step.

Therefore, we conclude by induction that for every $m\in \NN$, input $x\in \{0,1\}^m$ and key $k\in \{0,1\}^\secpar$ there exists a string $y\in\{0,1\}^\secpar$, such that the tree $\PRF$  evaluation (see \cref{fig:bot-PRF-GGM-construction}) output is in $\{y,\bot\}$ with probability at least $(1-\negl)^m\geq 1-m\cdot \negl$. Since $m\in\poly$ and $\negl$ is a negligible function, this concludes the proof of Condition~\ref{it:bot-determinism} of correctness (see \cref{def:bot-prf-correctness}) for the construction given in \cref{fig:bot-PRF-GGM-construction}.

\end{proof}

% The proposition has the immediate following corollary.
% \begin{corollary}\label{cor:correctness-bot-PRF-family}
%     If $n(\secpar),\mu(\secpar)\in \frac{1}{\poly}$ be such that $n(\secpar)\mu(\secpar)\in \frac{1}{\poly}$, then since $\epsilon(\secpar),\negl(\secpar)$ are negligible functions and $n(\secpar)\in \poly$ the construction in \cref{fig:bot-PRF-GGM-construction} satisfies correctness (see \cref{def:bot-prf-correctness}) with parameters $(n(\secpar)\cdot\mu(\secpar)+\negl,n(\secpar)\cdot\nu(\secpar))$ respectively for some negligible function $\negl$.
% \end{corollary}

\begin{proposition}[Pseudorandomness]\label{prop:pseudorandomness-adaptive-PRF}
Assuming the underlying $\botPRG$, $\BPRG$ satisfies multi-time $\botPR$ (see \Cref{def:multi-time-pseudorandomness}), the construction in \Cref{fig:bot-PRF-GGM-construction} satisfies pseudorandomness (see \Cref{def:bot-prf-security}). %$\botPRF$ adaptive\anote{Remove adaptive but add a remark stressing that the previous works do not achieve it.} security if the underlying $\botPRG$ $\BPRG$ satisfies $q$-fold parallel repetition of the multi-sample pseudorandomness.
\end{proposition}

\begin{proof}[Proof of \Cref{prop:pseudorandomness-adaptive-PRF}]
    The proof follows directly from extending the GGM security to the case of recognizable abort.
Suppose $\BPRG$ satisfies multi-time $\botPR$. 
    Fix a QPT distinguisher $D$ against the $\botPRF$ (\Cref{fig:bot-PRF-GGM-construction}) in $\prfdistingexpt$ (see \Cref{fig:botprfadaptexpt}). %\onote{Hopefully the next sentence could be removed} Let $q$ be the polynomial denoting the number of queries $D$ makes to the oracle given to her.

    We note that even though the security needs to be proven for inputs in $\{0,1\}^m$, the $\botPRF$ is well-defined for inputs of arbitrary length, i.e., for every $x\in \{0,1\}^*$.
    We use the term \emph{tree evaluation of input $x\in \{0,1\}^* $ and key $k$} or $\treeprf(k,x)$ to denote the $\botPRF$ evaluation with input $x$ and key $k$. 

%For every $0\leq i\leq m-1$, and input $z\in \{0,1\}^{m-i}$, and $k_i\in \{0,1\}^\secpar$, we use the term tree PRF evaluation of input $z$ and key $k_i$, to denote the $\prf$ evaluation with respect to the input $z$ and key $k_i$ using the $\bot$-$\PDPRG$ $\BPRG$ as per the construction given in \cref{fig:bot-PRF-GGM-construction}.
%For an input $z \in \{0,1\}^{m-i}$, and a key $k\in \{0,1\}^\secpar$, recall the definition of $k_i$ as defined in \cref{line:k_i} for all $i\in \{0,\ldots,m\}$ for that input $z$.  We use the term tree PRF evaluation of input $z$ and key $k$ for the value $k_i$.
%For $i=m$, the tree PRF evaluation of input $z$ and a key $k_m$ simply outputs $k_m$.

We do a series of hybrid arguments $\hybrid_0,\hybrid_1,\ldots,\hybrid_m$. For every $0\leq i\leq m$, $\hybrid_i$ is defined as follows.

\noindent \underline{$\hybrid_i$}:
\begin{enumerate}
\item Challenger $\ch$ samples a bit $b\xleftarrow{\$}\{0,1\}$ and keys $K_x$ for every $x\in \bigcup_{j=0}^{i}\{0,1\}^j$, where $\{0,1\}^0$ is the singleton string $\{\treeroot\}$.
\item $\ch$ samples a uniformly random function $f$ from $m$ bits to $\ell$ bits.
\item For every query $x$ made by the distinguisher $D$, $\ch$ evaluates $\BPRG(K_{x_{j,0}})$ for every $j\leq i-1$ in ascending order, and outputs $\bot$ if the evaluation fails at any stage, where recall that $x_{j,0}$ is the first $j$ prefix of $x$. %\anote{Use the previous notation}
\item If none of the evaluations output $\bot$ in the previous stage, use $K_{x^i}$ as the key and $x^{-(m-i)}$ as the input to do the tree $\PRF$  evaluation and outputs $\bot$ if the outcome is $\bot$. \label{it:bot-check}
\item Otherwise, if $b=0$, $\ch$ outputs the outcome of the tree evaluation, and if $b=1$, outputs $F(x)$.\label{it:final-output}
%\item If $b=0$, $\ch$ gives $D$ oracle access to $f_k(\cdot) $, and if $b=1$, gives oracle access to $g(\cdot)$.
\item $D$ outputs a bit $b'$.
\item $D$ wins if $b'=b$.
\end{enumerate}

It is easy to see that $\hybrid_0$ is the same as $\prfdistingexpt(\secpar)$ (see \Cref{fig:botprfadaptexpt}) for the $\bot$-function class defined in \cref{fig:bot-PRF-GGM-construction}, where $K_\treeroot$ in $\hybrid_0$ corresponds to $K$ in $\prfdistingexpt(\secpar)$.

\begin{lemma}[Impossibility to win in $\hybrid_n$]\label{lemma:final-hybrid-guarantee}
For every distinguisher $D$,
    \[\Pr[\text{$D$ wins in $\hybrid_n$}]= \frac{1}{2}.\]
\end{lemma}

\begin{lemma}[Indistinguishability of hybrids]\label{lemma:indistinguishability-hybrids}
    For every $\epsilon$ and for every $i\in [n]$ and any QPT distinguisher $D$ satisfying, 
    \[|\Pr[\text{$D$ wins in $\hybrid_{i-1}$}]-\Pr[\text{$D$ wins in $\hybrid_{i}$}]|=\epsilon,\] there exists a QPT distinguisher $D_i$ that breaks the $t$-fold multi-time $\botPR$ (see \Cref{def:multi-time-pseudorandomness}) of $\BPRG$ with probability $\epsilon$, where $t\equiv t(\secpar)$ is a polynomial representing the number of queries that $D$ makes.
\end{lemma}

%\anote{change q-sample security to  multi-time $\botPR$ everywhere.}
Before proving the lemmas, we show how to conclude the theorem assuming the lemmas. Note that by \Cref{thm:parallel-repetition-multi-sample-bot-pseudorandomness}, since $\BPRG$ satisfies multi-time $\botPR$, it also satisfies the $t$-fold  multi-time $\botPR$ since $t(\secpar)\in \poly$. Therefore by \Cref{lemma:indistinguishability-hybrids}, it holds that that for every $i\in [n]$, there exists a negligible function $\epsilon_i$, such that
\[|\Pr[\text{$D$ wins in $\hybrid_{i-1}$}]-\Pr[\text{$D$ wins in $\hybrid_{i}$}]|=\epsilon_i.\]

Therefore, by triangle inequality,
\[|\Pr[\text{$D$ wins in $\hybrid_{0}$}]-\Pr[\text{$D$ wins in $\hybrid_{n}$}]|\leq \sum_{i=1}^n\epsilon_i\leq n\cdot\max_{i\in [n]} \epsilon_i,\]
which is a negligible function of $\secpar$ since $n(\secpar)$ is a polynomial function of $\secpar$.
Combining the last result with \Cref{lemma:final-hybrid-guarantee}, we conclude that there exists a negligible function $\negl$, such that
\begin{align}
&\Pr[\prfdistingexpt(\secpar)=1]\\
&=\Pr[\text{$D$ wins in $\hybrid_{0}$}]\\
&\leq\Pr[\text{$D$ wins in $\hybrid_{n}$}]+\negl &\text{By the last equation}\\
&= \frac{1}{2}+\negl, &\text{By the \Cref{lemma:final-hybrid-guarantee}}
\end{align}
which concludes the theorem.
Next, we prove \Cref{lemma:final-hybrid-guarantee,lemma:indistinguishability-hybrids} to complete the proof.

\begin{proof}[Proof of \Cref{lemma:final-hybrid-guarantee}]
Note that in $\hybrid_n$, note that the challenger samples $K_x\xleftarrow{\$}\{0,1\}^\secpar$ independently.  Let $H(\cdot)$ be defined as $H(x)\equiv K_x$ for every $x\in \{0,1\}^n$. Clearly $H$ is a uniformly random function from $\left(\{0,1\}^{\secpar}\right)^{\{0,1\}^{\secpar}}$, and hence is identically distributed as $F$.
Hence in $\hybrid_n$, for every query $x$, if the outcome of the tree-evaluation is not $\bot$ in Item~\ref{it:bot-check}, then the outcome is $H(x)$, and hence as instructed in Item~\ref{it:final-output}, $\ch$ outputs $H(x)$ if $b=0$, and $F(x)$ if $b=1$. Since $F$ and $H$ are identically distributed independent of the rest of the experiment and the only difference between $b=0$ and $b=1$ case in $\hybrid_n$, is the use of $H$ and $F$ respectively, we conclude that the view of the distinguisher is independent of the sampled bit $b$, and hence 
\[\Pr[\text{$D$ wins $\hybrid_n$}]= \frac{1}{2}.\]
    
\end{proof}

\begin{proof}[Proof of \Cref{lemma:indistinguishability-hybrids}]
Recall that $t=t(\secpar)$ is the polynomial denoting the number of queries that are made by $D$.
Fix $i\in [n].$
We define $D_i$ that runs $D$ as follows. 
\begin{itemize}
    \item $D_i$ gets as input $t^2$ $2\secpar$-strings $\{\{v^j_r\}_{r\in [t]}\}_{j\in [t]}$, i.e., in $t$-tuples each tuple consisting of $t$ strings. Let $V^j=\{v^j_r\}_{r\in [t]}$ for every $j\in [t]$.
    \item $D_i$ plays the role of the challenger, samples a bit $b$, and runs $D$ by simulating the oracle as follows.
    \item For every query $x$, $D_i$ does the following:
    \begin{itemize}
        \item For every $0\leq j\leq i-2$, check if there exists a key $K_{x_{j,0}}$ associated to the $j$-bit string $x_{j,0}$, and if not found sample a key $K_{x_{j,0}}$ uniformly at random and mark it for the string $x_{j,0}$, where recall that $x_{j,0}$ denotes the first $j$ prefix of $x$ and $x_{0,0}$ denotes the fixed special symbol $\{\treeroot\}$. \label{it:D_i-sample-keys}
        \item Iteratively, evaluate $\BPRG(K_{x_{j,0}})$ for $0\leq j \leq i-2$, and abort if the outcome is $\bot$ at any stage.
        \item If all $i-1$ evaluations succeeded, i.e., all of their output was not $\bot$, check if there exists $j\in [t]$ such that $V^j$ is associated with the string $x_{i-1,0}$.
        \begin{itemize}
            \item If such a $j$ exists, select the first unused string from $V^j$, i.e., associated with no string, and mark it as associated with $x_{i-1,0}$. \label{it:find-old-tuple}
            \item Else, go to the first unused tuple $V^j$, i.e., associated with no string, and select the first string $v^j_1$ from it and mark the string $v^j_1$ as used. \label{it:find-new-tuple}
        \end{itemize}
        \item Let the selected $2\secpar$ string be $v$. Set $K_{x_{i,0}}\equiv v_{\secpar,x_i}$, i.e., if $x_i=0$ set $K_{x_{i,0}}$ to be the first $\secpar$ prefix of $v$, else set it to be the last $\secpar$ suffix of $v$.\label{it:insert-challenge}
        \item Perform the tree evaluation using $K_{x_{i,0}}$ as the key and $x_{\secpar-i,1}$ as the input, and output $\bot$ if the evaluation outcome is $\bot$.
        \item Otherwise if $b=0$, $\ch$ outputs the outcome of the tree evaluation, and if $b=1$, outputs $F(x)$.\label{it:final-output-D_i}
        %\item If $b=0$, $\ch$ gives $D$ oracle access to $f_k(\cdot) $, and if $b=1$, gives oracle access to $g(\cdot)$.
    \end{itemize}
  \item After the $t$-queries, $D$ outputs a bit $b'$.
  \item $D_i$ outputs $1$ if $b'=b$, and $0$ otherwise.    
\end{itemize}

Note that for every query $x$, $D$ only samples at most $n(\secpar)$ keys, in Step~\ref{it:D_i-sample-keys}, and there are only $t$-queries made by $D$. Clearly, the tree evaluation done is efficient. Therefore, simulating the queries made by $D$ can be done in polynomial time. Hence we conclude that the running time of $D_i$ is only a polynomial overhead of that of $D$, and hence $D_i$ is efficient. 

Next, note that while finding an unused tuple $V_j$ in Step~\ref{it:find-new-tuple}  or while finding an unused string inside a used tuple in $D_i$ needs to do either of Step~\ref{it:D_i-sample-keys} or  Step~\ref{it:find-new-tuple} only once per query and there are almost $t$ queries and $D_i$ was given as input $t$-tuples with each containing $t$-strings.
Next, note that if the $t^2$ strings are given to the adversary are evaluated by applying $\BPRG$ on the same $t$-times, for $t$ independent keys, then in the challenge simulation, Step~\ref{it:insert-challenge} is equivalent to sampling a key $K_{x_{i-1,0}}$ (if not sampled previously, otherwise using the previously sampled key), and doing the tree evaluation with $K_{x_{i-1,0}}$ as the key and $x_{\secpar-i+1,1}$ as the input. Hence the view of $D$ in this case is exactly the same as $\hybrid_{i-1}$.
On the other hand, if the $t^2$ input strings were sampled using $t$ uniformly and independently random strings, then in the challenge simulation, Step~\ref{it:insert-challenge} is equivalent to sampling a key $K_{x_{i-1,0}}$ (if not sampled previously, otherwise using the previously sampled key), then evaluating the $\BPRG$ on it and output $\bot$ if the outcome is $\bot$, else sampling a key $K_{x_{i,0}}$ (if not sampled previously, otherwise using the previously sampled key) which is used for all other queries as well, and doing the tree evaluation with $K_{x_{i,0}}$ as the key and $x_{\secpar-i,1}$ as the input. Hence the view of $D$ in this case is exactly the same as $\hybrid_{i}$.
Therefore, the probability of $D_i$ outputting $1$ in the $\BPRG$ case (respectively, the random case) is exactly the same as the probability of $D$ winning in $\hybrid_{i-1}$ (respectively, $\hybrid_i$). Therefore the distinguishing advantage of $D_i$ in the $t$-fold multi-time $\botPR$ is the same as

    \[|\Pr[\text{$D$ wins in $\hybrid_{i-1}$}]-\Pr[\text{$D$ wins in $\hybrid_{i}$}]|=\epsilon,\]
    which concludes the proof of the lemma.
\end{proof}
    
\end{proof}
% \begin{theorem}
%     Any adaptive-secure $\botPRF$ family also satisfies quantum adaptive security, i.e., security against quantum adaptive queries.
% \end{theorem}

% \begin{proof}
%     For the sake of contradiction, let $D$ be a quantum distinguisher making $q(\secpar)$ quantum queries for some polynomial $q(\cdot)$, with a distinguishing advantage $\epsilon$.  
% \end{proof}

\begin{remark}
    The previous works on $\PDPRF$s only achieve a nonadaptive form of pseudorandomness whereas in this work, we achieve adaptive pseudorandomness which is crucial for many of the applications considered in the work.
\end{remark}

% Combining \cref{prop:pseudorandomness-adaptive-PRF} with \cref{thm:parallel-repetition-multi-sample-bot-pseudorandomness}, we conclude the following corollary.

% \begin{corollary}\label{cor:security-of-construction-from-simpler-assumption}
%     The construction in \cref{fig:bot-PRF-GGM-construction} satisfies $\botPRF$ adaptive security if the underlying $\botPRG$ $\BPRG$ satisfies multi-sample pseudorandomness.
% \end{corollary}

\fi

\ifnum\iacr=0

\section{\texorpdfstring{$\bot$}{bot} Universal One-Way Hash Functions from \texorpdfstring{$\bot$}{bot} \textsf{PRG}}
\label{sec:qpduowhfsect}

In this section, we construct $\bot$-Universal One-Way Hash Functions ($\botUOWHF$) from $\botPRG$. First, we define $\botOWF$s and $\botUOWHF$. We introduce notions of entropy that capture the size of collision sets under any $\bot$-function. Next, we build $\botOWF$ from a $\botPRG$. With these tools at hand, we are able to adapt the proof of \cite{HHR+10} (which shows how to construct \textsf{UOWHF}s from \textsf{OWF}s via inaccessible entropy) to our setting achieving $\botUOWHF$s from $\botPRG$s.  

\subsection{Definitions}
We introduce the notion of $\botOWF$s\ifnum\iacr=1. \else and $\botUOWHF$.\fi \ifnum\iacr=1 Just like in the case of $\botUOWHF$s, the \else The \fi problem is that we need to allow the function to map to $\bot$ on a non-negligible fraction of inputs for this construction to work. Clearly, this allows an adversary to easily break security as it is easy to guess an inverse to $\bot$. Fortunately, in our applications, we do not care if the adversary find an inverse $\bot$ since we simply abort the operation if the output is $\bot$ anyway. Hence, in the following definition, security requires that it is difficult to find an inverse to a non-$\bot$ image. 

\begin{definition}[$\bot$-One-Way Functions]
\label{def:botOWF}
    Let $n=n(\lambda)$ and $m=m(\lambda)$ be polynomials in the security parameter $\lambda$. A family of QPT algorithms $\mathcal{F}\coloneqq \{F_\lambda \}_{\lambda\in\mathbb{N}}$ which on input $\{0,1\}^n$ output a string in $\{0,1\}^m$ or $\bot$,
    %mapping $F_\lambda:\mathcal{X}_{\lambda}\rightarrow \{0,1\}^m$ 
    is a \emph{$(\mu,n,m)$-pseudodeterministic one-way function with recognizable abort ($\botOWF$)} if the following conditions hold:
    \begin{itemize}
        \item \textbf{Pseudodeterminism:} There exist a constant $c>0$ such that $\mu(\lambda)= O(\lambda^{-c})$ and for sufficiently large $\lambda\in \mathbb{N}$ there exists a set $\mathcal{G} \subseteq \{0,1\}^n$ such that the following holds:
        \begin{enumerate}
            \item \[\Pr_{x\gets \{0,1\}^n}\left[x\in\mathcal{G} \right] \geq 1-\mu(\lambda).\] 
            \item For every $x\in \mathcal{G}$ there exists a non-$\bot$ value $y\in\{0,1\}^m$ such that: 
            \begin{align}
                \Pr\left[F_\lambda(x)=y \right] \geq 1 - \negl[\lambda].
            \end{align} 
          
            \item For every $x\in \{0,1\}^n$, there exists a non-$\bot$ value $y\in\{0,1\}^m$ such that: 
            \begin{align}
                \Pr\left[F_\lambda(x)\in \{y,\bot\} \right] \geq 1 - \negl[\lambda].
            \end{align} 
            If, additionally, there exists a polynomial $p$ such that the following is satisfied:
            \begin{align}
                \Pr\left[ F_\lambda(x)=y \right] \geq \frac{1}{p(n)},
            \end{align} 
            then, we denote this value $y$ as $F[[x]]$. Otherwise, we set $F[[x]]$ to $\perp$. 
            \end{enumerate}
        \item \textbf{Security:} For any QPT algorithm $\adv$,
        \begin{align*}
            \Pr_{x\xleftarrow{\$} \{0,1\}^n}[F(\adv(F(x)))=F_\top(x)]\leq \negl[\lambda],
        \end{align*}
        where $F_\top(\perp) :=\perp$\footnote{This might seem redundant, but we need $F^{-1}_\top(\perp)$ to be non-empty to more easily define appropriate notions of entropy.} and\onote{The old text contained this, which seems irrelevant: "$F_\top(\perp) :=\perp$ and". Note that $F_\top$ is evaluated on $x$, which is a random string, and never could be $\bot$.}\mo{Yes, but I need bot to have an inverse for future entropy arguments to be well defined.} $F_\top(x) :=\textsf{Flip}(F(x)) $\footnote{$\textsf{Flip}$ simply sends $\bot$ to $\top$ and keeps all other values the same.} otherwise.
    \end{itemize}
\end{definition}

\begin{definition}
    We say a family of QPT algorithms $\mathcal{F}\coloneqq \{F_\lambda\}_{\lambda\in\mathbb{N}}$ is a \emph{pseudodeterministic function with recognizable abort ($\bot$-function)} if it satisfies the pseudodeterministic condition of Definition \ref{def:botOWF}.
\end{definition}
\ifnum\iacr=0
    We define the notion of pseudo-deterministic universal one-way hash function families as the natural non-deterministic version of universal one-way hash functions \cite{NY89}.
    
    \begin{definition}[$\botUOWHF$]\label{def:botUOWHF}
       Let $n=n(\lambda)$ and $m=m(\lambda)$ be polynomials in the security parameter $\lambda\in\mathbb{N}$. A family of QPT algorithms $\mathcal{H}_\lambda \coloneqq \{{H}_\lambda(k,\cdot)\}_{k\in \{0,1\}^\ell} $ that on input in $ \{0,1\}^n$ output an element in $\{0,1\}^m \cup \{\bot\}$ %mapping $F_\lambda:\mathcal{X}_{\lambda}\rightarrow \{0,1\}^m$ 
        is a \emph{$(\mu,\ell, d=n/m,m)$-pseudodeterministic universal one-way hash family with recognizable abort ($\botUOWHF$)} if the following conditions hold: 
        \begin{itemize}
        \item \textbf{Compression:} $d(\lambda)>1$ for all $\lambda\in \mathbb{N}$. 
            \item \textbf{Pseudodeterminism:} There exist a constant $c>0$ such that $\mu(\lambda)= O(\lambda^{-c})$ and for sufficiently large $\lambda\in \mathbb{N}$ and $F\in \mathcal{H}_\lambda$, there exists a set $\mathcal{G} \subseteq \{0,1\}^n$ such that the following holds:
            \begin{enumerate}
                \item \[\Pr_{x\gets \{0,1\}^n}\left[x\in\mathcal{G} \right] \geq 1-\mu(\lambda).\] 
                \item For every $x\in \mathcal{G}$ there exists a non-$\bot$ value $y\in\{0,1\}^m$ such that: 
                \begin{align}
                    \Pr\left[F(x)=y \right] \geq 1 - \negl[\lambda].
                \end{align} 
              
                \item For every $x\in \{0,1\}^n$, there exists a non-$\bot$ value $y\in\{0,1\}^m$ such that: 
                \begin{align}
                    \Pr\left[F(x)\in \{y,\bot\} \right] \geq 1 - \negl[\lambda].
                \end{align} 
            \end{enumerate}
            \item \textbf{Security:} Two conditions need to be satisfied:
            \begin{enumerate}
                \item \textbf{One-wayness:} For any QPT algorithm $\adv$ and $F\in \mathcal{H}_\lambda$,
            \begin{align*}
                \Pr_{x\gets \{0,1\}^n}[F(\adv(F(x)))=F_\top(x)]\leq \negl[\lambda].
            \end{align*}
            
            \item \textbf{Collision-resistance:} For every QPT algorithm $\adv$, if $\adv$ outputs $x\in \{0,1\}^n$, then is given $F\gets \mathcal{H}_\lambda$, and then outputs $x'\neq x$, then
            \begin{align*}
                \Pr\left[F(x)=F_\top(x')\right]\leq \negl[\lambda].
            \end{align*}
            \end{enumerate}
        \end{itemize}
    \end{definition}
    
    We emphasize that while the variants of pseudodeterministic QPT algorithms introduced in this section are called ``functions'', they do not technically act as functions due to the non-determinism.\anote{There is something wrong with the last sentence. Can you fix?}\amit{Do you mean "technically they are not a function due to non-determinism"?} This terminology is used to highlight their similarity to the classical notions. 
\fi
\subsection{Entropy}
We will introduce different notions of real and accessible entropy for collisions sets under a $\bot$-function.

\begin{definition}
Let $X$ be a classical random variable. We define the following notions of entropy:
\begin{itemize}
    \item The \emph{sample entropy} (also called the {surprise}) for an element $x\in X$ is given by $H_X(x)\coloneqq \log(\frac{1}{\Pr[X=x]})$.
    \item The \emph{Shannon entropy} is given by $H(X)\coloneqq \mathbb{E}_{x\gets X}\left[H_X(x)\right]$.
    \item The \emph{max-entropy} is given by $H_0(X)\coloneqq \log(\lvert \textsf{Supp}(X)\rvert)$.
    \item The \emph{min-entropy} is given by $H_\infty(X)\coloneqq \min_{x\in X}H_X(x)$.
\end{itemize}   
\end{definition}

%We additionally require the notion of \emph{smooth min-entropy} which we do not define here but refer the reader to \cite{C97} for a formal definition.

We will need to study the entropy of inputs whose evaluations under a $\botOWF$ $F$ agree with non-negligible probability on a non-$\bot$ image. Our notions are an adaptation of those introduced in \cite{HHR+10} for standard deterministic functions. 

There are two types of entropy notions which we will need. Firstly, the \emph{real} entropy of an input $x$ corresponds to the entropy of the set of inputs whose evaluations agree with $x$ with non-negligible probability on a non-$\bot$ image. Secondly, the \emph{accessible} entropy of an input $x$ corresponds to the entropy of the set of inputs whose evaluations agree with $x$ with non-negligible probability on a non-$\bot$ image and that can be found by a QPT adversary. 

\begin{definition}
    Let $n=n(\lambda)$ and $m=m(\lambda)$ be polynomials in the security parameter $\lambda\in\mathbb{N}$ and let $\mu(\lambda)= O(\lambda^{-c})$ for some constant $c>0$. Let $F$ be a $(\mu,n,m)$-$\bot$-function. 
    \begin{itemize}
        \item $\textsf{Supp}_{F}(x)$ is defined as $\{\perp\}$ if $F[[x]]=\perp$ and $\{x'\in \{0,1\}^n:F[[x']]=F[[x]]\}$ otherwise. 
        \item We say $F_\top^{-1}$ has \emph{real Shannon entropy} $k$ if 
    \begin{align*}
        \underset{x\gets \{0,1\}^n}{\mathbb{E}}[\log(\lvert \textsf{Supp}_{F}(x)\rvert)]=k.
    \end{align*}
    \item We say $F_\top^{-1}$ has \emph{real max-entropy} at most $k$ if 
    \begin{align*}
        \Pr_{x\gets \{0,1\}^n}[\log(\lvert \textsf{Supp}_{F}(x)\rvert)\leq k]\geq 1-\negl[\lambda].
    \end{align*}
    \item We say $F_\top^{-1}$ has \emph{real min-entropy} at least $k$ if 
    \begin{align*}
        \Pr_{x\gets \{0,1\}^n}[\log(\lvert \textsf{Supp}_{F}(x)\rvert)\geq k]\geq 1-\negl[\lambda].
    \end{align*}
    \end{itemize}
    \end{definition}

We now define {accessible} Shannon and max entropy for $\bot$-functions.

%We say that a QPT algorithm $\textsf{A}$ is a \emph{$\top$-collision-finder} for a $\bot$-function $F$ if it outputs non-$\bot$ collisions. In particular, if for any $x$ in the input space, $\textsf{A}(x)$ outputs a value $x'$ such that $\Pr_{x\gets \{0,1\}^n}\left[x'\gets \textsf{A}(x):F_{x}=F[[x']]\right]\geq 1-\negl[\lambda]$. We note that $\textsf{A}$ can always output $x$ or $\perp$ so $\textsf{A}$ can always satisfy this requirement.

We say that a QPT algorithm $\textsf{A}$ is a \emph{$\top$-collision-finder} for a $\bot$-function $F$ if it has the following structure for some QPT algorithm $\textsf{A}'$:

\smallskip \noindent\fbox{%
    \parbox{\textwidth}{%
\textbf{Algorithm} $\textsf{A}(x)$:
        \begin{enumerate}
        \item $x'\gets \textsf{A}'(x)$.
        \item Compute $y\gets F(x)$ and $y'\gets F_\top(x')$.
        \item If $y=y'$, then output $x'$. 
        \item Else if $y=\perp$, output $\bot.$ 
        \item Else, output $x.$
        \end{enumerate}
}}

In particular, $\textsf{A}$ searches for non-$\bot$ collisions. We require that $\textsf{A}$ actually checks that the collision occurs to capture the requirement that the collision occurs with non-negligible probability. Note that any algorithm $\mathcal{A}$ that, given $x$, finds a value $x'$ such that $F(x)$ and $ F_\top(x')$ collide with non-negligible probability can be plugged in algorithm $\textsf{A}$ by substituting the call to $\textsf{A}'$ with $\mathcal{A}$ to obtain a $\top$-collision-finder that succeeds with non-negligible probability. 

\begin{definition}
     Let $n=n(\lambda)$ and $m=m(\lambda)$ be polynomials in the security parameter $\lambda\in\mathbb{N}$ and let $\mu(\lambda)= O(\lambda^{-c})$ for some constant $c>0$. Let $F$ be a $(\mu,n,m)$-$\bot$-function. We say $F_\top^{-1}$ has \emph{accessible Shannon entropy} at most $k$ if for every QPT $\top$-collision-finder $\textsf{A}$,
    \begin{align*}
        H(\textsf{A}(X)|X)\leq k.
    \end{align*}
for all sufficiently large $n$, where $X$ is the random variable uniformly distributed on $\{0,1\}^n$. 
\end{definition}

\begin{definition}
    Let $n=n(\lambda)$ and $m=m(\lambda)$ be polynomials in the security parameter $\lambda\in\mathbb{N}$ and let $\mu(\lambda)= O(\lambda^{-c})$ for some constant $c>0$. Let $p\coloneqq p(n)\in [0,1]$ and let $F$ be a $(\mu,n,m)$-$\bot$-function. We say $F_\top^{-1}$ has \emph{$p$-accessible max-entropy at most $k$} if for every QPT $\top$-collision-finder $\textsf{A}$, there exists a family of sets $\{L(x)\}_{x\in \{0,1\}^n}$, where $L(x)\subseteq \{0,1\}^n$, each of size at most $2^k$, such that 
    \begin{align*}
        \Pr[\textsf{A}(X)\in L(X)]\geq 1-p.
    \end{align*}
    for all sufficiently large $n$. If $p$ is negligible in $n$, then we simply say that $F^{-1}_\top$ has accessible max-entropy at most $k$. 
\end{definition}

Note that the accessible max-entropy is bounded by the real max-entropy, since for $x\in \{0,1\}^n$, if $x'\in L(x)$, then $F_\top(x')=F(x)\neq \perp$ with non-negligible probability. This implies that $F[[x]]=F[[x']]\neq \perp$ so $x'\in \textsf{Supp}_{F}(x)$. 

\begin{definition}
Let $n=n(\lambda)$ and $m=m(\lambda)$ be polynomials in the security parameter $\lambda\in\mathbb{N}$ and let $\mu(\lambda)= O(\lambda^{-c})$ for some constant $c>0$. Let $F$ be a $(\mu,n,m)$-$\bot$-function and let $p=p(n)\in [0,1]$. We say $F$ is \emph{$p$-$\top$-collision-resistant on random inputs} if for every QPT $\top$-collision-finder $\adv$,
\begin{align*}
    \Pr_{x\gets \{0,1\}^n}[\textsf{A}(x)\in \{x,\bot\}]\geq 1-p. 
\end{align*}
If $p\in \negl[n]$, then we simply say $F$ is $\top$-collision-resistant on random inputs.
\end{definition}

To build $\botUOWHF$, we need to construct a function family $\mathcal{F}$ that satisfies $\top$-target-collision-resistance: no QPT adversary can output an input $x$, then gets $F\gets \mathcal{F}$, and then find a $x'\neq x$ such that $F[[x]]=F[[x']]\neq \bot$ with non-negligible probability.  

We note that a function which is $\top$-collision-resistant on random inputs can easily be converted into a $\top$-target-collision-resistant family of functions. Specifically, if $F$ is  $\top$-collision-resistant on random inputs, then $\{\tilde{F}_y\}_{y\in \{0,1\}^n}$ defined by $F_y(x)\coloneqq F(y+x)$ is a target collision resistant family of functions.

Therefore, our notions of $\top$-collision-finders and $\top$-collision-resistance are the right notions for building $\botUOWHF$. Furthermore, the notion of $\top$-accessible max-entropy is useful as it allows us to deduce $\top$-collision-resistance as stated in the following lemma. The proof follows directly from the definitions. 

\begin{lemma}
    Let $F$ be a $(\mu,n,m)$-$\bot$-function and let $p=p(n)\in [0,1]$. If ${F}_\top^{-1}$ has $p$-accessible max-entropy 0, then $F$ is $p$-$\top$-collision-resistant on random inputs.  
\end{lemma}

\subsection{\texorpdfstring{$\botOWF$}{bot-OWF}
from \texorpdfstring{$\botPRG$}{bot-PRG}}

%In this section, we build a $\botOWF$ from a $\botPRG$ by adapting the well-known standard approach of building a \textsf{OWF} from a \textsf{PRG}. Our setting is more involved in order to deal with $\bot$ evaluations.
In this section, we show that every length-tripling $\botPRG$ is a $\botOWF$ (up to the necessary padding). The proof is slightly more involved than its classical counterpart due to the $\bot$ evaluations.

\begin{theorem}
    \label{thm:OWF}
 Assuming the existence of a $(\mu,\ell)\text{-}\botPRG$ $G$ where $\ell(\lambda)\geq 3\lambda$, there exists a $(\mu,\ell,\ell)\text{-}\botOWF$.
   \end{theorem}

See \ifnum\iacr=1 Supplement~\ref{app:owf proof}\else \cref{app:owf proof}\fi for the proof. 

\subsection{\texorpdfstring{$\bot$}{bot} Universal One-Way Hash Functions from \texorpdfstring{$\bot$}{bot} \textsf{OWF}s}
\label{sec:botUOWHF construction}

Universal One-Way Hash Functions were first introduced by Naor and Young \cite{NY89}, who showed how to use these functions to build digital signatures. Rompel was the first to construct \textsf{UOWHF}s from any arbitrary \textsf{OWF} \cite{R90}, thus basing digital signatures on \textsf{OWF}s. Later,  Haitner, Holenstein, Reingold, Vadhan, and Wee \cite{HHR+10} provided a shorter and simpler construction of \textsf{UOWHF}s based on the notion of accessible entropy. Their proof consists of two parts: $(1)$ Constructing a function with sufficient inaccessible Shannon entropy from any \textsf{OWF}. $(2)$ Building a \textsf{UOWHF} from any function with sufficient inaccessible entropy. 

We adapt their proof to achieve $\bot$-Universal One-Way Hash Functions from $\botOWF$s. Our proof follows the same route, but adapted to the notions of entropy for $\bot$-functions. This requires many steps to be modified and some bounds to be loosened in order for the construction to go through, especially in the first part of the construction. 

\subsubsection{Inaccessible Entropy from \texorpdfstring{$\bot$}{bot} \textsf{OWF}s}
Let $f$ be a $(\mu,n,n)\text{-}\botOWF$. Define the QPT algorithm $F$ with domain $\{0,1\}^n\times [n]$ and co-domain $\{0,1\}^n\cup \{\bot\}$ as follows: $F(x,i)$ runs $f(x)$ and outputs the first $i$ elements. In the case $f(x)$ returns $\bot$, then $F(x,i)$ outputs $\bot$. Note that $F$ is also a $\bot$-function with the same pseudodeterminism parameter $\mu$ as $f$. We will show a lower bound on the accessible Shannon entropy of $F^{-1}_\top$. In the next section, we amplify this entropy to build a $\botUOWHF$. 

\begin{theorem}
\label{thm:inaccessible entropy}
 $F^{-1}_\top$ has accessible Shannon entropy at most $H(F_\top^{-1}(F(Z))|F(Z))-1/(2^9\cdot n^4\cdot \log^2(n)),$ where $Z=(X,I)$ is uniformly distributed over $\{0,1\}^n\times [n]$.
\end{theorem}

\begin{proof}
Assume for contradiction that there exists a $\top$-collision-finder $\textsf{A}$ for $F$ such that: 
\begin{align*}
    H(F_\top^{-1}(F(Z))|F(Z))-H(\textsf{A}(Z)|Z)\leq \epsilon= 1/(2^9\cdot n^4\cdot \log^2(n)) .
\end{align*}
Let $\textsf{A}'$ be the algorithm which on input $(x,i)$, outputs the first component of $\textsf{A}(x,i)$. Then, 
\begin{align*}
    H(F_\top^{-1}(F(Z))_1|F(Z))-H(\textsf{A}'(Z)|Z)\leq \epsilon ,
\end{align*}
where $F_\top^{-1}(F(Z))_1$ denotes the first component of $F_\top^{-1}(F(Z))$. We will use $\textsf{A}'$ to construct an inverter for $f$. This is done in two steps: $(1)$ Build an inverter assuming access to an inefficient oracle $\textsf{Sam}_{\text{Ideal}}$. $(2)$ Approximate $\textsf{Sam}_{\text{Ideal}}$ with an efficient algorithm using $\textsf{A}'$. 

 By the same arguments used in \cite{HHR+10}, it can be shown that:
\begin{align*}
    \| (Z,F^{-1}_\top(F(Z))_1)-(Z,\textsf{A}'(Z))\|\leq \sqrt{\epsilon} .
\end{align*}
Furthermore, since $Z$ is distributed uniformly on $(X,I)$ and $\Pr{\left[ X\in \mathcal{G}\right]}\geq 1-\mu$ ($\mathcal{G}$ is the ``deterministic'' set, see Definition \ref{def:botOWF}), then the triangle inequality gives
\begin{align*}
    \| (G,F^{-1}_\top(F(G,I))_1)-(G,\textsf{A}'(G,I))\|\leq \frac{\sqrt{\epsilon}}{1-\mu} .
\end{align*}
where $G$ is the random variable uniformly distributed over $\mathcal{G}$. 

\smallskip \noindent\fbox{%
    \parbox{\textwidth}{%
\textbf{Algorithm} $\textsf{Sam}_{\text{Ideal}}$:
\begin{itemize}
    \item \textbf{Input:} $x\in \{0,1\}^n$, $i\in [n]$, and $b\in \{0,1\}$.
    \item \textbf{Output:} 
    \begin{enumerate}
        \item Sample $y\gets f(x)$. If $y=\bot$, then output $\bot$.
        \item Otherwise, return $x'\gets F_\top^{-1}(y_{1,\ldots,i-1}\|b)_1$.
    \end{enumerate}
\end{itemize}
}}
\smallskip 

We now present an algorithm $\textsf{Inv}_{\text{Ideal}}$ that inverts $f$ with non-negligible probability using oracle access to $\textsf{Sam}_{\text{Ideal}}$. 

\smallskip \noindent\fbox{%
    \parbox{\textwidth}{%
\textbf{Algorithm} $\textsf{Inv}_{\text{Ideal}}$:
\begin{itemize}
    \item \textbf{Input:} $y\in \{0,1\}^n\cup \{\bot\}$.
        \item \textbf{Oracle:} $\textsf{Sam}_{\text{Ideal}}$
    \item \textbf{Output:} 
    \begin{itemize}
        \item If $y=\bot$, then output $\bot$.
        \item Otherwise, for each $i\in [n]:$
    \begin{itemize}
        \item Sample $x^i\gets \textsf{Sam}_{\text{Ideal}}(x^{i-1},i,y_i)$ (where $x^0$ is chosen randomly). 
    \end{itemize}
    \item Output $x^{n}$
    \end{itemize}
\end{itemize}
}}
    \smallskip 

\begin{claim}
    $\textsf{Inv}_{\text{Ideal}}$ inverts $f$ with non-negligible probability if $\mu<\frac{1}{n}$.
\end{claim}

\begin{proof}
We need to show that the probability $\Pr_{x\gets X}\left[f_\top(\textsf{Inv}_{\text{Ideal}}(y))=y:y\gets f(x) \right]$ is non-negligible. 

Let $x\gets X$ and $y\gets f(x)$. Consider the algorithm $\textsf{Inv}_{\text{Ideal}}(y)$. If $x^1\in \mathcal{G}$, then there exists $y^1$ satisfying $y^1_1=y_1$ %and a negligible function 
such that:
 \begin{align}
    \Pr{\left[f(x^1)=y^1\right]} \geq 1 - \negl[\lambda] .
\end{align} 

Applying this argument for all $i\in [n]$, we get that if $x^1,x^2,\ldots, x^n\in \mathcal{G}$, then
 \begin{align}
                \Pr\left[f(x^n)=y \right] \geq 1 - \negl[\lambda] .
\end{align} 
So in this case, $\textsf{Inv}_{\text{Ideal}}$ inverts $y$ with high probability. Therefore, it is sufficient to show that $x^i\in \mathcal{G}$ for all $i\in [n]$ with non-negligible probability to deduce that $\textsf{Inv}_{\text{Ideal}}$ is an inverter of $f$. 

For any $i\in [n]$, let $\textsf{E}_i$ be the event that $x^i\in \mathcal{G}$ and let $\textsf{D}^i$ be the event that when $\textsf{Sam}_{\text{Ideal}}$ evaluates $f(x^i)$, it gets $f[[x^i]]$. Similarly, let $\textsf{E}^*$ be the event that $x\in \mathcal{G}$ and $\textsf{D}^*$ be the event that when $\textsf{Inv}_{\text{Ideal}}$ computes $f(x)$, it gets $f[[x]]$.
 Note that for any $i\in [n]$, $\Pr\left[ F^{-1}_\top(f[[x]],i)\in \mathcal{G}:x\gets \mathcal{G} \right] \geq (1-\mu)$ by the pseudodeterminism property of $f$. We want to bound the following probability:
\begin{align*}
    \Pr & \left[\textsf{E}^*\textsf{E}_1\textsf{E}_2\ldots \textsf{E}_n  \right]=\\
    &\Pr\left[\textsf{E}_n|\textsf{E}^*\textsf{E}_1\ldots \textsf{E}_{n-1} \right]\cdot \Pr\left[\textsf{E}_{n-1}|\textsf{E}^*\textsf{E}_1\ldots \textsf{E}_{n-2} \right]\cdot \ldots \cdot \Pr\left[\textsf{E}_1|\textsf{E}^*\right]\cdot \Pr\left[\textsf{E}^*\right].
\end{align*}
Given that $x$ is sampled uniformly at random, we have $\Pr\left[\textsf{E}^*\right]\geq (1-\mu)$. Next, 
\begin{align*}
    \Pr\left[\textsf{E}_1|\textsf{E}^*\right] &= \Pr\left[\textsf{E}_1|\textsf{E}^*\textsf{D}^*\right]\Pr\left[\textsf{D}^*| \textsf{E}^*\right]+\Pr\left[\textsf{E}_1|\textsf{E}^*\overline{D}^*\right]\Pr\left[\overline{D}^*| \textsf{E}^*\right] \\ 
    &\geq \Pr\left[\textsf{E}_1|\textsf{E}^*\textsf{D}^*\right] (1-\negl[\lambda]) +0 \\
    &\geq (1-\mu)(1-\negl[\lambda]).
\end{align*}
This is because $\Pr\left[\textsf{D}^*| \textsf{E}^*\right]\geq 1-\negl[\lambda]$ for some negligible function by the pseudodeterminism property of $f$ and $\Pr\left[\textsf{E}_1|\textsf{E}^*\textsf{D}^*\right]=\Pr\left[F^{-1}_\top(f[[x]],1)\in \mathcal{G}:x\gets \mathcal{G} \right]\geq (1-\mu)$. 
Continuing in this fashion, we have that
\begin{align*}
    \Pr\left[\textsf{E}^*\textsf{E}_1\textsf{E}_2\ldots \textsf{E}_n  \right]\geq (1-\negl[\lambda])^n(1-\mu)^n > \frac{3(1-\mu)^n}{4}.
\end{align*}
%\begin{align*}
 %    \Pr\left[G_2|G_1G^*\right]\geq \Pr\left[G_2|G^*G^1\textsf{D}^1\right]\Pr\left[\textsf{D}^1| G^*G^1\right] \geq (1-\negl)\Pr\left[G_2|G^*G^1\textsf{D}^1\right]\geq (1-\negl)(1-\mu).
%\end{align*}
%Let $x\in \mathcal{G}$ and consider the probability $ \alpha(i,b,x)   \Pr\left[x^i\gets \textsf{Sam}_{\text{Ideal}}(x^{i-1},i,b): \forall j\in [i], x^j\in \mathcal{G} \right]$. Since $y$ is sampled randomly from the image of $f$, it is not difficult to see that the probability that $x^i\in \mathcal{G}$ for all $i\in [n]$ is at least $(1-\mu)^n$, which is polynomial if $\mu>n$. 
Note that $\frac{3(1-\mu)^n}{4}$ is non-negligible if $\mu<1/n$. Hence, $\textsf{Inv}_{\text{Ideal}}$ is an inverter of $f$.

\qed
\end{proof}

We will now approximate $\textsf{Sam}_{\text{Ideal}}$ with the following algorithm which is efficient given our assumption on the existence of an efficient $\top$-collision-finder $\textsf{A}'$ for $F$.

\smallskip \noindent\fbox{%
    \parbox{\textwidth}{%
\textbf{Algorithm} $\textsf{Sam}_{\text{Approx}}$:
\begin{itemize}
    \item \textbf{Input:} $x\in \{0,1\}^n$, $i\in [n]$, and $b\in \{0,1\}$.
\item \textbf{Oracle:} $\textsf{A}'$.
\item \textbf{Output:} 
\begin{itemize}
    \item Sample $y\gets f(x)$. If $y=\bot$, then output $\bot$.
    \item Otherwise, repeat $\lceil 20n\cdot \log(n)\rceil$ times:
    \begin{enumerate}
    \item $x'\gets \textsf{A}'(x,i)$.
    \item Compute $y'\gets f_\top(x')$.
    \item If $y_{1,\ldots,i-1}=y'_{1,\ldots,i-1}$ and $ y'_{i}=b$, then output $x'$. 
\end{enumerate}
\item Otherwise, abort. 
\end{itemize}

\end{itemize}
}}
    \smallskip 

\begin{claim}
    Let $i\in [n]$ and let $\delta_i\coloneqq \| (G,F^{-1}_\top(F(G,i))_1)-(G,\textsf{A}'(G,i))\|$, then
    \begin{align*}
        \| (G,\textsf{Sam}_{\text{Ideal}}(G,i,f(G)_{i}))-(G,\textsf{Sam}_{\text{Approx}}(G,i,f(G)_{i}))\|\leq \frac{1}{2n}+20n\cdot \log(n)\cdot \delta_i.
    \end{align*}
\end{claim}

\begin{proof}
The difference between the two distributions comes from two sources. Firstly, $\textsf{Sam}_{\text{Approx}}$ samples a collision using $\textsf{A}'$ whereas $\textsf{Sam}_{\text{Ideal}}$ samples a collision uniformly at random. In particular, each query to $\textsf{A}'$ contributes an error of $\delta_i$, and there are $20n\cdot \log(n)$ queries, so this accounts for $20n\cdot \log(n)\cdot \delta_i$ in the statistical distance. Secondly, the algorithm $\textsf{Sam}_{\text{Approx}}$, may abort without returning a valid collision. In more details, $\textsf{Sam}_{\text{Approx}}(x,i,b)$ aborts if for each iteration, the output $x'\gets \textsf{A}'(x,i)$ and the corresponding image $y'\gets f_\top(x')$ do not satisfy $y_{1,\ldots,i-1}=y'_{1,\ldots,i-1}$ and $ y'_{i}=b$. 

To understand the contribution of each source of error, we analyze both cases separately. Firstly, assume $\delta_i=0$. For $x\in \{0,1\}^n$, define $\alpha(x,i)\coloneqq 
\Pr[(f[[x]]_{i}=f[[x']]_{i}) \wedge (x'\in \mathcal{G})|f[[x]]_{1,\ldots, i-1}=f[[x']]_{1,\ldots, i-1}\neq \perp]$. 
Then,  the following claim holds for any $\beta>0$.
\begin{claim}
\begin{align*}
    \Pr[\alpha(G,i)<(1-\mu)\beta] := \Pr_{x\gets \mathcal{G}}{\left[\alpha(x,i)
    < (1-\mu)\beta\right]} < \beta .
\end{align*}    
\end{claim}

\begin{proof}
Assume for contradiction that $\Pr[\alpha(G,i)<(1-\mu)\beta] \geq \beta$. Let $i\in [n]$ and $y^i\in \{0,1\}^{i-1}$. Consider the following sets,
\begin{align*}
    A(y^i,b,i)&\coloneqq \{x\in \mathcal{G}:(f[[x]]_{1,\ldots, i-1}=y^i)\wedge (f[[x]]_i=b)\},\\
    E(y^i,i)&\coloneqq \{x\in \{0,1\}^n:f[[x]]_{1,\ldots, i-1}=y^i\},\\
    S(i,\beta)&\coloneqq \{x\in \mathcal{G}: |A(f[[x]]_{1,\ldots, i-1},f[[x]]_i,i)|\leq (1-\mu)\beta \cdot |E(f[[x]]_{1,\ldots,i-1},i)|\},\\
    \overline{S}(i,\beta)&\coloneqq \{x\in \mathcal{G}:x\notin S(i,\beta)\} .
\end{align*} 
Note that if $A(y^i,b,i)\subseteq S(i,\beta)$, then $|A(y^i,b,i)|\leq (1-\mu)\beta \cdot |E(y^i,i)|$. Furthermore, $|S(i,\beta)|\leq (1-\mu)\beta \cdot 2^n$. 

Our assumption can be rephrased as $|S(i,\beta)|\geq \beta (|S(i,\beta)|+|\overline{S}(i,\beta)|)$. As a result,
\begin{align*}
    |\mathcal{G}|=|S(i,\beta)|+|\overline{S}(i,\beta)|\leq \frac{1}{\beta}|S(i,\beta)|\leq \frac{1}{\beta} \cdot \beta (1-\mu)\cdot 2^n = (1-\mu)\cdot 2^n .
\end{align*}
This contradicts the pseudodeterminism condition of $f$, which states that $|\mathcal{G}|>(1-\mu)\cdot 2^n$ so we must have that $\Pr[\alpha(G,i)<(1-\mu)\beta] < \beta$.
\qed 
\end{proof}

Also, note that $\textsf{Sam}_{\text{Approx}}$ aborts with probability $(1-\frac{1}{4n}(1-\mu)-\negl[\lambda])^{20n\cdot \log(n)}<\frac{1}{4n}$ when $\alpha(x,i)\geq \frac{1}{4n}(1-\mu)$ and $\mu<1/n$. If $\textsf{Sam}_{\text{Approx}}$ does not abort then, it returns the same distribution as $\textsf{Sam}_{\text{Ideal}}$. Hence, we have:
\begin{align*}
\| &(G,\textsf{Sam}_{\text{Ideal}}(G,i,f(G)_{i}))-(G,\textsf{Sam}_{\text{Approx}}(G,i,f(G)_{i}))\| \\         &\leq \Pr[\alpha(X,i)< \frac{1}{4n}(1-\mu)]+\Pr[\textsf{Sam}_{\text{Approx}} \text{ aborts }|\alpha(G,i)\geq  \frac{1}{4n}(1-\mu)]\\
         &\leq \frac{1}{4n}+\frac{1}{4n}\\
         &\leq \frac{1}{2n}. 
\end{align*}
The statistical distance between the distribution of $\textsf{Sam}_{\text{Approx}}(G,i,f(G)_{i})$ when $\delta_i=0$ and the general case is at most the maximal number of calls made to $\textsf{A}'$ times $\| (G,F^{-1}_\top(F(G,i))_1)-(G,\textsf{A}'(G,i))\|$ which is bounded by $20n\cdot \log(n)\cdot \delta_i$. Therefore, 
\begin{equation*}
\begin{split}
    \| (G,\textsf{Sam}_{\text{Ideal}}(G,i,f(G)_{i}))-(G,\textsf{Sam}_{\text{Approx}}(G,i,f(G)_{i}))\| \\
    \leq \frac{1}{2n} + 20n\cdot \log(n)\cdot \delta_i .
\end{split}
\end{equation*}
\qed
    \end{proof}

The rest of the proof uses $\textsf{Sam}_{\text{Approx}}$ in place of $\textsf{Sam}_{\text{Ideal}}$ in $\textsf{Inv}_{\text{Ideal}}$ to invert $f$ as follows. For any $i\in [n]$, let $\textsf{Inv}^i_{\text{Approx}}$ be the algorithm which uses $\textsf{Sam}_{\text{Ideal}}$ in the first $i-1$ iterations and uses $\textsf{Sam}_{\text{Approx}}$ in the rest of the iterations. We want to show that $\textsf{Inv}^1_{\text{Approx}}$ succeeds with non-negligible probability. 
\begin{align*}
    &\Pr\left[\textsf{Inv}^1_{\text{Approx}}(f(G))\in f_\top^{-1}(f(G))\right]\\
    &\geq \Pr\left[\textsf{Inv}^{n+1}_{\text{Approx}}(f(G))\in f_\top^{-1}(f(G))\right] -\sum_{i\in [n]}\| \textsf{Inv}^{i+1}_\delta(f(G))-\textsf{Inv}^i_\delta(f(G))\| \\
    &= \Pr\left[\textsf{Inv}_{\text{Ideal}}(f(G))\in f_\top^{-1}(f(G))\right] -\sum_{i\in [n]}\| \textsf{Inv}^{i+1}_\delta(f(G))-\textsf{Inv}^i_\delta(f(G))\| \\
    &\geq 3(1-\mu)^n/4-\sum_{i\in [n]}\frac{1}{2n} + 20n\cdot \log(n)\cdot \delta_i\\
    &\geq 3(1-\mu)^n/4-\frac{1}{2}-20n^2\cdot \log(n)\cdot \delta.\\
\end{align*}
By our assumption $\delta= \frac{\sqrt{\epsilon}}{(1-\mu)}$, which with some elementary maths 
yields that $\Pr\left[\textsf{Inv}_\delta(f(G))\in f_\top^{-1}(f(G))\right]$ is non-negligible for large enough $n$ and $\mu<1/n^2$. 
\qed
\end{proof}

\subsubsection{\texorpdfstring{$\bot$}{bot} \textsf{UOWHF} from Inaccessible Entropy}

The next part of the $\botUOWHF$ construction involves amplifying the inaccessible Shannon entropy established in Theorem \ref{thm:inaccessible entropy} in order to build a function that compresses the input and achieves an accessible max-entropy of 0 following the same steps as in \cite{HHR+10}. The steps involve increasing the entropy gap between the input and output through repetition, reducing the entropy of the input by outputting a hash of the input, and reducing the output length by hashing the output. 

These steps can be applied in the same way to a $\botOWF$s to obtain a function that compresses the input and is $\bot$-target-collision-resistant. We note that in our case, our notion of accessible max-entropy differs from the corresponding notion in \cite{HHR+10} as ours is adapted to $\bot$-functions, however, this does not affect this part of the proof. Therefore, we do not include the proof of this part and instead, refer the reader to Sec.~4.2 in \cite{HHR+10}. 
\begin{theorem}
    \label{thm:uowhf}
    Let $f$ be a $(\mu,n,n)\text{-}\botOWF$\footnote{Note that a $\bot$-function is actually an algorithm which is $\bot$-deterministic.} with $\mu<1/n^{30}$. Then, there exists a $(\mu',k,m,\ell)\text{-}\botUOWHF$, with pseudodeterminism error $\mu'=O(n^{28}\mu)$, key length $k=\tilde{O}(n^{34})$, compression $m=n^2\log(n)$, and output length $\ell=\tilde{O}(n^{36})$.  
\end{theorem}

\fi

\section{Digital Signatures}
\label{sec:Sig}
In this section, we construct digital signatures from $\SPRS$. This is done in several steps. In \cref{sec:sig def}, we adapt strong unforgeability for deterministic one-time signatures to the non-deterministic setting yielding what we call \emph{one-message signatures (OMS)}. 
In \cref{sec:sig 1}, we build length-restricted OMSs using our construction for $\botUOWHF$ (see \Cref{def:botUOWHF}). In \cref{sec:sig 1.5}, we enable signing long messages using $\botUOWHF$s. 
\ifnum\iacr=1
We build $\botUOWHF$s from $\botPRG$s in Supplement~\ref{sec:botUOWHF construction}. This requires multiple adaptions from the classical construction. First, we introduce the notion of $\botOWF$s and $\botUOWHF$. We further introduce notions of entropy that capture the size of collision sets under any $\bot$-function. Next, we build $\botOWF$ from a $\botPRG$. With these tools at hand, we are able to adapt the proof of \cite{HHR+10} (which shows how to construct \textsf{UOWHF}s from \textsf{OWF}s via inaccessible entropy) to our setting achieving $\botUOWHF$s from $\botPRG$s.  
\fi
Then, in \cref{sec:sig 2}, we enable signing many messages but where the signer needs to store previous signatures. As a result, we get a stateful signature scheme. Finally, in \cref{sec:sig 3}, we use the construction for $\botPRF$s to enable signing without storing previous signatures, resulting in a full-fledged stateless digital signature scheme.

\subsection{Definition (One-Message) Signatures}
\label{sec:sig def}
One-time SUF requires that an adversary which receives a signature of a message of its choice, cannot forge a new valid message-signature pair. However, in our case, the signing algorithm may be non-deterministic. This causes issues since the scheme may be insecure if the same message is signed multiple times, which occurs in our application of OMSs to many-time signatures. Hence, we need to guarantee SUF against an adversary that receives multiple signatures of a single message of its choice. We call this experiment one-message SUF.

\smallskip \noindent\fbox{%
    \parbox{\textwidth}{%
\textbf{Experiment} $\textsf{Sign}^{\text{OM-SUF}}_{\Pi, \adv}({\lambda})$:
\begin{enumerate}
    \item Sample $(\sk,\vk)\leftarrow \textsf{Gen}(1^\lambda)$.
    \item $\adv$ is given $\vk$ and classical access to the signing oracle $\textsf{Sign}(\sk, \cdot)$.
    \item Let $\mathcal{Q}$ denote the set of messages-signature pairs returned by the signing oracle.
    \item $\adv$ outputs a forgery $(m^*,\sigma^*).$
   \item The output of the experiment is $1$ if:
   \begin{enumerate}
       \item All of the queries were on the same message.
       \item $ (m^*,\sigma^*)\notin \mathcal{Q}$.
       \item $\textsf{Verify}(\vk, m^*, \sigma^*)=1$.
   \end{enumerate}
   \item Otherwise, the output is 0. 
\end{enumerate}}}
\smallskip

\begin{definition}
    A signature scheme $\Pi$ is \emph{one-message strongly unforgeable ({OM-SUF})}, if for any QPT $\adv$,
    \begin{align*} 
        \Pr [\textsf{Sign}^{\textit{OM-SUF}}_{\Pi, \adv}({\lambda})=1]\leq \negl[\lambda].
    \end{align*}
\end{definition}
\ifnum\iacr=1
    Next, we define the following tool that we use in our construction of OM-SUF.
    
    \begin{definition}[$\botUOWHF$]\label{def:botUOWHF}
       Let $n=n(\lambda)$ and $m=m(\lambda)$ be polynomials in the security parameter $\lambda\in\mathbb{N}$. A family of QPT algorithms\footnote{We emphasize that even though we call $\botUOWHF$s as $\bot$-functions they are actually algorithms. The term \emph{function} is used to emphasize their similarity to classical notions upto $\bot$-determinism.}
       $\mathcal{H}_\lambda \coloneqq \{{H}_\lambda(k,\cdot)\}_{k\in \{0,1\}^\ell} $ that on input in $ \{0,1\}^n$ output an element in $\{0,1\}^m \cup \{\bot\}$ %mapping $F_\lambda:\mathcal{X}_{\lambda}\rightarrow \{0,1\}^m$ 
        is a \emph{$(\mu,\ell, d=n/m,m)$-pseudodeterministic universal one-way hash family with recognizable abort ($\botUOWHF$)} if the following conditions hold: 
        \begin{itemize}
        \item \textbf{Compression:} $d(\lambda)>1$ for all $\lambda\in \mathbb{N}$. 
            \item \textbf{Pseudodeterminism:} There exist a constant $c>0$ such that $\mu(\lambda)= O(\lambda^{-c})$ and for sufficiently large $\lambda\in \mathbb{N}$ and $F\in \mathcal{H}_\lambda$, there exists a set $\mathcal{G} \subseteq \{0,1\}^n$ such that the following holds:
            \begin{enumerate}
                \item \[\Pr_{x\gets \{0,1\}^n}\left[x\in\mathcal{G} \right] \geq 1-\mu(\lambda).\] 
                \item For every $x\in \mathcal{G}$ there exists a non-$\bot$ value $y\in\{0,1\}^m$ such that: 
                \begin{align}
                    \Pr\left[F(x)=y \right] \geq 1 - \negl[\lambda].
                \end{align} 
              
                \item For every $x\in \{0,1\}^n$, there exists a non-$\bot$ value $y\in\{0,1\}^m$ such that: 
                \begin{align}
                    \Pr\left[F(x)\in \{y,\bot\} \right] \geq 1 - \negl[\lambda].
                \end{align} 
            \end{enumerate}
            \item \textbf{Security:} Two conditions need to be satisfied:
            \begin{enumerate}
                \item \textbf{One-wayness:} For any QPT algorithm $\adv$ and $F\in \mathcal{H}_\lambda$,
            \begin{align*}
                \Pr_{x\gets \{0,1\}^n}[F(\adv(F(x)))=F_\top(x)]\leq \negl[\lambda].
            \end{align*}
            
            \item \textbf{Collision-resistance:} For every QPT algorithm $\adv$, if $\adv$ outputs $x\in \{0,1\}^n$, then is given $F\gets \mathcal{H}_\lambda$, and then outputs $x'\neq x$, then
            \begin{align*}
                \Pr\left[F(x)=F_\top(x')\right]\leq \negl[\lambda].
            \end{align*}
            \end{enumerate}
        \end{itemize}
    \end{definition}
    
    %We emphasize that while the variants of pseudodeterministic QPT algorithms introduced in this section are called ``functions'', they are not technically functions due to the non-determinism. This terminology is used to highlight their similarity to the classical notions. 
\fi
\subsection{(Length-Restricted) One-Message Signatures}

\label{sec:sig 1}
We now provide a construction for a signature scheme on length-restricted messages satisfying OM-SUF using $\botUOWHF$ (see \Cref{def:botUOWHF}). In this scheme, the message length is smaller than the length of the verification key, but we enable signing longer messages in the next section. In our scheme, there is a non-negligible probability that the verification of a valid signature fails, however, we will later achieve statistical correctness in the final signature scheme.

\begin{construct}
    \label{con:one-time sig}
    Let $\lambda\in \mathbb{N}$ be the security parameter. Let $q$ be polynomial in $\lambda$. Let $\mu= O(1/q^2)$. Let $\mathcal{H}_{\lambda}=\{H_{\lambda}(k,\cdot)\}_{k\in \{0,1\}^\ell}$ be a $(\mu,\ell,2,\lambda)$-$\botUOWHF$ family for some polynomial $\ell$ in $\lambda$. The construction for a OMS scheme on $q$-bit messages is as follows:
    \begin{itemize}
        \item $\textsf{KeyGen}(1^\lambda):$ 
        \begin{enumerate}
            \item For each $j\in [q]$, and $b\in \{0,1\}:$
            \begin{enumerate}
            \item Sample $k_{j,b}\gets \{0,1\}^\ell$.
                \item Sample $x_{j,b}\gets \{0,1\}^{2\lambda}$.
                \item Sample $y_{j,b}\gets H(k_{j,b},x_{j,b})$.
            \end{enumerate}
            \item Set $\sk\coloneqq (x_{j,b})_{j\in [q],b\in \{0,1\}}$ and $\vk\coloneqq (k_{j,b}, y_{j,b})_{j\in [q],b\in \{0,1\}}$.
            \item Output $(\sk,\vk)$.
        \end{enumerate}   

        \item $\textsf{Sign}(\sk,m):$ 
        Output $\sigma \coloneqq (x_{j,m_j})_{j\in [q]}$. 
        
        \item $\textsf{Verify}(\vk,m,\sigma):$ 
        \begin{enumerate}
        \item Interpret $\sigma$ as $(\tilde{x}_{j})_{j\in [q]}$.
        \item For each $j\in [q]$, compute $\tilde{y}_{j}\gets H(k_{j,m_j},\tilde{x}_{j})$. 
            \item If $\tilde{y}_{j}=y_{j,m_j}\neq \bot$ for all $j\in [q]$, then output $\top$.
            \item Otherwise, output $\perp $.
        \end{enumerate} 
    \end{itemize}
\end{construct}

\begin{lemma}[Correctness]\label{lem:ots corr}
    \cref{con:one-time sig} is $(1-\mu)^{4q}$-correct OMS scheme on $q$-bit messages assuming the existence of $\botUOWHF$s. 
\end{lemma}

\begin{proof}
Fix $m\in \{0,1\}^q$ and sample $(\sk,\vk)\gets \textsf{KeyGen}(1^\lambda)$, where $\sk\coloneqq (x_{j,b})_{j\in [q],b\in \{0,1\}}$ and $\vk\coloneqq (k_{j,b}, y_{j,b})_{j\in [q],b\in \{0,1\}}$. Let $\mathcal{G}_{j,b}$ denote the good set of inputs for $H(k_{j,b},\cdot)$. 

By the pseudodeterminism of $\botUOWHF$s, the probability that $x_{j,b}\in \mathcal{G}_{j,b}$ for all $j\in [q]$ and $b\in \{0,1\}$ is at least $p\coloneqq (1-\mu)^{2q}$. In this case, there is a negligible function $\negl[\lambda]$ such that for any $j\in [q]$ and $b\in \{0,1\}$, there exists values $(\hat{y}_{j,b})_{j,b}$ such that $\Pr[H_{j,b}(k_{j,b},x_{j,b})=\hat{y}_{j,b}]\geq 1-\negl[\lambda]$. In which case, the verification algorithm will accept with probability at least $(1-\negl[\lambda])^{2q}$. 
Hence, the scheme is correct with probability at least $p(1- \negl[\lambda])^{2q}\geq (1-\mu)^{4q}$.
\qed
\end{proof}

\begin{lemma}[Security]
\label{lem:ots}
\cref{con:one-time sig} is OM-SUF signature scheme on $q$-bit messages assuming the existence of $\botUOWHF$s. 
\end{lemma}

The proof is the same as in the classical case given in Section 6.5.2 in \cite{Gol04}. Notice that in this scheme, \ifnum\iacr=0 getting \fi access to multiple signatures of the same message does not affect the experiment since the signing algorithm is deterministic. 

\subsection{One-Message Signatures}
\label{sec:sig 1.5}

We now show how to transform a \emph{length-restricted} OM-SUF signature scheme into a OM-SUF signature scheme based on the approach in the classical case given in Section 6.5.2 of \cite{Gol04}. This is important since we will need to sign messages that are double the length of the verification key in order to build many-message signatures in the following sections. 

\begin{construct}
\label{con:one-time sig 2}
    Let $\lambda\in \mathbb{N}$ be the security parameter. Let $q$ be polynomial in $\lambda$ and $\mu= O(1/q^2)$. Let $\mathcal{H}_{\lambda}=\{H_{\lambda}(k,\cdot)\}_{k\in \{0,1\}^\ell}$ be a $(\mu,\ell,q/\lambda,\lambda)$-$\botUOWHF$ family for some polynomial $\ell$ in $\lambda$ such that $\ell+\lambda < q/2$. Let $(\overline{\textsf{KeyGen}}, \overline{\textsf{Sign}},\overline{\textsf{Verify}})$ be the algorithms of Construction \ref{con:one-time sig} on $(\ell +\lambda)$-bit messages. The construction for a OMS scheme on $q$-bit messages is as follows:
    \begin{itemize}
        \item $\textsf{KeyGen}(1^\lambda):$ 
        \begin{enumerate}
            \item Sample $(\overline{\sk},\overline{\vk})\gets \overline{\textsf{KeyGen}}(1^\lambda)$.
            \item Sample $k\gets \{0,1\}^\ell$. 
            \item Set $\sk\coloneqq (\overline{\sk},k)$ and $\vk\coloneqq \overline{\vk}$.
            \item Output $(\sk,\vk)$. 
        \end{enumerate}

        \item $\textsf{Sign}(\sk,m):$ 
        \begin{enumerate}
            \item Compute $y\gets H(k,m)$. If $y=\bot$, output $\bot$.
            \item Otherwise, compute $ s \gets \overline{\textsf{Sign}}(\overline{\sk}, k\|y)$.
            \item Output $\sigma \coloneqq (k, s)$.
        \end{enumerate}

        \item $\textsf{Verify}(\vk,m,\sigma):$ 
        \begin{enumerate}
        \item If $\sigma=\bot$, output $\bot$.
        \item Otherwise, interpret $\sigma$ as $(k, s)$.
        \item Compute $y\gets H(k,m)$. If $y=\bot$, output $\bot$.
        \item Otherwise, output $\overline{\textsf{Verify}}(\vk, k\|y, s)$.
        \end{enumerate} 
    \end{itemize}
\end{construct}

\begin{lemma}[Correctness]
\label{lem:corr ots}
    \cref{con:one-time sig 2} is $(1-\mu)^{2q+3}$-correct OMS scheme assuming the existence of $\botUOWHF$s. 
\end{lemma}

\begin{proof}
Fix $m\in \{0,1\}^q$ and sample $(\sk,\vk)\gets \textsf{KeyGen}(1^\lambda)$ as described in the construction. Let $\mathcal{G}$ denote the good set of inputs for $H(k,\cdot)$. The probability that $m\in \mathcal{G}$ is at least $1-\mu$. 

If $m\in \mathcal{G}$, there is a negligible function $\negl[\lambda]$ and value $\hat{y}$ such that $\Pr[H(k,m)=\hat{y}]\geq 1-\negl[\lambda]$. In this case, it is clear that the verification algorithm accepts with high probability if and only if $\overline{\textsf{Verify}}(\vk, k\|\hat{y}, \overline{\textsf{Sign}}(\sk, k\|\hat{y}))$ accepts. By Lemma \ref{lem:ots corr}, this occurs with at least $(1-\mu)^{4(\lambda+\ell)}$ probability. 

All in all, the scheme is correct with probability at least $(1-\mu)^{4(\lambda+\ell)}(1-\mu)(1- \negl[\lambda])^{2}\geq (1-\mu)^{2q+3}$.
\qed
\end{proof}

\begin{lemma}[Security]
\label{lem:ots 2}
\cref{con:one-time sig 2} is OM-SUF signature scheme assuming the existence of $\botUOWHF$s. 
\end{lemma}

\begin{proof}
Assume there exists a QPT algorithm $\adv$ such that $\Pr [\textsf{Sign}^{\textit{OM-SUF}}_{\Pi, \adv}({\lambda})=1]$ is non-negligible. The experiment samples a key pair $(\sk,\vk)\gets \textsf{KeyGen}(1^\lambda)$. Let $m$ be the message queried by $\adv$. We assume that $\adv$ queries the same message as, otherwise, the experiment outcome is 0. 
There exists a negligible function $\negl[\lambda]$ and value $y$ such that $\Pr[H(k,m)\in \{\bot, y\}]\geq 1-\negl[\lambda]$. Let $ s = \overline{\textsf{Sign}}(\overline{\sk}, k\|y)$. With high probability, the signing oracle will respond to all queries submitted by $\adv$ with either $\bot$ or $(k,s)$. 

Let $(\tilde{m}, (\tilde{k},\tilde{s}))$ be the forgery submitted at the end of the experiment and let $\tilde{y}\gets H(\tilde{k},\tilde{m})$ be the evaluation computed by the experiment during verification. 
If the output of the experiment is 1, then there are 2 cases to consider: 
\begin{itemize}
\item $(\tilde{k},\tilde{y},\tilde{s})\neq (k,y,s):$ In this case, $\adv$ has produced a new valid message-signature pair under the the OMS scheme, contradicting SUF (Lemma \ref{lem:ots}). 
    \item $(\tilde{k},\tilde{y},\tilde{s})= (k,y,s)$ and $m\neq \tilde{m}:$ In this case, the pair $(m,\tilde{m})$ act as a collision on $H(k,\cdot)$. Given that $\adv$ must choose $m$ prior to knowing $k$, this implies that $\adv$ breaks the collision-resistance property of $H$.    
\end{itemize}
Hence, both cases contradict security assumptions, yielding a contradiction. 
\qed
\end{proof}

%Note that in \cref{con:one-time sig}, the length of the verification key is $2\cdot n\cdot w=4\cdot \mu^2 \cdot w$. In the next section, we will need to sign a messages consisting of concatenation of two verification keys. Hence, we will need $\ell\geq 8\times \mu^2 \times w$. 

\subsection{Definition (Many-Message) Signatures}

We recall the definition of a (many-message) signature scheme. 

\begin{definition}[Digital Signatures]
A \emph{digital signature (DS)} scheme over classical message space $\hildd{M}$ consists of the following QPT algorithms: 
\begin{itemize}
    \item $\textsf{Gen}(1^\lambda)$: Outputs a secret key $\sk$ and a verification key $\vk$.
    \item $\textsf{Sign}(\sk,\mu):$ Outputs a signature ${\sigma}$ for $\mu \in \hildd{M}$ using $\sk$. 
    \item $\textsf{Verify}(\vk,\mu', {\sigma'})$: Verifies whether ${\sigma'}$ is a valid signature for $\mu' \in \hildd{M}$ using $\vk$ and correspondingly outputs $\top/\perp$.
\end{itemize}
\end{definition}

\begin{definition}[Correctness]
A digital signature scheme is $p$-\emph{correct} if for any message $\mu\in \hildd{M}$:
\begin{align*} \Pr{\left[
\begin{tabular}{c|c}
 \multirow{2}{*}{$\textsf{Verify}(\vk,\mu, {\sigma})=\top\ $} &   $(\sk,\vk)\ \leftarrow \textsf{Gen}(1^\lambda)$ \\ 
&  ${\sigma}\ \leftarrow \textsf{Sign}(\sk,\mu)$\\
 \end{tabular}\right]} \geq p .
\end{align*}
If $p\geq 1-\negl[\lambda]$, then we simply say that the scheme is \emph{statistically correct}. 
\end{definition}

We recall the security experiment testing \emph{strong unforgeability} (SUF) of a digital signature scheme. 

\smallskip \noindent\fbox{%
    \parbox{\textwidth}{%
\textbf{Experiment} $\textsf{Sign}^{\text{SUF}}_{\Pi, \adv}({\lambda})$:
\begin{enumerate}
    \item Sample $(\sk,\vk)\leftarrow \textsf{Gen}(1^\lambda)$.
    \item $\adv$ is given $\vk$ and classical access to the signing oracle $\textsf{Sign}(\sk, \cdot)$.
    \item Let $\mathcal{Q}$ denote the set of messages-signature pairs returned by the signing oracle.
    \item $\adv$ outputs a forgery $(m^*,\sigma^*).$
   \item The output of the experiment is $1$ if
   \[
       (m^*,\sigma^*)\notin \mathcal{Q} \text{ and } \textsf{Verify}(\vk, m^*, \sigma^*)=1 .
   \]
   \item Otherwise, the output is 0. 
\end{enumerate}}}
\smallskip

\begin{definition}
    A signature scheme $\Pi$ is \emph{strongly unforgeable ({SUF})}, if for any QPT $\adv$,
    \begin{align*} 
        \Pr [\textsf{Sign}^{\textit{SUF}}_{\Pi, \adv}({\lambda})=1]\leq \negl[\lambda].
    \end{align*}
\end{definition}

\subsection{Stateful (Many-Message) Signatures}
\label{sec:sig 2}

We upgrade a OMS scheme to a many-message scheme following the approach given in \cite{NY89} with minimal modifications. The signer needs to store previous signatures so this approach is called \emph{stateful}. 

The idea is to construct a tree where each node contains a pair of keys sampled from a one-time DS scheme. In any non-leaf node, the secret key is used to sign the verification keys of the children nodes. Meanwhile, the leaf keys are used to sign messages. We let $\alpha$ signify the label of the root node in the authentication tree. More generally, we label the rest of the nodes in a natural way: for any node labeled $\alpha v$, we label its left child and right child by $\alpha v0$ and $\alpha v1$, respectively. 

\begin{construct}
\label{con:sig}
Let $n\coloneqq n(\lambda)$ be even and polynomial in the security parameter $\lambda\in \mathbb{N}$. Let $({\textsf{OT.Gen}},\textsf{OT.Sign},\textsf{OT.Ver})$ be the algorithms of the OTS scheme given in Construction \ref{con:one-time sig 2} on messages of length at least twice the length of the verification key \footnote{This can be achieved using a $\botUOWHF$ with sufficient compression in Construction \ref{con:one-time sig 2}.}. Let $\textbf{M}$ denote the signer's memory, which evolves as more messages are signed. The construction for a many-message (stateful) signature scheme on $n$-bit messages is as follows:
\begin{itemize}
    \item $\textsf{KeyGen}(1^\lambda):$ Sample $(\textsf{s}_\alpha,\textsf{v}_\alpha)\gets \textsf{OT.Gen}(1^\lambda)$
    and set $\textbf{M}=\emptyset$. 
    
    \item $\textsf{Sign}(\sk,\textbf{M},m):$ 
    At any point in the procedure, if any computation returns $\bot$, then abort and output $\bot$. 
    \begin{enumerate}
            \item For each $i\in [n]$: 
        \begin{enumerate}
        \item For $\tau =0,1$, try to retrieve the key pair associated to the node labeled $\alpha m_1m_2\ldots m_{i-1}\tau$. If the node is not found, then generate a key pair:
                \begin{align*}
            (\textsf{s}_{\alpha m_1m_2\ldots m_{i-1}\tau},\textsf{v}_{\alpha m_1m_2\ldots m_{i-1}\tau})\gets \textsf{OT.Gen}(1^\lambda).
        \end{align*}
        and store it in node $\alpha m_1m_2\ldots m_{i-1}\tau$ in $\textbf{M}$. 
        
        \item Try to retrieve a signature in node $ \alpha m_1m_2\ldots m_{i-1}$. If no signature is found then generate one:
        \begin{align*}
            \sigma_{\alpha m_1m_2\ldots m_{i-1}}\gets \textsf{OT.Sign}(\textsf{s}_{\alpha m_1m_2\ldots m_{i-1}}, (\textsf{v}_{\alpha m_1m_2\ldots m_{i-1}0},\textsf{v}_{\alpha m_1m_2\ldots m_{i-1}1})).
        \end{align*} 
        and store it in node $\alpha m_1m_2\ldots m_{i-1}$ in $\textbf{M}$.
        \end{enumerate}
    \item Sign the message $\sigma_{\alpha m}\gets \textsf{OT.Sign}(\textsf{s}_{\alpha m},m)$.
    \item Output $(\sigma_{\alpha m_1\ldots m_{i-1}},\textsf{v}_{\alpha m_1\ldots m_{i-1}0},\textsf{v}_{\alpha m_1\ldots m_{i-1}1})_{i\in [n]}$.    \end{enumerate}

     \item $\textsf{Verify}(\vk,m,\sigma):$ 
     \begin{enumerate}
         \item If $\sigma=\bot$, then output $\bot$. 
         \item Otherwise, interpret $\sigma$ as 
         \begin{align*}
             (\sigma_{\alpha m_1\ldots m_{i-1}},\textsf{v}_{\alpha m_1\ldots m_{i-1}0},\textsf{v}_{\alpha m_1\ldots m_{i-1}1})_{i\in [n]}.
         \end{align*} 
    \item Output $\top$ if the following conditions are satisfied:
    \begin{enumerate}
        \item For each $i\in [n]$: 
        \begin{align*}
            \textsf{OT.Ver}(\textsf{v}_{\alpha m_1\ldots m_{i-1}}, (\textsf{v}_{\alpha m_1\ldots m_{i-1}0},\textsf{v}_{\alpha m_1\ldots m_{i-1}1}), \sigma_{\alpha m_1\ldots m_{i-1}})=\top.
        \end{align*}
        \item $\textsf{OT.Ver}(\textsf{v}_{\alpha m}, m, \sigma_{\alpha m})=\top$.
    \end{enumerate} 
    \item Otherwise, output $\perp$.  
         \end{enumerate}

\end{itemize}
\end{construct}

The security proof follows similarly to the proof in the classical case \cite{Gol04}.

\begin{lemma}[Sect.~6.4.2.2. of \cite{Gol04}]
\label{lem:sec 2}
\cref{con:sig} is a SUF many-message (stateful) signature scheme on $n$-bit messages assuming the existence of $\botUOWHF$s.
\end{lemma}

%See Sec.~6.4.2.2. of \cite{Gol04} for the proof. 

\subsection{Stateless Digital Signatures}
\label{sec:sig 3}

We upgrade \cref{con:sig} to a stateless digital signature scheme using our construction for $\botPRF$s following a similar approach to \cite{Gol04}, which is based on traditional $\PRF$s. The idea is to use a $\botPRF$ to determine the output of the coin tosses used in the OMS key generation procedures employed in the authentication tree. 
Specifically, let ${\textsf{OT.Gen}}(1^\lambda)$ be a pseudodeterministic QPT algorithm for Construction \ref{con:one-time sig 2}. The algorithm first classically samples random inputs and keys for appropriate $\botUOWHF$s. We let ${\textsf{OT.Gen}}(1^\lambda;r)$ denote the algorithm where these random decisions are instead determined by $r$. 

Instead of choosing $r$ randomly each time, it can be set to the outcome of a $\botPRF$, allowing the signer to generate any key-pair in the authentication key on demand without needing to store previous signatures. While the keys generated in this fashion are not deterministic and may sometimes include $\perp$, we simply abort and output $\perp$ in this scenario. 

\begin{construct}
\label{con:many-time sig}
Let $\lambda\in \mathbb{N}$ be the security parameter and $n\coloneqq n(\lambda)$ be an even polynomial on $\lambda$. Let $({\textsf{OT.Gen}},\textsf{OT.Sign},\textsf{OT.Ver})$ be the algorithms of the OTS scheme given in Construction \ref{con:one-time sig 2} on messages at least twice the length of the verification key such that the key generation algorithm makes $q$ coin tosses ($q$ polynomial in $\lambda$). For $i\in [n]$, let $\textsf{PRF}^i$ be a $(\mu,i+q,q)$-$\botPRF$. The algorithms for a (stateless) DS scheme on $n$-bit messages is as follows:
\begin{itemize}
    \item $\textsf{KeyGen}(1^\lambda):$ 
    \begin{enumerate}
        \item Sample $(\textsf{s}_\alpha,\textsf{v}_\alpha)\gets \textsf{OT.Gen}(1^\lambda)$.
        \item For each $i\in [n]$ and $\tau=0,1$, sample $k_i\gets \{0,1\}^\lambda$ and $r_{i,\tau}\gets \{0,1\}^{i+q}$. 
        \item Set $\sk\coloneqq (\textsf{s}_\alpha,{k}_i,r_{i,\tau})_{i\in [n],\tau \in \{0,1\}}$ and $\vk\coloneqq \textsf{v}_\alpha$.
        \item Output $(\sk,\vk)$. 
    \end{enumerate} 
    
    \item $\textsf{Sign}(\sk,m):$ 
        At any point in the procedure, if any computation returns $\bot$, then abort and output $\bot$. 
    \begin{enumerate}
        \item For each $i\in [n]:$
   \begin{enumerate}
        \item For $\tau =0,1$, generate the key pair:
        \begin{align*}
            (\textsf{s}_{\alpha m_1m_2\ldots m_{i-1}\tau},\textsf{v}_{\alpha m_1m_2\ldots m_{i-1}\tau})\gets \textsf{OT.Gen}(1^\lambda;\textsf{PRF}^i(k_i, m_1m_2\ldots m_{i-1}\tau 0^{q-1}\oplus r_{i,\tau})).
        \end{align*}
        \item Sign the verification keys:
        \begin{align*}
            \sigma_{\alpha m_1m_2\ldots m_{i-1}}\gets \textsf{OT.Sign}(\textsf{s}_{\alpha m_1m_2\ldots m_{i-1}}, (\textsf{v}_{\alpha m_1m_2\ldots m_{i-1}0},\textsf{v}_{\alpha m_1m_2\ldots m_{i-1}1})).
        \end{align*}
    \end{enumerate}
    \item Sign the message $\sigma_{\alpha m}\coloneqq \textsf{OT.Sign}(\textsf{s}_{\alpha m},m)$.
    \item Output $(\sigma_{\alpha m_1\ldots m_{i-1}},\textsf{v}_{\alpha m_1\ldots m_{i-1}0},\textsf{v}_{\alpha m_1\ldots m_{i-1}1})_{i\in [n]}$.
  \end{enumerate}

     \item $\textsf{Verify}(\vk,m,\sigma):$ 
        \begin{enumerate}
         \item If $\sigma=\bot$, then output $\bot$. 
    \item Otherwise, interpret $\sigma$ as $(\sigma_{\alpha m_1\ldots m_{i-1}},\textsf{v}_{\alpha m_1\ldots m_{i-1}0},\textsf{v}_{\alpha m_1\ldots m_{i-1}1})_{i\in [n]}$. 
    \item Output $\top$ if the following conditions are satisfied:
    \begin{enumerate}
        \item For each $i\in [n]$: 
        \begin{align*}
            \textsf{OT.Ver}(\textsf{v}_{\alpha m_1\ldots m_{i-1}}, (\textsf{v}_{\alpha m_1\ldots m_{i-1}0},\textsf{v}_{\alpha m_1\ldots m_{i-1}1}), \sigma_{\alpha m_1\ldots m_{i-1}})=\top.
        \end{align*}
        \item $\textsf{OT.Ver}(\textsf{v}_{\alpha m}, m, \sigma_{\alpha m})=\top$.
    \end{enumerate} 
    \item Otherwise, output $\perp$. 
    \end{enumerate}
\end{itemize}
\end{construct}

\begin{lemma}\label{lemma:cor-digital-sig}
 \cref{con:many-time sig} is $(1-\mu)^p$-correct where $p={4n\lambda+10n}$ assuming the existence of $\botPRF$s and $\botUOWHF$s. 
\end{lemma}

\begin{proof}
Fix $m\in \{0,1\}^n$ and sample $(\sk,\vk)\gets \textsf{KeyGen}(1^\lambda)$. To sign $m$, for each $j\in [n]$ and $\tau=0,1$, we evaluate $\textsf{PRF}^j$ on the input $(k_j,x^{j,\tau})$ where $x^{j,\tau}\coloneqq m_1m_2\ldots m_{j-1}\tau 0^{q-1}\oplus r_{j,\tau}$. Let $\mathcal{G}_j$ denote the good set of inputs for $\textsf{PRF}^j(k_j,\cdot)$. 
By the pseudodeterminism of the $\botPRF$s and the randomness of the inputs $x^{j,\tau}$, there is at least $(1-\mu)^{2n}$ probability that $x^{j,\tau}\in \mathcal{G}_j$ for all $j\in [n]$ and $\tau =0,1$. 

Next, we also need correctness of $n$ OMSs employed in the signing procedure. By Lemma \ref{lem:corr ots}, Construction \ref{con:one-time sig 2} satisfies $(1-\mu)^{2\lambda+3}$-correctness using a $\botUOWHF$ with sufficient compression. 

A small issue is that in our setting, we use pseudorandomness through $\botPRF$s instead of actual randomness when generating the key pairs for the OMSs. This may affect the pseudodeterminism of the $\botUOWHF$ utilized in the OTS scheme. However, if using pseudorandomness significantly worsens the correctness of the OTS scheme, then this easily gives an adversary that can distinguish between the outputs of a $\botPRF$ and a random $\bot$-function, which yields a contradiction. Hence, the correctness error of the pseudorandom based OTS is within $(1-\negl[\lambda])$ factor from $(1-\mu)^{4\lambda+3}$. We generate $2n$ one-message signatures when signing $m$. 

To sum up, the probability that a signature of $m$ passes verification is $(1-\mu)^{2n}(1-\mu)^{4n\lambda+6n}(1-\negl[\lambda])^{2n}\leq (1-\mu)^{4n\lambda+10n}$.
    \qed
\end{proof}

\begin{theorem}
\label{thm:sec sig}
    \cref{con:many-time sig} is a SUF digital signature scheme assuming the existence of $\botPRF$s and $\botUOWHF$s.
\end{theorem}

\begin{proof}
Let $\textbf{H}_0$ be the actual scheme. For $i\in [n]$, let $\textbf{H}_i$ be the hybrid that is the same as $\textbf{H}_{i-1}$ except $\textsf{PRF}^j(k_i,\cdot)$ is replaced with the $\bot$-random function $R_{\textsf{PRF}^i}$, where $R_{\textsf{PRF}^i}$ is defined as the algorithm which evaluates $\textsf{PRF}^i(k_i,\cdot)$ and if the result is $\bot$, it returns $\bot$, and otherwise outputs a random string (see \cref{def:bot-prf-security}). Such a hybrid is efficient since to evaluate $R_{\textsf{PRF}^i}$, the experiment does not need to sample an entire random function, but rather it can compute $\textsf{PRF}^i(k_i,\cdot)$ and if the result is $\bot$, it returns $\bot$, and otherwise it outputs a random output and stores the evaluation in its memory for future use. 

We first argue that the final hybrid is secure. Then, we argue that if there exists $j\in [n]$ such that there exists a QPT distinguisher that can differentiate between hybrids $\textbf{H}_{j-1}$ and $\textbf{H}_{j}$, then there exists QPT distinguisher that can distinguish ${\textsf{PRF}^j}$ and $R_{\textsf{PRF}^j}$, contradicting the security of $\botPRF$s. Let $\Pi_{\textsf{stf}}$ be the stateful scheme (\cref{con:sig}).

The scheme $\textbf{H}_n$ is the same as $\Pi_{\textsf{stf}}$ except the key generation algorithm uses random $\bot$-functions instead of making random choices during key generation. Now if no random $\bot$-function evaluates to $\perp$ during the procedure, then a valid signature is generated in the same way as in $\Pi_{\textsf{stf}}$. Alternatively, if there is a $\perp $, then the process is aborted and the signature is $\perp$. Informally, this scenario clearly does not aid the adversary in forging signatures since whether the signature is $\perp$ or not is independent of the signing key. 

To prove this formally, assume for contradiction that there exists a QPT $\adv$ such that $\Pr[\textsf{Sign}^{\text{SUF}}_{\textbf{H}_n, \adv}({\lambda})=1]$ is non-negligible. We can construct an adversary $\mathcal{B}$ such that $\Pr[\textsf{Sign}^{\text{SUF}}_{\Pi_{\textsf{stf}}, \mathcal{B}}({\lambda})=1]$ is non-negligible. 

$\mathcal{B}$ in the SUF security experiment for $\Pi_{\textsf{stf}}$ gets query access to the signing oracle and a verification key $\vk$. $\mathcal{B}$ samples $k_i\leftarrow \{0,1\}^\lambda$ for each $i\in [n]$ and runs $\adv$ on $\vk$. Assume $\adv$ requests the signature of a message $m$. $\mathcal{B}$ queries the oracle on $m$ and gets as response $\sigma_m$. $\mathcal{B}$ computes $ {\textsf{PRF}^i}(k_i, m_1m_2\ldots m_{i}\tau 0^{q-1}\oplus r_{i,\tau})$ for all $i\in [n]$ and $\tau\in \{0,1\}$, and if any of the evaluations return $\perp$, then $\mathcal{B}$ returns $\perp$ to $\adv$ and otherwise returns $\sigma_m$. $\adv$ submits an output $(m^*,\sigma^*)$. $\mathcal{B}$ simply outputs the same pair $(m^*,\sigma^*)$. Clearly, $\Pr[\textsf{Sign}^{\text{SUF}}_{\Pi_{\textsf{stf}}, \mathcal{B}}({\lambda})=1]$ is non-negligible which contradicts \cref{lem:sec 2}. Hence, $\textbf{H}_n$ is a SUF scheme. 

Now assume there exists a QPT adversary $\mathcal{A}$ that can distinguish between hybrid $\textbf{H}_{j-1}$ and $\textbf{H}_j$ with non-negligible advantage. We construct a QPT algorithm $\mathcal{D}$ with oracle access to a function $G $ which distinguishes whether $G$ is a $\bot$-random function or a $\botPRF$ with non-negligible probability. 
First, $\mathcal{D}$ generates $(\sk,\vk)\leftarrow \textsf{KeyGen}(1^\lambda)$ and runs $\adv(\textsf{v})$. Whenever $\adv$ queries the signing oracle, $\mathcal{D}$ answers the query by following the signing procedure, except it replaces any computation of ${\textsf{PRF}^i}(k_i,\cdot)$ with $R_{{\textsf{PRF}^i}}$ if $i<j$ and with $G$ if $i=j$.

$\adv$ succeeds in distinguishing the two hybrids with non-negligible probability meaning $\mathcal{D}$ distinguishes whether between a $\botPRF$ and a random $\bot$-function, yielding a contradiction. Therefore, no QPT adversary can distinguish between any two consecutive hybrids with non-negligible advantage.
Hence, by the triangle inequality, $H_0$ is a SUF signature scheme. 
\qed
\end{proof}

We can improve correctness, reducing the error to a negligible value, through repetition. In particular, Construction \ref{con:many-time sig} is $(1-\mu)^p$-correct where $p=(1-\mu)^{4n\lambda+10n}$. As long as use $\botUOWHF$s and $\botPRF$s with pseudodeterminism error of $\mu<1/p^2$, then $(1-\mu)^p\geq  \frac{1}{4}$. 

Therefore, by signing a message under $4\lambda$ independent keys, then one of the signatures will pass verification except with negligible probability. In the resulting scheme, the verifier accepts as long as one of the signatures is valid. It is easy to see that this scheme satisfies statistical correctness and unforgeability. However, it does not satisfy strong unforgeability since it is easy to create a new signature given a valid signature. 

It seems difficult to achieve strong unforgeability along with statistical correctness through our approach and, hence, this is left as an open problem. Nevertheless, our results give the option to choose between statistical correctness and strong unforgeability.

To sum up, the existence of a stateless DS satisfying UF and statistical correctness based on $\botPRF$ and $\botUOWHF$ is implied by \cref{thm:sec sig,lem:sec 2,lem:ots 2}. Furthermore, by \cref{cor:bot-prg-from-sprs,thm:uowhf,cor:bot-PRF-from-SPRS}, $\botPRF$ and $\botUOWHF$ can be based on $\SPRS$. All in all, we get the following:

\begin{theorem}[Digital signatures from $\SPRS$]
    \label{thm:main sig}
    Let $\lambda \in \mathbb{N}$ be the security parameter. Assuming the existence of $c\log(\lambda)$-$\PRS$ with large enough $c\in \mathbb{N}$, there exists stateless (many-message) digital signatures on $\lambda$-bit messages with classical keys satisfying UF as well as statistical correctness. 
\end{theorem}

A direct application of our digital signature scheme pertains to symmetric encryption. It is well-known that a CPA-secure symmetric encryption scheme can be constructed from $\PRF$s. This construction can be adapted to $\botPRF$s without difficulty. Another well-known point is that a CPA-secure scheme can be upgraded to enable CCA-security by signing the ciphertexts. Hence, as a corollary to \cref{cor:bot-PRF-from-SPRS} and \cref{thm:main sig}, we obtain the following result:

\begin{corollary}
Assuming the existence of $\SPRS$, there exists a CCA-secure symmetric encryption scheme with classical ciphertexts.
\end{corollary}

\bibliographystyle{splncs04}
\bibliography{mybib2}
%\printbibliography %for biblatex

\appendix
\section{Generalization Lemmas}

We show simple and frequently satisfied conditions allowing to replace $\PRG$s and $\PRF$s with $\PDPRG$ and $\PDPRF$, respectively,
while preserving security. 

The next lemma shows that $\PDPRG$ can safely replace $\PRG$s in a setting where the generator is evaluated on uniformly sampled inputs.

\begin{lemma}
\label{lem:gener PRG}
Let $\textsf{Exp}^F_{\adv}(\lambda)$ be a security experiment with binary output using oracle access to a function $F$ against an adversary $\adv$. Assume the experiment queries $F$ on uniformly sampled inputs. Let $G$ be a truly random function and $\tilde{G}$ be a $\PDPRG$, both mapping $\{0,1\}^{m(\lambda)}$ to $ \{0,1\}^{\ell(\lambda)}$. Assume that for any QPT adversary $\adv$:
  \begin{align*}
      \Pr[\textsf{Exp}^{G}_{\adv}(\lambda)=1] \leq \negl[\lambda].
  \end{align*}
    Then, for any QPT adversary $\adv'$:
        \begin{align*}
         \Pr[\textsf{Exp}^{\tilde{G}}_{\adv'}(\lambda)=1] \leq \negl[\lambda].
    \end{align*} 
 \end{lemma}

\begin{proof}
    Assume for contradiction that the theorem is false. In other words, if ${G}$ is a truly random function, then for any QPT adversary $\adv$:
  \begin{align*}
      \Pr[\textsf{Exp}^{G}_{\adv}(\lambda)=1] \leq \negl[\lambda].
  \end{align*}
    However, there exists a QPT adversary $\adv'$ and polynomial $p$ such that:
        \begin{align*}
         \Pr[\textsf{Exp}^{\tilde{G}}_{\adv'}(\lambda)=1] \geq \frac{1}{p(\lambda)}.
    \end{align*} 

    Define $\mathcal{B}^{F}_{\adv'}$ as the algorithm with oracle access to an arbitrary function $F$ which runs $\textsf{Exp}^{F}_{\adv'}(\lambda)$ using oracle access to $F$ and outputs the result of the experiment. 
    
    Note that if $F$ is a truly random function, then $\Pr[\mathcal{B}^{F}_{\adv'}=1] \leq \negl[\lambda]$ by our assumption. However, if $F=\tilde{G}$, then $\Pr[\mathcal{B}^{\tilde{G}}_{\adv'}=1] \geq \frac{1}{p(\lambda)}$. Hence, $\mathcal{B}$ can distinguish $\tilde{G}$ from random on uniformly sampled inputs and thus breaks $\PDPRG$ security.
    \qed
\end{proof}

Next, we present a similar result but for $\PDPRF$. This result is not relevant for our work but may be of independent interest. See \cite{ALY23} for the rigorous definition of $\PDPRF$.

\begin{lemma}
\label{thm:gener PRF}
Let $\textsf{Exp}^F_{\adv}(\lambda)$ be a security experiment with binary output using oracle access to a function $F$ against an adversary $\adv$. Assume for any QPT $\adv$, the experiment queries the function $F$ non-adaptively on distinct inputs (except with negligible probability). Let $G:\{0,1\}^{m(\lambda)}\rightarrow \{0,1\}^{\ell(\lambda)}$ be a truly random function and $\tilde{G}:\{0,1\}^\lambda \times \{0,1\}^{m(\lambda)}\rightarrow \{0,1\}^{\ell(\lambda)}$ be a $\PDPRF$. Assume that for any QPT adversary $\adv$:
\begin{align*}
    \Pr[\textsf{Exp}^{G}_{\adv}(\lambda)=1] \leq \negl[\lambda].
\end{align*}
    Then, for any QPT adversary $\adv'$:
        \begin{align*}
         \Pr_{k\leftarrow \{0,1\}^\lambda}[\textsf{Exp}^{\tilde{G}(k,\cdot)}_{\adv'}(\lambda)=1] \leq \negl[\lambda].
    \end{align*} 
 \end{lemma}

The proof follows the same ideas as the proof of \cref{lem:gener PRG}. 

\section{Quantum Public-Key Encryption}
\label{sec:pub key intro}

\subsection{Definition}
We recall the definition of quantum public-key encryption with tamper-resilient quantum public-keys~\cite{KMNY23} or simply, $\qpke$. For the sake of simplicity, we only consider bit encryption which can be further boosted to encrypt messages of arbitrary fixed length by standard parallel repetition. \anote{I think we do not need the discussion below, so I am putting them in notes instead.} \anote{In this setting, multiple key copies must be created and distributed due to the ``no-cloning theorem". Hence, we include an algorithm $\textsf{KeySend}$ describing how to send a quantum public key and an algorithm $\textsf{KeyReceive}$ describing how to process a public-key copy and extract a classical string, which we call the \emph{encryption key}, that enables secure encryption. We define our scheme with pure-state public-keys, however, mixed-state public-keys constructions exist \cite{KMNY23,BS232}, and can be defined in a similar manner.}

\begin{definition}[Quantum Public Key Encryption with tamper-resilient quantum public keys]
A \emph{quantum public encryption scheme (QPKE)} over message space $\hildd{M}$ consists of the following QPT algorithms\footnote{In~\cite{KMNY23}, $\skgen$ and $\dec$ were defined to be PPT algorithms, but we consider a more general definition allowing all algorithms to be QPT, to accommodate our construction.}: 
\begin{itemize}
    \item $\skgen(1^\lambda):$ On input a security parameter $1^\lambda$, outputs $(\sk,\vk)$, where $\sk$ is a classical private key and $\vk$ is a classical verification key from the security parameter $\lambda$. 
    \item $\pkgen(\sk):$ On input a classical secret key $\sk$, either outputs $\bot$ or outputs a quantum (possibly, mixed) state $\pk$ as the quantum public key.
    \item $\enc(\vk,\pk,b):$ On input a classical verification key $\vk$, a quantum public key $\pk$, and a classical message bit $b$, either outputs $\bot$ or a classical ciphertext $\ct$. %for $\mu \in \hildd{M}$ using $\textsf{k}$.
    \item $\dec(\sk, \ct)$: On input a classical secret key $\sk$, and a classical cipher text $\ct$, either outputs $\bot$ or a message bit $b$.
\end{itemize}
\paragraph{Correctness}
For every $\delta(\cdot)$, a QPKE scheme is said to be $1-\delta(\lambda)$-\emph{correct} (or, $\delta(\lambda)$ correctness error) if for any message bit $b \in \{0,1\}$,
\begin{align*} \Pr{\left[
\begin{tabular}{c|c}
 \multirow{4}{*}{$b'=b\ $} &   $(\sk,\vk)\ \gets \skgen(1^\lambda)$ \\ 
 & $pk\gets \pkgen(\sk)$\\
 & $\ct\ \gets \enc(\vk,\pk,b)$\\
 & $b'\gets  \dec(\sk,\ct)$\\
 \end{tabular}\right]} \geq 1-\delta(\lambda).
\end{align*}\label{def:qpke-correctness}
\paragraph{IND-pkT-CPA security}
    A QPKE scheme $\Pi$ is $\cpatamp$ secure, i.e., satisfies tamper-proof CPA indistinguishability if for any QPT $\adv$ in security experiment $\QPKE^\cpatamp_{\Pi,\adv}$ (see \Cref{fig:pkt-cpa}), there exists a negligible function $\negl[\lambda]$ such that
    \begin{align*}
         \Pr[\QPKE^\cpatamp_{\Pi,\adv}({\lambda})=1] \leq \frac{1}{2} + \negl[\lambda].
    \end{align*}\label{def:qpke}
\end{definition}
%The following experiment tests \emph{indistinguishability against public-key tempering chosen plaintext attacks} (IND-pkT-CPA).

\begin{figure}[!htb]
   \begin{center} 
   \begin{tabular}{|p{16cm}|}
    \hline 
\begin{center}
\underline{$\QPKE^{\cpatamp}_{\Pi,\adv}({\lambda})$}: 
\end{center}
\begin{enumerate}
\item Challenger $\ch$ generates $(\sk,\vk)\gets \skgen(1^\lambda)$.
    \item $\adv$ sends a number $m(\lambda)$.
    \item $\ch$ generates $\pk_i\gets \pkgen(\sk)$ for every $i\in [m(\lambda)]$, and sends $(\{\pk_i\}_{i\in[m(\lambda)]},\vk)$ to $\adv$.
    \item $\adv$ sends a state $\pk'$.
    \item $\ch$ samples a bit $b\xleftarrow{\$}  \{0,1\}$ and generates $\ct\gets \enc(\vk,\pk',b)$, and sends $\ct$ to $\adv$.
    \item $\adv$ submits a bit $b'$.
  % \item $(\rho,\textsf{st})\gets \adv^{\pkgen(\sk)}(\vk)$.
  %   \item $\ct\leftarrow \enc(\textsf{k},\mu_b)$.
  %   \item $b'\leftarrow \adv^{\enc(k,\cdot)}(\ct,\textsf{st})$.
    \item The output of the experiment is $1$ if $b=b'$, and $0$ otherwise. 
\end{enumerate}
\ \\ 
\hline
\end{tabular}
    \caption{CPA security experiment for tamper-proof QPKE.}
    \label{fig:pkt-cpa}
    \end{center}
\end{figure}

% \begin{figure}
% \noindent\fbox{%
%     \parbox{\textwidth}{%
% \textbf{Experiment} $\QPKE^{\cpatamp}_{\Pi,\adv}({\lambda})$:
% \begin{enumerate}
%     \item Challenger $\ch$ generates $(\sk,\vk)\gets \skgen(1^\lambda)$.
%     \item $\adv$ sends a number $m(\lambda)$.
%     \item $\ch$ generates $\pk_i\gets \pkgen(\sk)$ for every $i\in [m(\lambda)]$, and sends $(\{\pk_i\}_{i\in[m(\lambda)]},\vk)$ to $\adv$.
%     \item $\adv$ sends a state $\pk'$.
%     \item $\ch$ samples a bit $b\xleftarrow{\$}  \{0,1\}$ and generates $\ct\gets \enc(\vk,\pk',b)$, and sends $\ct$ to $\adv$.
%     \item $\adv$ submits a bit $b'$.
%   % \item $(\rho,\textsf{st})\gets \adv^{\pkgen(\sk)}(\vk)$.
%   %   \item $\ct\leftarrow \enc(\textsf{k},\mu_b)$.
%   %   \item $b'\leftarrow \adv^{\enc(k,\cdot)}(\ct,\textsf{st})$.
%     \item The output of the experiment is $1$ if $b=b'$, and 0 otherwise. 
%     \end{enumerate}}}\label{fig:pkt-cpa}
% \end{figure}    
%\smallskip  

\subsection{Construction}\label{subsec:construct-qpke-proof}

In this section, we show that $\qpke$ can be constructed based on $\PDPRF$.

\begin{theorem}\label{thm:qpke-from-pdprf-restated}
    Assuming the existence of a $c\log(\secpar)$-$\PRS$ for some large enough $c\in \NN$, there exists a $\cpatamp$ secure $\QPKE$ scheme (see \Cref{def:qpke}).
\end{theorem}

\begin{proof}
The proof follows by combining \Cref{thm:qpke-from-pdprf-restated-imperfect-correctness,thm:qpke-parallel-repetition-lift}.
\qed
\end{proof}

\begin{theorem}\label{thm:qpke-from-pdprf-restated-imperfect-correctness}
    Assuming the existence of a $c\log(\secpar)$-$\PRS$ for some large enough $c\in \NN$, there exists a $(1-\frac{1}{\secpar})$-correct $\QPKE$ scheme that is $\cpatamp$ secure (see \Cref{def:qpke}).
\end{theorem}
\begin{proof}
By \Cref{thm:main sig}, assuming the existence of a $c\log(\secpar)$-PRSs with sufficiently large $c$ and $\secpar$, there exists a $(1-\frac{1}{\secpar})$-correct digital signature scheme for messages of length $\secpar$ that is strongly unforgeable and has unique non-$\bot$ signatures, meaning the signature string for any message $m$ is either a fixed string $\sigma$ or $\bot$, except with negligible probability.

    We use this $(1-\frac{1}{\secpar})$-correct digital signature with the above-mentioned properties to instantiate the same generic construction of $\qpke$ in~\cite{KMNY23}. It is easy to see that the resulting $\QPKE$ scheme is $(1-\frac{1}{\secpar})$-correct by the unique non-$\bot$ signatures and the fact that for a fixed message the signature is not $\bot$ except with probability $\frac{1}{\secpar}$. 
    
    Next note that in the $\cpatamp$-security game for the construction in~\cite{KMNY23}, having a $(1-\frac{1}{\secpar})$-correct digital signature instead of a digital signature with negligible correctness only means that in some cases, the public key given by the challenger to the adversary at the very beginning might be $\bot$, which makes it only harder for the adversary. Hence it is easy to see that if the construction in~\cite{KMNY23} is $\cpatamp$ secure when instantiated with a digital signature with negligible correctness error and strong unforgeability, then so is the case when the construction is instantiated with a digital signature with $(1-\frac{1}{\poly})$-correctness, unique non-$\bot$ signatures and strong unforgeability, as is the case with our construction. By~\cite[Theorem 6.1]{KMNY23}, the construction in~\cite{KMNY23} is $\cpatamp$ secure when instantiated with a digital signature with negligible correctness error, unique signatures and strongly unforgeability is $\cpatamp$ secure. Therefore, the same holds for our construction as well. %Therefore the construction we get by instantiating the generic the construction in~\cite{KMNY23} by a $(1-\mu\cdot \secpar)$-correct digital signature scheme with strong unforgeability and unique non-$\bot$ signatures, based on $\mu(\secpar)$ $\bot$-PRG is also  $\cpatamp$.
    \qed
\end{proof}

\begin{theorem}\label{thm:qpke-parallel-repetition-lift}
    Suppose $\Pi$ is a $(1-\delta(\secpar)$-correct $\QPKE$ that  is $\cpatamp$ secure, where $\delta(\secpar)\in \poly$. Let $q=q(\secpar)\equiv \secpar.$ Define $\Pi^q$ to be $q$-parallel repetition of $\Pi$, i.e., 
    \begin{enumerate}
        \item $\Pi^q.\skgen\equiv(\Pi.\skgen)^{\tensor q}$.
        \item $\Pi^q.\pkgen\equiv(\Pi.\pkgen)^{\tensor q}$.
        \item $\Pi^q.\enc\equiv(\Pi.\enc)^{\tensor q}$.
        \item $\Pi^q.\dec((\sk_1,\ldots,\sk_q),(\ct_1,\ldots,\ct_q))$:
        \begin{itemize}
            \item For every $i\in [q(\secpar)]$, generate $b_i\gets \Pi.\dec(\sk_i,\ct_i)$, and if $b_i\neq \bot$, add $i$ to $S$. 
            \item Output $\bot$ if $S=\emptyset$ or $b_i\neq b_j$ for some $i,j\in S$, else output $b_i$ where $i\equiv\min(S)$.  %$\bigwedge_{i=1}^q(\Pi.\dec(\sk_i,))^{\tensor q}$.
        \end{itemize}
    \end{enumerate}
    Then,
    \begin{enumerate}
        \item $\Pi^{q(\secpar)}$ has negligible correctness error.
        \item $\Pi^{q(\secpar)}$ is $\cpatamp$ secure.
    \end{enumerate}
\end{theorem}
\begin{proof}
The proof of the correctness follows by noting that the correctness experiment of $\Pi^q$ is the same as $q$-parallel correctness experiment of $\Pi$ where success for the experiment with respect to $\Pi^q$ corresponds to the event that at least one the $q$-parallel correctness experiment of $\Pi$ are successful. Hence, the probability of success in the correctness experiment for $\Pi^q$ or the correctness precision is at least $1-\delta^q$ which is overwhelming by the choice of $q$. 

We first reconsider the $\cpatamp$ security game (see \Cref{fig:pkt-cpa}) in the form of an indistinguishability game and then use standard hybrid arguments wherein we consider $q-1$ intermediate hybrids, and then a reduction to the $\cpatamp$-security of $\Pi$ with a loss of $\frac{1}{q}$ on the success probability. The hybrid argument is essentially the same as that in the proof of LOR security guarantees of the encryption scheme, or in the proofs showing parallel repetition of CPA-secure bit-encryption implies CPA-secure encryption for arbitrary message length.
\qed
\end{proof}
\ifnum\iacr=1
    
\fi
\ifnum\iacr=1

\fi

\section{Proofs} \label{app:signatures}

\subsection{Proof of \cref{thm:OWF}}
\label{app:owf proof}
\begin{proof}
    Define the function $F$ by $F(z)=F(x,y)\coloneqq G(x)$ where the input $z\in \{0,1\}^{\ell(\lambda)}$ is parsed into $x\in \{0,1\}^{\lambda} $ and $y\in \{0,1\}^{\ell(\lambda)-\lambda}$. It is clear that $F$ satisfies the pseudo-determinism condition for $\botOWF$s (Definition \ref{def:botOWF}) by the pseudodeterminism of $G$. We just need to show that $F$ also satisfies the security condition. 

    Assume this is not true, then there exists a QPT algorithm $\adv$ and polynomial $p(\cdot)\geq 0$ such that:
    \begin{align*}
        \Pr_{z\gets \{0,1\}^{\ell(\lambda)}}[F(\adv(F(z)))=F_\top (z)]\geq \frac{1}{p(\lambda)}.
    \end{align*}
This implies: 
    \begin{align*}
        \Pr_{x\gets \{0,1\}^{\lambda}}[G(\adv(G(x)))=G_\top(x)]\geq \frac{1}{p(\lambda)}.
    \end{align*}
Define the set $\text{Im}_\top(G)$ of pseudodeterministic images of $G$ as elements $y\neq \perp$ such that there exists $x \in \{0,1\}^\lambda$ satisfying $\Pr[G(x)=y]\geq \frac{1}{2^\lambda}$. Each element of the input can map to at most $2^\lambda$ elements in $\text{Im}_\top(G)$ so $|\text{Im}_\top(G) |\leq 2^{2\lambda}$. This means $\Pr_{y\gets \{0,1\}^{\ell(\lambda)}}[y\in \text{Im}_\top(G)]\leq \negl[\lambda]$ since $\ell(\lambda)\geq 3\lambda$. 

Suppose we sample a bit $b\gets \{0,1\}$, an input $x\gets \{0,1\}^{\lambda}$, a string $y_0\gets R_G(x)$\footnote{$R_G(x)$ evaluates $G(x)$ and outputs $\perp$ if the result is $\perp$ and outputs a random string in $\{0,1\}^{\ell(\lambda)}$ otherwise.}, and $y_1\gets G(x)$. By the security of $\botPRG$, no QPT algorithm can distinguish $b$ from random given $y\coloneqq y_b$. Now, let $\textsf{A}$ be the event that $G(\adv(y)=y\neq \bot$, $\textsf{B}$ be the event that $y\in \text{Im}_\top(G)$, and $\overline{\textsf{B}}$ be the event that $y\notin \text{Im}_\top(G)$.

Then, we have the following\footnote{For any two events $\textsf{E}$ and $\textsf{E}'$, we let $\textsf{E}\textsf{E}'$ denote the event where both occur.}:
\[
   \Pr[b=0|\textsf{A}]=\Pr[b=0|\textsf{A}\textsf{B}]\Pr[\textsf{B} |\textsf{A}] +\Pr[b=0|\textsf{A} \overline{\textsf{B}} ]\Pr[\overline{\textsf{B}} |\textsf{A}] \leq \negl[\lambda].
\]
This is because 
$\Pr[b=0|\textsf{A}\textsf{B}]\leq \negl[\lambda]$ since in event $A$, $y\neq \bot$ but $\Pr_{y\gets \{0,1\}^{\ell(\lambda)}}[y\in \text{Im}_\top(G)]\leq \negl[\lambda]$. Furthermore, $\Pr[ \overline{\textsf{B}}|\textsf{A}]\leq \negl[\lambda]$ by definition of $\text{Im}_\top(G)$. 
Now, we construct a QPT algorithm $\adv'$ which can distinguish outputs of $G$ from random.  

\smallskip \noindent\fbox{%
    \parbox{\textwidth}{%
\textbf{Algorithm} $\adv'(y)$:
        \begin{enumerate}
            \item Run $\adv(y)$ and let $x'$ be the result. 
            \item If $G_\top(x')=y$ then output 1 and output 0 otherwise. 
        \end{enumerate}}}
        
It is easy to see that $\adv'$ can distinguish $b$ from random using $\adv$ and break $\botPRG$ security, yielding a contradiction. Hence, $F$ satisfies the security condition of $\botOWF$s. 
\qed
\end{proof}

\subsection{Proof of Theorem \ref{thm:pdprg_implies_botPRG}}
\label{sec:proof_pdprg_implies_botPRG}
\begin{proof}[Proof of \cref{thm:pdprg_implies_botPRG}]
    The theorem follows by combining these two propositions: \begin{enumerate}
        \item Weak $\botPRG$ from weak $\PDPRG$ (\cref{prop:PropPDPRG1}). 
        \item  Multi-time $\botPRG$ from weak $\botPRG$ assuming $\ell(\secpar)>\secpar^2$ (\cref{prop:PropPDPRG2}).
    \end{enumerate}
    \qed
\end{proof}

\begin{proposition}[$\botPRG$ from $\PDPRG$] 
    Let $G$ be a weak $\epsilon$-pseudorandom $(\mu,\nu,\ell)$-$\PDPRG$, then there exist a $(\mu,\ell)$-$\botPRG$, described in \cref{con:weak_botPRG_from_weak_pdprg}, with weak pseudorandomness $(\epsilon(\secpar)+\mu(\secpar) +\nu(\secpar)+\negl)$, for some negligible function $\negl$.
    \label{prop:PropPDPRG1}
\end{proposition}
\begin{proof}
    Consider the following construction.
    \begin{definition}[Voting Algorithm]
    $\textsf{vote}_\secpar(y_1,...,y_\secpar)$ outputs $y$ if at least $60\%$ of $y_i$ for $i \in [\secpar]$ satisfy $y_i = y$, and outputs $\bot$ if no such $y$ exists.
    \end{definition}
\begin{construct}[Weak $\botPRG$]
    Let $G$ be some weak $(\mu,\nu,\ell)$-$\PDPRG$.
    We define a family of QPT algorithms $\Gvote = \{\Gvote_\secpar\}_{\secpar\in\mathbb{N}}$ given by $\Gvote_\secpar(k)\coloneqq\textsf{vote}_\secpar(y_1,...,y_\secpar)$, %where $y_1,...,y_\secpar \gets G_\secpar(k)$. 
    where $y_1 \gets G_\secpar(k),...,y_\secpar \gets G_\secpar(k)$.
    \label{con:weak_botPRG_from_weak_pdprg}
\end{construct}
    We will now show the construction above is a $(\mu,\ell)$-$\botPRG$ with weak pseudorandomness $(\epsilon(\secpar)+\mu(\secpar) +\nu(\secpar)+\negl)$.
    \begin{itemize}
        \item \textbf{Pseudodeterminism:} 
        \begin{itemize}
            \item \cref{def:bot-pseudodeterminism-part1}: We define for all $\secpar$, $\Klam'=\Klam$, where $\Klam$ is the set from \cref{def:PD-PRG} for $G$. Then \cref{def:bot-pseudodeterminism-part1} holds trivially.
            \item \cref{def:bot-pseudodeterminism-part2}: By \cref{def:pseudodeterminism-part2} for G, for all $k\in\Klam'$, there exists a $y\in\{0,1\}^{\ell(\secpar)}$ such that,
            \begin{align*}
                \Pr[Y\gets G_\secpar(k) : Y = y]\geq 1- \nu(\secpar)
            \end{align*}
            We define for all $i\in [\secpar]$, the Poisson trial 
            \begin{align*}
                X_i=
                \begin{cases}
                1, \ Y_i = y
                \\ 0, \ otherwise.
                \end{cases}
            \end{align*}
            where $Y_i$ are defined as in \cref{con:weak_botPRG_from_weak_pdprg}.
            For sufficiently large $\secpar$, for all $i\in [\secpar]$,
            \begin{align*}
                 \Pr[X_i = 1] = \Pr[Y_i\gets G_\secpar(k):Y_i=y] \geq 0.9.
            \end{align*}
            %We now use a standard argument to show that $\Gvote_\secpar$ outputs $y$ with overwhelming probability: %We define $X = \sum_{i=1}^\secpar X_i$. Then,
            %$\EE[X] \geq 0.9 \secpar$.
            By \cref{thm:chernoffbound}, %for $\delta = \frac{1}{3}$,
            \begin{align}
               \Pr[Y\neq y : Y\gets \Gvote_\secpar(k)]  %= \Pr[X \leq 0.6] 
               %\\ \leq \Pr[X \leq (1 - \frac{1}{3}) \EE[X]] 
               %\leq e^{- \EE[X] 2/9} 
               \leq e^{-\secpar/5}=\negl. \label{eq:GoodVoteProb}
            \end{align}
            % \begin{align}
            %     \Pr[Y = y: Y\gets \Gvote_\secpar(k)]\geq 1 - e^{-\secpar/5} = 1 - \negl
            % \end{align}
            \item \cref{def:bot-pseudodeterminism-part3}: We use two combinatorial lemmas for this item, which are proved in \cpageref{proof:vote-50}.
            \begin{lemma}
                Let $G$ be some $\PDPRG$ with stretch $\ell(\secpar)$. Let $k\in\{0,1\}^\secpar$ and let $S_k\subset\{0,1\}^{\ell(\secpar)}$ such that 
                \begin{align}
                    \Pr\left[ Y \gets G_\secpar(k) :Y \in S_k \right]\leq 0.5. 
                    \label{eq:skprobability}
                \end{align} 
                Then $\Pr\left[Y \gets \Gvote_\secpar(k): Y\in S_k\right] \leq e^{-\secpar/100}$.
                \label{vote-50}
            \end{lemma}
            \begin{lemma}\label{lem: Set-Division}
                Let $k\in \{0,1\}^\secpar $ and $G$ be $\PDPRG$ such that for all $y\in\{0,1\}^{\ell(\secpar)}$, $\Pr\left[ Y\gets G_\secpar(k):Y = y\right]< 0.5$. Then it is possible to divide $\{0,1\}^{\ell(\secpar)}$ into 3 sets $S_1,S_2,S_3$ such that for any $i \in [3]$ it holds $\Pr\left[Y \gets G_\secpar(k): Y \in S_i\right] \leq 0.5$.
            \end{lemma}
            
            Let $k\in\{0,1\}^\secpar$, and let $y\in\{0,1\}^{\ell(\secpar)}$ such that $y$ is a most probable outcome of $G_\secpar(k)$. Consider the following cases.
            \begin{enumerate}
                \item $\Pr\left[Y\gets G_\secpar(k): Y=y\right] \geq 0.5$. 
            
                Let $S =\{0,1\}^{\ell(\secpar)}\setminus\{y\}$, then,
                \begin{align*}
                    \Pr\left[Y\gets G_\secpar(k): Y\in S\right] \leq 0.5.
                \end{align*}
                By \cref{vote-50},
                \begin{align*}
                    \Pr\left[Y\gets \Gvote_\secpar(k): Y\in S\right] \leq e^{-\secpar/100} = \negl.
                \end{align*}
                Therefore for the compliment part,
                \begin{align*}
                    \Pr\left[Y\gets \Gvote_\secpar(k) : Y\in \{y,\bot\} \right] \geq 1 - e^{-\secpar/100}=1- \negl.
                \end{align*}
                \item $\Pr\left[Y\gets G_\secpar(k): Y=y\right]< 0.5$. 
            Then for all $y'\in\{0,1\}^{\ell(\secpar)}$, it holds, 
            \begin{align*}
                {\Pr\left[Y\gets G_\secpar(k):Y=y' \right]< 0.5}.
            \end{align*}
            By \cref{lem: Set-Division}, it is possible to divide $\{0,1\}^{\ell(\secpar)}$ into 3 sets $S_1,S_2,S_3$ such that for any $i \in [3]$,
            \begin{align*}
                \Pr\left[Y\gets G_\secpar(k): Y\in S_i\right] < 0.5.
            \end{align*} 
            By \cref{vote-50},
            \begin{align*}
                \Pr\left[Y \gets \Gvote_\secpar(k) : Y\in S_1\cup S_2 \cup S_3\right]\leq 3\cdot e^{-\secpar/100}.
            \end{align*} 
            Therefore for the compliment part, 
            \begin{align}
                \Pr[Y \gets \Gvote_\secpar(k): Y = \bot]\geq 1 -3\cdot e^{-\secpar/100} =1-\negl.
            \end{align}
        \end{enumerate}
    \end{itemize} 
    \item\textbf{Pseudorandomness:} Consider the following hybrids,
    \begin{itemize}
        \item $H_0 = \Pr_{k\xleftarrow{\$} \{0,1\}^\secpar}\left[ y\gets \Gvote_\secpar(k): \adv(y)=1 \right] $
        \item $H_1 = \Pr_{k\xleftarrow{\$} \{0,1\}^\secpar}\left[y\gets G_\secpar(k) : \adv(y) =1\right]$
        \item $H_2 = \Pr_{y\xleftarrow{\$} \{0,1\}^{\ell(\secpar)}}\left[\adv(y) =1\right]$
    \end{itemize}
    
    By the $\epsilon(\secpar)$-pseudorandomness of $G$ (see \cref{def:pdprg_pseudorandomness}), for all QPT $\adv$, there exists $\secpar_0$ such that the following holds for any $\secpar\geq \secpar_0$:
    \begin{equation}
        \left|H_2 -H_1\right| \leq \epsilon(\secpar).
    \end{equation}
    By \cref{def:PD-PRG}, \cref{def:pseudodeterminism-part2}, for any $k\in\Klam$, there is a $y\in\{0,1\}^{\ell(\secpar)}$ such that, $\Pr[Y\gets G_\secpar(k): Y=y]\geq 1-\nu(\secpar)$. By \cref{eq:GoodVoteProb}, $\Pr[Y\gets \Gvote_\secpar(k): Y=y]\geq 1-\negl$.
    
    Therefore,
    \begin{align*}
        \Pr\left[\begin{matrix}
            Y_1 \gets G_\secpar(k)
            \\ Y_2 \gets \Gvote_\secpar(k)
        \end{matrix} : Y_1=Y_2 =y\right] %\geq (1 - \nu(\secpar))(1 - \negl) \\ 
        \geq 1 - \nu(\secpar) - \negl.
    \end{align*}
    
    By using the above with \cref{def:pseudodeterminism-part1}, we can switch the statement for all $k\in \Klam$, to a random $k\in \{0,1\}^\secpar$, with a correction of $\mu(\secpar)$, i.e.: 
    \begin{align*}
    \Pr_{k\xleftarrow{\$} \{0,1\}^\secpar}\left[ 
    \begin{matrix}
        Y_1 \gets G_\secpar(k)
        \\ Y_2 \gets \Gvote_\secpar(k)
    \end{matrix}:Y_1=Y_2\right] %\geq (1-\mu(\secpar))(1 - \nu(\secpar) - \negl) \\ 
    \geq 1 -\mu(\secpar) -\nu(\secpar) - \negl. 
    \end{align*}
    The statistical closeness of $G_\secpar, \Gvote_\secpar$ we have shown above implies,
    \begin{align*}
        \left|H_1 - H_0\right| &= \left|\Pr_{k\xleftarrow{\$} \{0,1\}^\secpar}\left[ Y_1 \gets G_\secpar(k) : \adv(Y_1) =1\right] - \Pr_{k'\xleftarrow{\$} \{0,1\}^\secpar}\left[ Y_2\gets \Gvote_\secpar(k'): \adv(Y_2)=1 \right]\right| \\
         &\leq \Pr_{k\xleftarrow{\$} \{0,1\}^\secpar}\left[\begin{matrix} Y_1 \gets G_\secpar(k) \\ Y_2 \gets \Gvote_\secpar(k)
        \end{matrix} : Y_1\neq Y_2\right] \leq \mu(\secpar) +\nu(\secpar) + \negl.
    \end{align*}
    Therefore,
    \begin{align*}
        \left|H_2 - H_0\right| \leq \left|H_2 - H_1\right| + \left|H_1 - H_0\right| \leq \epsilon(\secpar)+\mu(\secpar) +\nu(\secpar) + \negl.
    \end{align*}
    \end{itemize}
    \qed
\end{proof}

Next, we will prove the two lemmas which were mentioned above:

\begin{proof}[Proof of \cref{vote-50}]\label{proof:vote-50}
    Let $k\in\{0,1\}^\secpar$ and let $S_k\subset\{0,1\}^{\ell(\secpar)}$ such that \cref{eq:skprobability} holds. Recall \cref{con:weak_botPRG_from_weak_pdprg} of $\Gvote_\secpar(k)$. We define for all $i\in [\secpar]$, the Poisson trial 
    \begin{align*}
        X_i=
        \begin{cases}
        1,  & \text{if }  y_i\in S_k\\ 
        0,  & \text{otherwise,}
        \end{cases}
    \end{align*}
    where $y_i$ are defined as in \cref{con:weak_botPRG_from_weak_pdprg}.
    Therefore, by applying the Chernoff bound in \cref{thm:chernoffbound},
    \begin{align*}
         \Pr[y \gets \Gvote(k): y\in S_k] \leq e^{-\secpar/100}. 
    \end{align*}
    \qed
\end{proof}

\begin{proof}[Proof of \cref{lem: Set-Division}]
    Consider the following algorithm for division:
    \begin{itemize}
        \item Set $i:=1$, and $S_1:=\emptyset$.
        \item For each $y\in\{0,1\}^\secpar$ do:
        \begin{enumerate}
            \item If $\Pr\left[Y \gets G_\secpar(k): Y\in S_i\right] +\Pr\left[Y \gets G_\secpar(k) :Y=y\right]\leq 0.5$, then add $y$ to $S_i$.
            \item Else: 
            \begin{enumerate}
                \item $i\gets i+1.$
                \item Set $S_i=\{y\}$. 
            \end{enumerate}
        \end{enumerate}
    \end{itemize}
    
    We will next show that the algorithm will always terminate with at most $3$ sets. By contradiction, suppose the algorithm finished with $i>3$, then:
    \begin{enumerate}
        \item  $\Pr\left[ Y \gets G_\secpar(k)\right]: Y\in S_1]+\Pr\left[ Y \gets G_\secpar(k): Y\in S_2\right]>0.5$.
        \item $\Pr\left[Y \gets G_\secpar(k): Y\in S_3 \right]+\Pr\left[Y \gets G_\secpar(k): Y\in S_4\right]>0.5$.
    \end{enumerate}
    then $\Pr\left[Y \gets G_\secpar(k): Y\in S_1 \cup S_2 \cup S_3 \cup S_4 \right]>1$ since the sets are disjoint, and this is a contradiction.
    \qed
\end{proof} 

\begin{proposition}[Multi-Time $\bot-$Pseudorandomness From Weak Pseudorandomness]
    Suppose there exists $\epsilon$-pseudorandom $(\mu,\ell)$-$\botPRG$ with stretch $\ell(\secpar)>\secpar^2$ and there exist a constant $c>1$ such that $\mu(\secpar)= O(\secpar^{-c})$. Then there exists a $(\secpar\mu,\ell')$-$\botPRG$ with stretch $\ell'(\secpar^2) = \ell(\secpar)$ and multi-time $\botPR$.
    \label{prop:PropPDPRG2}
\end{proposition}

\begin{proof}
Consider the following construction,
\begin{construct}
    Let $G$ be a $\botPRG$ with stretch $\ell(\secpar)>\secpar^2$. Define the family of QPT Algorithms $G^{\bot\oplus} = \{G_\secpar^{\bot\oplus}\}_{\secpar\in\mathbb{N}}$
    and $G_\secpar^{\bot\oplus}(k_1,...,k_\secpar):=(\bot\bigoplus)^\secpar_{i=1}G_\secpar(k_i)$  where the operator $\bot\bigoplus_\secpar$ is defined by: 
    \[\bot\bigoplus_\secpar (x_1,...,x_\secpar) = \begin{cases}
        \bigoplus^\secpar_{i=1}x_i, & \text{if } \forall i\in [\secpar] \ x_i\neq\bot \\
        \bot, & \text{otherwise.}
    \end{cases}\]
    We sometimes use $a\bot\bigoplus b$ to denote $\bot\bigoplus_2 (a,b)$. 
    \label{con:XorbotPRG}
\end{construct}
We will now show that assuming the conditions in \cref{prop:PropPDPRG2} regarding $G$ hold, the construction above is a $(\secpar\mu,\ell')$-$\botPRG$ with stretch $\ell'(\secpar^2) = \ell(\secpar)$ and has multi-time $\botPR$.

\begin{itemize}
    \item \textbf{Pseudodeterminism:} 
    \begin{enumerate}
        \item By the pseudodeterminism \cref{def:bot-pseudodeterminism-part1} of $G$, there exists a set $\Klam\subset \{0,1\}^\secpar$ such that,
        \begin{align*}
            \Pr_{k\xleftarrow{\$} \{0,1\}^\secpar}[k\in\Klam] \geq 1-\mu(\secpar).
        \end{align*}
        Then by Bernoulli's inequality,
        \begin{align}
            \Pr_{k_1\xleftarrow{\$} \{0,1\}^\secpar,...,k_\secpar\xleftarrow{\$} \{0,1\}^\secpar}[k_1,...,k_\secpar \in\Klam] \geq (1 - \mu(\secpar))^\secpar \geq 1 - \secpar\mu(\secpar).
            \label{eq:secparpseudodeterminism}
        \end{align}
        We define $\mathcal{K}'_{\secpar^2} = \{(k_1,...,k_\secpar): k_1,...,k_\secpar\in \Klam\}$.
        Hence \cref{def:bot-pseudodeterminism-part1} holds for $\mathcal{K}'_{\secpar^2}$.
        \item Let $k' = k_1,...,k_\secpar \in \mathcal{K}'_{\secpar^2}$ and let $y_1,...,y_\secpar$ be the
        most probable outcomes of $G(k_1),...,G(k_\secpar)$, then by the same argument as \cref{eq:secparpseudodeterminism}, 
        \begin{align*}
            \Pr\left[\begin{matrix}
                Y_1\gets G_\secpar(k_1) \\
                \ldots \\
                Y_\secpar \gets G_\secpar(k_\secpar)
            \end{matrix}: 
            Y_1 =y_1,\ldots,Y_\secpar = y_\secpar\right]\geq 1 - \secpar\cdot\negl := 1-\widetilde\negl,
        \end{align*}
        %where $\widetilde\negl:=n \cdot \negl $. 
        This means that the most probable outcome of $G^{\bot\oplus}_\secpar(k'):=\bigoplus^\secpar_{i=1}G(k_i)$ is $y^* = \bigoplus^\secpar_{i=1} y_i$ and it holds 
        \begin{align*}
        \Pr[Y\gets G_\secpar^{\bot\oplus}(k'): Y = y^*] \geq 1 - \widetilde\negl.
        \end{align*}
        Hence \cref{def:bot-pseudodeterminism-part2} holds.
        
        \item Let $k'= (k_1,\ldots, k_\secpar) \in \{0,1\}^{\secpar^2}$, then by the pseudodeterminism of 
        $G$, $\forall i \in [\secpar]$, there exists a $y_i\in\{0,1\}^{\ell(\secpar)}$ such that, 
        \begin{align*}
            \Pr\left[Y_i \gets G_\secpar(k_i): Y_i \in \{y_i,\bot\}\right] \geq 1 - \negl.
        \end{align*} 
        Then there exists $y^*=\bigoplus^\secpar_{i=1}y_i$ such that, 
        \begin{align*}
            \Pr[Y \gets G^{\bot\oplus}_\secpar(k): Y \in \{y^*, \bot\}]\geq (1-\negl)^\secpar \geq 1- \secpar\cdot\negl = 1-\widetilde\negl.
        \end{align*}
    \end{enumerate}
        
    \item \textbf{Stretch:}
    Recall that by our assumption, the output length of $G_\secpar$, denoted $\ell(\secpar)$,  satisfies $\ell(\secpar)>\secpar^2$. The input length of $G^{\bot \oplus}_\secpar$ is $\secpar^2$ and  its output length, $\ell'(\secpar^2)=\ell(\secpar)$ and therefore $\ell'(n^2)>n^2$.
    \item \textbf{Pseudorandomness:} We will now show that $G^{\bot\oplus}$ has multi-time $\botPR$  (see \cref{def:multi-time-pseudorandomness}), by using the following hybrids:
    \begin{itemize}
        \item \textbf{$H_0$:} This hybrid is the same as the multi-time definition recast as a guessing game.
        \begin{enumerate}
            \item The challenger samples $k_1,...k_\secpar \xleftarrow{\$} \{0,1\}^\secpar$ and $y'\xleftarrow{\$}\{0,1\}^{\ell(\secpar)}$.
            \item The challenger calculates $\forall j\in[q] \ y_j\gets G^{\bot\oplus}_\secpar(k_1,...,k_\secpar)$.
            \item The challenger samples $b \xleftarrow{\$}\{0,1\}$. If $b=0$, the challenger sends the adversary $y_1,...,y_q$, and if $b=1$ then the challenger sends $\isbot(y_1,y'),...,\isbot(y_q,y')$.
            \item The adversary outputs $b'$.
            \item If $b=b'$ then the output of this hybrid is $1$.
        \end{enumerate}
        \item \textbf{$H_1$:}  Let \[p=\Pr_{k\xleftarrow{\$}\{0,1\}^\secpar}[k\in\Klam].\] 
        In this hybrid, instead of randomly sampling the keys $k_1,\ldots,k_\secpar$, the challenger first samples a Bernoulli trial with parameter $p$ that determines whether a key is from $\Klam$ or not, and then continues to sample a key accordingly. 
        Therefore, this hybrid is statistically equal to the previous hybrid. 
           
        \begin{enumerate}
            \item \sout{The challenger samples $k_1,...k_\secpar \xleftarrow{\$} \{0,1\}^\secpar$} \\
            \fbox{The challenger samples $\forall i \in[\secpar] \ x_i\gets \textbf{B}(p)$ and  $\forall i \in [\secpar]$}. \\ 
            \fbox{ The challenger samples $k_i \xleftarrow{\$} 
            \begin{cases}
                \Klam, \text{if } x_i=1 \\
                \Klam^c \coloneqq\{0,1\}^\secpar\setminus\Klam, \text{if } x_i=0
            \end{cases}$ }\\
             and $y'\xleftarrow{\$}\{0,1\}^{\ell(\secpar)}$.
           \item The challenger calculates $\forall j\in[q] \ y_j\gets G^{\bot\oplus}_\secpar(k_1,...,k_\secpar)$.
            \item The challenger samples $b \xleftarrow{\$}\{0,1\}$. If $b=0$, the challenger sends the adversary $y_1,...,y_q$, and if $b=1$ then the challenger sends $\isbot(y_1,y'),...,\isbot(y_q,y')$.
            \item The adversary outputs $b'$.
            \item If $b=b'$ then the output of this hybrid is $1$.            
        \end{enumerate}
        \item \textbf{$H_2$:}
        In this hybrid, instead of using \cref{con:XorbotPRG} we will use the given $\botPRG$ in \cref{prop:PropPDPRG2}. This hybrid is  statistically equal to the previous hybrid. 
        \begin{enumerate}
            \item The challenger samples $\forall i \in[\secpar] \ x_i\gets \textbf{B}(p)$ and  $\forall i \in [\secpar]$. \\
            The challenger samples $k_i \xleftarrow{\$} 
            \begin{cases}
                \Klam, \text{if } x_i=1 \\
                \Klam^c, \text{if } x_i=0
            \end{cases}$  \\
             and $y'\xleftarrow{\$}\{0,1\}^{\ell(\secpar)}$.
           \item \sout{The challenger calculates $\forall j\in[q] \ y_j\gets G^{\bot\oplus}_\secpar(k_1,...,k_\secpar)$.} \\
           \noindent\fbox{The challenger calculates $\forall j\in[q] \ y_j\gets \bot\bigoplus_\secpar (G_\secpar(k_1),...,G(k_\secpar))$.}
            \item The challenger samples $b \xleftarrow{\$}\{0,1\}$. If $b=0$, the challenger sends the adversary $y_1,...,y_q$, and if $b=1$ then the challenger sends $\isbot(y_1,y'),...,\isbot(y_q,y')$.
            \item The adversary outputs $b'$.
            \item If $b=b'$ then the output of this hybrid is $1$.
        \end{enumerate}
        \item \textbf{$H_3$:} In this hybrid, the evaluation on the ``good'' keys (the keys from $\Klam$) is only done once. By the pseudodeterminism property of the $\botPRG$, it is assured that the output would be the same, and not $\bot$, with overwhelming probability. Hence, the two hybrids are statistically indistinguishable.
        \begin{enumerate}
            \item The challenger samples $\forall i \in[\secpar] \ x_i\gets \textbf{B}(p)$ and  $\forall i \in [\secpar]$. \\
            The challenger samples $k_i \xleftarrow{\$} 
            \begin{cases}
                \Klam, \text{if } x_i=1 \\
                \Klam^c , \text{if } x_i=0
            \end{cases}$  \\
             and $y'\xleftarrow{\$}\{0,1\}^{\ell(\secpar)}$.
            \item \sout{The challenger calculates $\forall j\in[q] \ y_j\gets \bot\bigoplus_\secpar (G_\secpar(k_1),...,G(k_\secpar))$.} \\
            \fbox{ The challenger calculates $y^{good}\gets (\bot\bigoplus)_{i : x_i=1}G_\secpar(k_i)$.} \\
            \fbox{The challenger calculates $\forall j\in[q],\ y^{bad}_j\gets (\bot\bigoplus)_{i:x_i=0}G_\secpar(k_i)$.} \\
            \fbox{The challenger calculates $\forall j\in [q], \ y_j:= y^{good}\bot\bigoplus y^{bad}_j$}
            \item The challenger samples $b \xleftarrow{\$}\{0,1\}$. If $b=0$, the challenger sends the adversary $y_1,...,y_q$, and if $b=1$ then the challenger sends $\isbot(y_1,y'),...,\isbot(y_q,y')$.
            \item The adversary outputs $b'$.
            \item If $b=b'$ then the output of this hybrid is $1$.    

                \end{enumerate}
        \item \textbf{$H_4$:} In this hybrid the challenger samples a random string that is used instead of the outcome of the ``good'' keys.
        %This hybrid is computationally indistinguishable from the previous hybrid and it is proven by \cref{lem:mcpr-claim3}.
        The computational indistinguishability from the previous hybrid is shown in \cref{lem:mcpr-claim3}.
        \begin{enumerate}
            \item The challenger samples $\forall i \in[\secpar] \ x_i\gets \textbf{B}(p)$ and  $\forall i \in [\secpar]$. \\
            The challenger samples $k_i \xleftarrow{\$} 
            \begin{cases}
                \Klam, \text{if } x_i=1 \\
                \Klam^c, \text{if } x_i=0
            \end{cases}$ \\
                and $y'\xleftarrow{\$}\{0,1\}^{\ell(\secpar)}$.
            \item \sout{The challenger calculates $y^{good}\gets (\bot\bigoplus)_{i : x_i=1}G_\secpar(k_i)$.} \\
            \fbox{The challenger samples $y^{good}\xleftarrow{\$}\{0,1\}^{\ell(\secpar)}$} \\
            The challenger calculates $\forall j\in[q],\ y^{bad}_j\gets (\bot\bigoplus)_{i:x_i=0}G_\secpar(k_i)$. \\
            The challenger calculates $\forall j\in [q], \ y_j:= y^{good}\bot\bigoplus y^{bad}_j$
            \item The challenger samples $b \xleftarrow{\$}\{0,1\}$. If $b=0$, the challenger sends the adversary $y_1,...,y_q$, and if $b=1$ then the challenger sends $\isbot(y_1,y'),...,\isbot(y_q,y')$.
            \item The adversary outputs $b'$.
            \item If $b=b'$ then the output of this hybrid is $1$.
        \end{enumerate}
        \item \textbf{$H_5$:} \lnote{Read with Or}In this hybrid we  to separate the notion of $\bot$ from the output value. In this hybrid we act as if $G$ has a probability of $\bot$ and is completely deterministic otherwise. We use the pseudodeterminism \cref{def:bot-pseudodeterminism-part3} of $G$, to show that $\forall k\in \Klam^c, \exists y_k\in\{0,1\}^{\ell(\secpar)} \ s.t \ G_\secpar(k)\in\{\bot,y_k\}$. This hybrid is statistically indistinguishable from the previous hybrid.
        \begin{enumerate}
            \item The challenger samples $\forall i \in[\secpar] \ x_i\gets \textbf{B}(p)$ and  $\forall i \in [\secpar]$. \\
            The challenger samples $k_i \xleftarrow{\$} 
            \begin{cases}
                \Klam, \text{if } x_i=1 \\
                \Klam^c, \text{if } x_i=0
            \end{cases}$  \\
                and $y'\xleftarrow{\$}\{0,1\}^{\ell(\secpar)}$.
            \item The challenger samples $y^{good}\xleftarrow{\$}\{0,1\}^{\ell(\secpar)}$ \\
            \sout{The challenger calculates $\forall j\in[q],\ y^{bad}_j\gets (\bot\bigoplus)_{i:x_i=0}G_\secpar(k_i)$.} \\
            \fbox{The challenger calculates $\forall j\in[q] \ y^{bad}_j=(\bot\bigoplus)_{i:x_i=0}\isbot (G(k_i),y_{k_i})$} \\
            The challenger calculates $\forall j\in [q], \ y_j:= y^{good}\bot\bigoplus y^{bad}_j$
            \item The challenger samples $b \xleftarrow{\$}\{0,1\}$. If $b=0$, the challenger sends the adversary $y_1,...,y_q$, and if $b=1$ then the challenger sends $\isbot(y_1,y'),...,\isbot(y_q,y')$.
            \item The adversary outputs $b'$.
            \item If $b=b'$ then the output of this hybrid is $1$.

        \end{enumerate}        
        \item \textbf{$H_6$:} In this hybrid we say $y^{good}$, which is a random string acts as a One Time Pad to $y^{bad}$. This hybrid is statistically indistinguishable from the previous hybrid. Note that in this hybrid there is no difference between $b=0$ and $b=1$, and at this point the adversary wins with probability $1/2$.
        \begin{enumerate}
            \item The challenger samples $\forall i \in[\secpar] \ x_i\gets \textbf{B}(p)$ and  $\forall i \in [\secpar]$. \\
            The challenger samples $k_i \xleftarrow{\$} 
            \begin{cases}
                \Klam, \text{if } x_i=1 \\
                \Klam^c, \text{if } x_i=0
            \end{cases}$  \\
                and $y'\xleftarrow{\$}\{0,1\}^{\ell(\secpar)}$.
            \item The challenger samples $y^{good}\xleftarrow{\$}\{0,1\}^{\ell(\secpar)}$ \\
            The challenger calculates $\forall j\in[q] \ y^{bad}_j=(\bot\bigoplus)_{i:x_i=0}\isbot (G(k_i),y_{k_i})$
            \item The challenger samples $b \xleftarrow{\$}\{0,1\}$. If $b=0$, \sout{the challenger sends the adversary $y_1,...,y_q$} \fbox{the challenger sends the adversary $\isbot(y^{bad}_1,y),...,\isbot(y^{bad}_q,y)$}, and if $b=1$ then the challenger sends $\isbot(y_1,y'),...,\isbot(y_q,y')$.
            \item The adversary outputs $b'$.
            \item If $b=b'$ then the output of this hybrid is $1$.

        \end{enumerate}
    \end{itemize} 
\end{itemize}
    \qed
\end{proof}
The following lemmas were used in the proof above to show $H_3, H_4$ are computationally indistinguishable. 
\begin{lemma}
    For all QPT $\adv$, $|H_3-H_4|\leq \negl$.
    \label{lem:mcpr-claim3}
\end{lemma}
\begin{proof}
   %Let $G$ be a $\mu-\botPRG$ with stretch $\ell(\secpar)$ and weak pseudorandomness $\epsilon$.  
   Let $p=\Pr_{k\xleftarrow{\$}\{0,1\}^\secpar}[k\in\Klam]$ where $\Klam$ is the set from the pseudodeterminism \cref{def:bot-PRG} for G.
   We define a distribution $D$ over $2^{[\secpar]}$: 
    \begin{enumerate}
        \item $\forall i \in[\secpar], \ x_i\gets \textbf{B}(p)$
        \item $ S: = \{i:x_i=1\}$
        \item Output $S$.
    \end{enumerate}
    We define a distribution $H(S)$: 
     \begin{enumerate}
         \item $\forall i \in [\secpar], \ k_i \xleftarrow{\$} \begin{cases}
             \Klam, \text{if } i\in S \\
             \Klam^c, \text{otherwise}
         \end{cases}$
         \item $\forall i \in S, y_i\gets G_\secpar(k_i)$
         \item $y = (\bot\bigoplus)_{i\in S} y_i$
         \item $K_{bad} = \{k_i: i \notin S \}$
         \item Output $(y, K_{bad})$.
     \end{enumerate}
    Then for all QPT $\adv$,
    \begin{align*}
         \left|\Pr\left[\begin{matrix}
            S\gets D 
            \\ (y, K_{bad}) \gets H(S)
         \end{matrix}:\adv(y,S,K_{bad})=1 \right] - \Pr\left[ \begin{matrix} 
            S\gets D \\
            (y,K_{bad})\gets H(S) \\
            y'\xleftarrow{\$}\{0,1\}^{\ell(\secpar)}
         \end{matrix}: A(y',S,K_{bad})=1\right]\right| \leq \negl.
     \end{align*}
    \lnote{Read with Or}
    We complete the proof by the law of total probability, by partitioning the above to an event and it's complement. We separate the events $|S|\geq \secpar/2$ and $|S|<\secpar/2$, using $P(A) = P(B)\cdot P(A|B)+P(B^c)\cdot P(A|B^c)$,
    \begin{align*}
    \epsilon:=&\left|\Pr\left[\begin{matrix}
            S\gets D 
            \\ (y, K_{bad}) \gets H(S)
         \end{matrix}:\adv(y,S,K_{bad})=1 \right] - \Pr\left[ \begin{matrix}
            S\gets D \\
            (y,K_{bad})\gets H(S) \\
            y'\xleftarrow{\$}\{0,1\}^{\ell(\secpar)}
         \end{matrix}: A(y',S,K_{bad})=1\right]\right| \\
        =& \Bigg| \Pr[S \gets D :|S|\geq\secpar/2] \cdot \bigg(\Pr\left[\begin{matrix}
            S\gets D_{\geq\secpar/2} \\
            (y, K_{bad}) \gets H(S)
         \end{matrix}:\adv(y,S,K_{bad})=1 \right] \\
        -& \Pr\left[ \begin{matrix}
            S\gets D_{\geq\secpar/2}\\
            (y,K_{bad})\gets H(S) \\
            y'\xleftarrow{\$}\{0,1\}^{\ell(\secpar)}
         \end{matrix}: A(y',S,K_{bad})=1\right] \bigg) \\
        + &\Pr[S \gets D : |S|<\secpar/2] \cdot \bigg(\Pr\left[\begin{matrix}
            S\gets D_{<\secpar/2} \\
            (y, K_{bad}) \gets H(S)
         \end{matrix}:\adv(y,S,K_{bad})=1 \right] \\
        -& \Pr\left[ \begin{matrix}
            S\gets D_{<\secpar/2} \\
            (y,K_{bad})\gets H(S) \\
            y'\xleftarrow{\$}\{0,1\}^{\ell(\secpar)}
         \end{matrix}: A(y',S,K_{bad})=1\right]\bigg)\Bigg|,
    \end{align*}
    where $D_{\geq\secpar/2}, D_{<\secpar/2}$ represent $D$ given either $|S|\geq\secpar/2, |S|<\secpar/2$ respectively. By triangle inequality,
    \begin{align}
    \epsilon \leq & \Pr[S \gets D :|S|\geq\secpar/2] \cdot \Bigg|\Pr\left[\begin{matrix}
            S\gets D_{\geq\secpar/2}
            \\ (y, K_{bad}) \gets H(S)
         \end{matrix}:\adv(y,S,K_{bad})=1 \right] \\
        -& \Pr\left[ \begin{matrix}
            S\gets D_{\geq\secpar/2}\\
            (y,K_{bad})\gets H(S) \\
            y'\xleftarrow{\$}\{0,1\}^{\ell(\secpar)}
         \end{matrix}: A(y',S,K_{bad})=1\right] \Bigg| \\
        +& \Pr[S \gets D : |S|<\secpar/2] \cdot \Bigg|\Pr\left[\begin{matrix}
            S\gets D_{<\secpar/2} \\
            (y, K_{bad}) \gets H(S)
         \end{matrix}:\adv(y,S,K_{bad})=1 \right] \\
        -& \Pr\left[ \begin{matrix}
            S\gets D_{<\secpar/2}\\
            (y,K_{bad})\gets H(S) \\
            y'\xleftarrow{\$}\{0,1\}^{\ell(\secpar)}
         \end{matrix}: A(y',S,K_{bad})=1\right]\Bigg|
        \label{eq:claim3proof}
    \end{align}
    Since $p\geq 1-\mu(\secpar)$, for sufficiently large $\secpar$, by the Chernoff bound in \cref{thm:chernoffbound}, 
    \begin{align}
        \Pr\left[|S|\leq\secpar/2 : S \gets D\right] \leq e^{-\secpar/5} = \textsf{negl}_1(\secpar).
        \label{eq:ChernoffforS}
    \end{align}
    By \cref{lem:mcpr-claim2.5} for all QPT $\adv$, and $S\subseteq[\secpar]$ such that $|S|>\secpar/2$
    \begin{align}
         &\Bigg|\Pr\left[\begin{matrix}
            S\gets D_{\geq\secpar/2}
            \\ (y, K_{bad}) \gets H(S)
         \end{matrix}:\adv(y,S,K_{bad})=1 \right] \\
        -& \Pr\left[ \begin{matrix}
            S\gets D_{\geq\secpar/2}\\
            (y,K_{bad})\gets H(S) \\
            y'\xleftarrow{\$}\{0,1\}^{\ell(\secpar)}
         \end{matrix}: A(y',S,K_{bad})=1\right] \Bigg| \leq \textsf{negl}_2(\secpar),
             \label{eq:negl2fromclaim2.5}
    \end{align}
    for some negligible function $\textsf{negl}_2(\secpar)$.
    Using \cref{eq:ChernoffforS,eq:negl2fromclaim2.5}, we can bound $\epsilon$ in \cref{eq:claim3proof}:
    \begin{align*}
        \epsilon&\leq \Pr[S \gets D : |S|>\secpar/2 ]\cdot \textsf{negl}_2(\secpar) \\ &+ \textsf{negl}_1(\secpar) \cdot  \Bigg|\Pr\left[\begin{matrix}
            S\gets  D_{<\secpar/2} \\
            (y, K_{bad}) \gets H(S)
         \end{matrix}:\adv(y,S,K_{bad})=1 \right] \\
        -& \Pr\left[ \begin{matrix}
            S\gets  D_{<\secpar/2}\\
            (y,K_{bad})\gets H(S) \\
            y'\xleftarrow{\$}\{0,1\}^{\ell(\secpar)}
         \end{matrix}: A(y',S,K_{bad})=1\right]\Bigg| \\ &\leq \textsf{negl}_3(\secpar),
    \end{align*}
   where $\textsf{negl}_3(\secpar) = \textsf{negl}_1(\secpar)+\textsf{negl}_2(\secpar)$.

   By the data processing inequality, where if $C$ and $C'$ are computationally indistinguishable distributions, then for any QPT algorithm $\bdv$, the distributions $\bdv(C)$ and $\bdv(C')$ are also computationally indistinguishable.
   In our case, we know that for every QPT $\adv$,
   \begin{align*}
    &\left|\Pr\left[\begin{matrix}
            S\gets D \\ 
            (y, K_{bad}) \gets H(S)
         \end{matrix}:\adv(y,S,K_{bad})=1 \right] - \Pr\left[ \begin{matrix}
            S\gets D \\
            (y,K_{bad})\gets H(S) \\
            y'\xleftarrow{\$}\{0,1\}^{\ell(\secpar)}
         \end{matrix}: A(y',S,K_{bad})=1\right]\right|
    \end{align*}
   and therefore, for every QPT algorithm $\bdv$ and QPT $\adv$
   \begin{align*}
    &\left|\Pr\left[\begin{matrix}
            S\gets D \\ 
            (y, K_{bad}) \gets H(S)
         \end{matrix}:\adv(\bdv(y,S,K_{bad}))=1 \right] - \Pr\left[ \begin{matrix}
            S\gets D \\
            (y,K_{bad})\gets H(S) \\
            y'\xleftarrow{\$}\{0,1\}^{\ell(\secpar)}
         \end{matrix}: \adv(\bdv(y',S,K_{bad}))=1\right]\right| \leq \negl
    \end{align*}
    \qed
\end{proof}
We  define $\bdv_{q(\secpar)}(y,S,K_{bad})$ as follows:
\begin{enumerate}
    \item calculate $\forall j\in[q], \ y^{bad}_j\gets (\bot\bigoplus)_{k:k\in K_{bad}}G_\secpar(k)$. \\
    \item calculate $\forall j\in[q], \ y_j:= y\bot\bigoplus y^{bad}_j$.
    \item return $\isbot(y^{bad}_1,y_1),...,\isbot(y^{bad}_q,y_q)$.
\end{enumerate}
Then,
\begin{align*}
   \left|H_3 - H_4\right|  &= \bigg|\Pr\left[\begin{matrix}
        S\gets D \\ 
        (y, K_{bad}) \gets H(S)
        \end{matrix}:\adv(\bdv_{q(\secpar)}(y,S,K_{bad}))=1 \right] 
        \\&- \Pr\left[ 
        \begin{matrix}
            S\gets D \\
            (y, K_{bad})\gets H(S) \\
            y'\xleftarrow{\$}\{0,1\}^{\ell(\secpar)}
        \end{matrix}: \adv(\bdv_{q(\secpar)}(y',S,K_{bad}))=1\right]\bigg|  \leq \negl.
\end{align*}

\begin{lemma}
     Let $G$ be a weak $\epsilon$-pseudorandom $(\mu,\ell)$-$\botPRG$. For every $S \subseteq [\secpar]$ such that $|S|>\secpar/2$  and for all QPT $\adv$ the following holds:
    \begin{align*}
         \left|\Pr\left[(y, K_{bad}) \gets H(S): \adv(y,S,K_{bad})=1\right] - \Pr\left[\begin{matrix}
             (y,K_{bad})\gets H(S) \\
             y'\xleftarrow{\$}\{0,1\}^{\ell(\secpar)}
             \end{matrix} :A(y',S,K_{bad})=1\right]\right| \leq \negl.
     \end{align*}
     for some negligible function.
     \label{lem:mcpr-claim2.5}
\end{lemma}
\begin{proof}
     Let $G$ be a weak $\epsilon$-pseudorandom $(\mu,\ell)$-$\botPRG$. Let $S\subseteq [\secpar]$ such that $t:=|S|>\secpar/2$. Consider the following hybrids,
    \begin{itemize}
        \item $H_0 := \left|\Pr\left[(y, K_{bad}) \gets H(S): \adv(y,S,K_{bad})=1\right] - \Pr\left[\begin{matrix}
             (y,K_{bad})\gets H(S) \\
             y'\xleftarrow{\$}\{0,1\}^{\ell(\secpar)}
             \end{matrix} :\adv(y',S,K_{bad})=1\right]\right| $
        \item $H_1 := \left|\Pr\left[(y,K_{bad})\gets H(S):\adv(y)=1\right] - \Pr\left[y\xleftarrow{\$}\{0,1\}^{\ell(\secpar)} : \adv(y)=1\right]\right|$
    \end{itemize}
    By \cref{lem:mcpr-claim2}, $H_1\leq \negl$ for some negligible function. Suppose for the sake of contradiction the distinguishing advantage is non-negligible in $H_0$. Since $(S,K_{bad})$ are independent of $y$ there exist $(K'_{bad},S)$ we can rewrite $H_0$ in the following way,
    \begin{align*}
        \left|\EE_{S,K_{bad}}\left[\Pr\left[\begin{matrix}
        \forall i \in S, \ k_i \xleftarrow{\$} \Klam \\
        \forall i \in S, \ y_i\gets G_\secpar(k_i) \\ 
        y = (\bot\bigoplus)_{i\in S} y_i 
        \end{matrix}: \adv(y,S,K_{bad})=1\right] - \Pr\left[\begin{matrix}
             (y,K_{bad})\gets H(S) \\
             y'\xleftarrow{\$}\{0,1\}^{\ell(\secpar)}
             \end{matrix} :\adv(y',S,K_{bad})=1\right]\right]\right|.
    \end{align*}
    By our assumption, this is non-negligible, then there exist $S,K'_{bad}$ such that,
    \begin{align*}
        \Pr\left[\begin{matrix}
        \forall i \in S, \ k_i \xleftarrow{\$} \Klam \\
        \forall i \in S, \ y_i\gets G_\secpar(k_i) \\ 
        y = (\bot\bigoplus)_{i\in S} y_i 
        \end{matrix}: \adv(y,S,K'_{bad})=1\right] - \Pr\left[\begin{matrix}
             (y,K_{bad})\gets H(S) \\
             y'\xleftarrow{\$}\{0,1\}^{\ell(\secpar)}
             \end{matrix} :\adv(y',S,K'_{bad})=1\right] > \frac{1}{v(\secpar)},
    \end{align*}
    for some polynomial $v$. Hence, fixing the string $S, K'_{bad}$ given to $\adv$ the distinguishing advantage is,
    \begin{align*}
        \left|\Pr\left[(y, K_{bad}) \gets H(S): \adv(y,S,K'_{bad})=1\right] - \Pr\left[\begin{matrix}
             (y,K_{bad})\gets H(S) \\
             y'\xleftarrow{\$}\{0,1\}^{\ell(\secpar)}
             \end{matrix} :\adv(y',S,K'_{bad})=1\right]\right| > \frac{1}{v(\secpar)}
    \end{align*}
    Since $S,K'_{bad}$, is a fixed string it can be considered as a classical advice (which can be encoded as a quantum advice) \lnote{check if necessary} for $\adv$. Therefore,
    \begin{align*}
        \left|\Pr\left[(y, K_{bad}) \gets H(S): \adv(y)=1\right] - \Pr\left[\begin{matrix}
             (y,K_{bad})\gets H(S) \\
             y'\xleftarrow{\$}\{0,1\}^{\ell(\secpar)}
             \end{matrix} :\adv(y')=1\right]\right| > \frac{1}{v(\secpar)}
    \end{align*}
    By the above, $H_1$ is non-negligible, in contradiction to \cref{lem:mcpr-claim2}.

     \onote{I don't understand how Lemma 15 shows lemma 14. Is that some sort of data processing inequality? It isn't clear. The main problem is the following: Recall that in order to use the data processing inequality in the computational setting, one has to use an efficient "processor". In this case, even though $K_{bad}$ is independent of y, it is not efficiently generatable. Perhaps this could be resolved by using a classical advice? In this case, the advice could be a fixed bad set of keys.}\amit{My initial plan (where I tried to aim for a uniform reduction) was to show that for every fixed bad key $k_{bad}$, the adversary has negligible distinguishing advantage in the security game given in lemma 15, where along with the challenge $y$, we also give the fixed string $k_{bad}$ in both the cases. The proof should follows from Lemma 15, since we are only giving a fixed key, but now I am thinking this proof is itself non-uniform. 
    So I guess the simplest argument to say the following: 
    Suppose for the the sake of contradiction the distinguishing advantage is non-negligible in \Cref{lem:mcpr-claim2.5}, then there exists a fixed bad key $k_{bad}$ such that the adversary has at least the same distinguishing probability in the security game in \Cref{lem:mcpr-claim2} up to replacing $K_{bad}$ with the fixed string $k_{bad}$. Next we can use $k_{bad}$ as a classical advice encoded as a quantum advice to break \Cref{lem:mcpr-claim2}, which is a contradiction. Is that okay?}
    \qed
\end{proof}

\begin{lemma}
    Let $G$ be a weak $\epsilon$-pseudorandom $(\mu,\ell)$-$\botPRG$, then for any $t>\secpar/2$ and for all QPT $\adv$:
     \begin{align*}
         \left|\Pr\left[\begin{matrix}
             \forall i \in [t] \ k_i\xleftarrow{\$} \Klam \\
             \forall i \in [t] \ y_i\gets G_\secpar(k_i) \\
             y = (\bot\bigoplus)^t_{i=1}y_i
         \end{matrix}:\adv(y)=1\right] - \Pr\left[y\xleftarrow{\$}\{0,1\}^{\ell(\secpar)} : \adv(y)=1\right]\right| \leq \negl.
     \end{align*}
     For some $\negl$ and where $\Klam$ is the set of keys from \cref{def:bot-PRG}.
     \label{lem:mcpr-claim2}
\end{lemma}
\begin{proof}   
    Let $G$ be a weak $\epsilon$-pseudorandom $(\mu,\ell)$-$\botPRG$ and let $t>\secpar/2$. By \cref{lem:mcpr-claim1},
    \begin{align*}
        \left|  \Pr_{k\xleftarrow{\$} \Klam}\left[ y\gets G_\secpar(k): \adv(y)=1\right]-\Pr_{y\xleftarrow{\$} \{0,1\}^{\ell(\secpar)}}[\adv(y)=1]\right|\leq 2(\mu+\epsilon).
    \end{align*}
    We move to a hybrid where all $y_i$ are non-$\bot$. By \cref{def:bot-pseudodeterminism-part2} this hybrid is statistically indistinguishable from the previous one. Now $\bot\bigoplus$ can be switched with $\bigoplus$. By \cref{thm:Xor-Amp} for $\delta = 2(\mu+\epsilon)$ and $t(n) = n/2$,
    \begin{align*}
        \left|\Pr\left[\begin{matrix}
             \forall i \in [t] \ k_i\xleftarrow{\$} \Klam \\
             \forall i \in [t] \ y_i\gets G_\secpar(k_i) \\
             y = \bigoplus^t_{i=1}y_i
         \end{matrix}:\adv(y)=1\right] - \Pr\left[y\xleftarrow{\$}\{0,1\}^{\ell(\secpar)} : \adv(y)=1\right]\right| \leq \negl.
    \end{align*}
    Hence for $t>\secpar/2$,
    \begin{align*}
         \left|\Pr\left[\begin{matrix}
             \forall i \in [t] \ k_i\xleftarrow{\$} \Klam \\
             \forall i \in [t] \ y_i\gets G_\secpar(k_i) \\
             y = (\bot\bigoplus)^t_{i=1}y_i
         \end{matrix}:\adv(y)=1\right] - \Pr\left[y\xleftarrow{\$}\{0,1\}^{\ell(\secpar)} : \adv(y)=1\right]\right| \leq \negl.
     \end{align*}
     \qed
\end{proof}
\begin{lemma}
     Let $G$ be a weak $\epsilon$-pseudorandom $(\mu,\ell)$-$\botPRG$, then for all QPT $\adv$:
     \begin{align*}
        \left|  \Pr_{k\xleftarrow{\$} \Klam}\left[ y\gets G_\secpar(k): \adv(y)=1\right]-\Pr_{y\xleftarrow{\$} \{0,1\}^{\ell(\secpar)}}[\adv(y)=1]\right|\leq 2(\mu+\epsilon).
    \end{align*}
     Where $\Klam$ is the set of keys from \cref{def:bot-PRG}.
     \label{lem:mcpr-claim1}
\end{lemma}
\begin{proof} 
     By the pseudorandomness of $G$,
     \begin{align}
        &\left|  \Pr_{k\gets \{0,1\}^\secpar}\left[y\gets G_\secpar(k):\adv(y)=1\right]-\Pr_{y\gets \{0,1\}^{\ell(\secpar)}}[\adv(y)=1]\right|\leq \epsilon(\secpar). \\
        \label{eq:mcprclaim1proof}
    \end{align}
    Then we rewrite the first term,
    \begin{align*}
        &\Pr_{k\gets \{0,1\}^\secpar}\left[y\gets G_\secpar(k):\adv(y)=1\right] \\
        & = \Pr\left[k\in \Klam\right]\cdot\Pr_{k\xleftarrow{\$} \Klam}\left[y\gets G_\secpar(k):\adv(y)=1 \right] \\ 
        &+ \Pr\left[k\in \Klam^c\right]\cdot\Pr_{k\xleftarrow{\$} \Klam^c}\left[y\gets G_\secpar(k):\adv(y)=1\right].
    \end{align*}
    where $\Klam^c = \{0,1\}^\secpar \setminus \Klam$.
    Then by the reverse triangle inequality of \cref{eq:mcprclaim1proof},
    \begin{align*}
         &\epsilon(\secpar) \geq \bigg|\Pr\left[k\in \Klam\right]\cdot\Pr_{k\xleftarrow{\$} \Klam}\left[y\gets G_\secpar(k):\adv(y)=1 \right] \\
         &+ \Pr\left[k\in \Klam^c\right]\cdot\Pr_{k\xleftarrow{\$} \Klam^c}\left[y\gets G_\secpar(k):\adv(y)=1\right] - \Pr_{y\gets \{0,1\}^{\ell(\secpar)}}[\adv(y)=1]\bigg| \\ 
         &\geq \Pr\left[k\in \Klam\right]\left|\Pr_{k\xleftarrow{\$} \Klam}\left[y\gets G_\secpar(k):\adv(y)=1 \right]- \Pr_{y\gets \{0,1\}^{\ell(\secpar)}}[\adv(y)=1]\right| \\
         &- \Pr[k\in \Klam^c]\left|\Pr_{k\xleftarrow{\$} \Klam^c}\left[y\gets G_\secpar(k):\adv(y)=1\right]- \Pr_{y\gets \{0,1\}^{\ell(\secpar)}}[\adv(y)=1]\right|
    \end{align*}
     By \cref{def:bot-pseudodeterminism-part1} of the pseudodeterminism of $G$ for sufficiently large $\secpar$,
     \begin{align*}
         \epsilon(\secpar) \geq 0.5 \left|\Pr_{k\xleftarrow{\$} \Klam}\left[y\gets G_\secpar(k):\adv(y)=1 \right]- \Pr_{y\gets \{0,1\}^{\ell(\secpar)}}[\adv(y)=1]\right| - \mu\cdot 1 
     \end{align*}
     Hence,
     \begin{align*}
         \left|\Pr_{k\xleftarrow{\$} \Klam}\left[y\gets G_\secpar(k):\adv(y)=1 \right]- \Pr_{y\gets \{0,1\}^{\ell(\secpar)}}[\adv(y)=1]\right| \leq 2(\mu + \epsilon).
     \end{align*}
     \qed
\end{proof}

\subsection{XOR Amplification of Pseudorandomness}
\onote{Change title heading to something more concise.}\anote{Is it better?}
%Some claims that we need in the proof.
In the entirety of this section, we use the term pseudorandomness to denote the indistinguishability of a distribution from the uniformly random distribution. For any quantum algorithm $A$ with quantum advice $\sigma$, we denote a circuit implementation of it by a tuple $(C,\rho)$ where $C$ is a quantum circuit (with classical) and $\rho$ is a quantum state such that on input $x$, $C(\rho,x)$ produces the same distribution as $A$ on input $x$ with advice $\sigma$. All the definitions in this section that are defined concerning an algorithm with quantum advice can be naturally extended to circuit implementations, and vice versa.
The main result in this section which will be used in the proof of \cref{lem:mcpr-claim2} is as follows.

\begin{theorem}
    Let $\Xx$ be a probabilistic ensemble on $\{0,1\}^{\ell(\secpar)}$ which is not necessarily efficiently samplable. Suppose  there exists a  polynomial function $p(\cdot)$ such that for every QPT algorithm $A$ with quantum advice, \[\left|\Pr[y\gets X_\secpar: A(y)=1]-\Pr[y\xleftarrow{\$}\{0,1\}^{\ell(\secpar)}: A(y)=1]\right|\leq \delta ,\]
where $\delta \equiv 1-\frac{1}{p(\secpar)}$.
Let $t(\secpar)\equiv p(\secpar)\cdot\secpar$, and let $\Xx^t\equiv \{X^t_\secpar\}_\secpar$, 
where $X^t_\secpar$ is the $t(\secpar)$-fold independent copies of $X_\secpar$. Then, for every QPT  algorithm $B$ with quantum advice, there exists a negligible function $\negl$ such that 
\[\left|\Pr[y\gets  X^t_\secpar: B(y)=1]-\Pr[y\xleftarrow{\$}\{0,1\}^{\ell(\secpar)}: B(y)=1]\right|\leq \negl.\]
\label{thm:Xor-Amp}
\end{theorem}

\begin{proof}
We will prove the theorem in three steps. 
\begin{itemize}
\item First, we show in \Cref{thm:XOR-lemma-for-correlations} that algorithmic correlation to predicates (see \Cref{def:algorithmic-correlation}) can be amplified via the XOR operation even when the algorithms are allowed to have quantum advice, by observing a generalization of the proof in~\cite{GNW11}, who proved the same statement but only concerning classical advice. 
\item Second, we use the proof of~\cite[Lemma A.3]{ALY23} up to minor adaption, to show in \Cref{thm:ALY-equivalence-unpredictability-pseudorandomness} that pseudorandomness (i.e., indistinguishability from the uniformly random distribution) of any (arbitrary) distribution is equivalent to the next-bit unpredictability (see \cref{def:next-bit-unpredictability}). 
\item Third, we observe that next-bit unpredictability can be formulated as an algorithmic correlation with respect to a predicate, and therefore can be amplified using the first step, see \Cref{thm:non-uniform-XOR-amplification-next-bit-unpredictabiity}.
\end{itemize}
The proof of the theorem follows immediately by combining the amplification result for next-bit unpredictability from the third step (\Cref{thm:non-uniform-XOR-amplification-next-bit-unpredictabiity}) and the equivalence result from the second step (\Cref{thm:ALY-equivalence-unpredictability-pseudorandomness}).%, to conclude the theorem. %Formally, the proof follows immediately by combining \Cref{thm:non-uniform-XOR-amplification-next-bit-unpredictabiity,thm:ALY-equivalence-unpredictability-pseudorandomness}, that we state and prove next.
\end{proof}
\subsubsection{XOR Amplification for Algorithmic Correlation}\label{sec:XOR-amp-correlation}
We define the following definition of algorithmic correlation to predicate.

\begin{definition}[Algorithmic Correlation to Predicates~{\cite{GNW11}}]\label{def:algorithmic-correlation}
    For every $\secpar\in \NN$, Let $P$ be a randomized algorithm/process mapping $\{0,1\}^\secpar$ to $\{0,1\}$, and let $\Xx\equiv\{X_\secpar\}_\secpar$ be a probabilistic ensemble on $\{0,1\}^\secpar$. Then the correlation of a circuit family with quantum advice $\Cc=\{(C_\secpar,\sigma_\secpar)\}_\secpar$ with input space $\{0,1\}^n$, is defined as 
    \[c(\secpar)=\EE_{X_\secpar,C_\secpar,P}\left[(-1)^{C_n(\sigma_n,X_n)}(-1)^{P(X_n)}\right].\]
\end{definition}

In this section, we will generalize the non-uniform XOR amplification result~\cite[Lemma 2]{GNW11} in the setting of quantum advice.
The first step in achieving the generalized result would be the following generalization of the isolation lemma~\cite[Lemma 4]{GNW11}.

\begin{lemma}[Isolation lemma~{(Adapted from~\cite{GNW11})}]\label{lemma:isolation-lemma}
    Let $P_1$, $P_2$ be two predicates, $\ell:\NN\rightarrow\NN$ be a length-function such that $\ell(\secpar)\leq \secpar$. Let $P(x)\equiv P_1(y)\oplus P_2(z)$ where $x=y\|z$ and $y=\ell(|x|)$. Similarly, let $\Xx=\equiv\{X_\secpar\}_\secpar$ be a probabilistic ensemble such that the marginal on the first $\ell(\secpar)$ bits are independent of that of the rest of the bits. Let $\Yy\equiv\{Y_\secpar\}_\secpar$ and $\Zz\equiv\{Y_\secpar\}_\secpar$ be the projection of $\Xx$ onto the first $\ell(\secpar)$ and onto the last $\secpar-\ell(\secpar)$ bits respectively.

    Suppose that $\delta_1$ (respectively, $\delta_2$) is an upper bound on the correlation of families of $\Cc^{s_1,r_1}$ (respectively, $\Cc^{s_2,r_2}$) with $P_1$ over $\Yy$ ($P_2$ over $\Zz$). Then,  for every function $\epsilon:\NN\rightarrow \RR$, the function
    \[\delta(\secpar)\equiv\delta_1(\ell(\secpar))\cdot\delta_2(\secpar-\ell(\secpar)) + \epsilon(\secpar),\]
    is an upper bound for the correlation of $\Cc^{s,r}$) with $P$ over $\Xx$, where 
    \[s(\secpar)\equiv \min\left(\frac{s_1(\ell(n))}{\frac{\poly(\secpar)}{\epsilon(\secpar)}}, s_2(\secpar-\ell(\secpar))-\secpar\right),\]
    and
     \[r(\secpar)\equiv \min\left(\frac{r_1(\ell(n))}{\frac{\poly(\secpar}{\epsilon(\secpar)})}, r_2(\secpar-\ell(\secpar))-\secpar\right).\]
\end{lemma}
\begin{proof}
    We observe that the proof of~\cite[Lemma 4]{GNW11} also works for the case of quantum advice. This is because, in the reduction, any algorithm that is executed $w$ times can be simulated using $w$ copies of the quantum advice. Therefore, the blowup in the size of the quantum advice is at most the blowup in the circuit size.\anote{I do not know if this is enough.}
    \qed
\end{proof}

Next, we conclude the XOR amplification results for algorithmic correlation using the Isolation lemma (see \Cref{lemma:isolation-lemma}).

\begin{theorem}[{Adapted from~\cite{GNW11}}]\label{thm:XOR-lemma-for-correlations}
For any function $s:\NN\rightarrow\NN$ and $r:\NN\rightarrow\NN$, let $\Cc^{s,r}$ denote the set of all circuits of size at most $s(\secpar)$ with quantum advice of at most $r(\secpar)$ qubits, where $\secpar$ is the input length for the circuit.
   Let $P$ be a randomized algorithm/process mapping $\{0,1\}^n$ to $\{0,1\}$ and $\Xx\equiv\{X_\secpar\}_\secpar$ be a probabilistic ensemble on $\{0,1\}^\secpar$. For every $t:\NN\rightarrow \NN$, and $x_1,\ldots,x_{t(\secpar)}\in \{0,1\}^\secpar$ define
   \[P^{t}(x_1,x_2,\ldots,x_{t(\secpar)})=\bigoplus_{i=1}^{t(\secpar)}P(x_i),\Xx^t\equiv\{X^t_\secpar\}_\secpar,\]
   where $X^t_\secpar$ is the $t(\secpar)$-fold independent copies of $X_\secpar$.

   Then for any functions $s(\cdot),r(\cdot)$, if the correlation (see \cref{def:algorithmic-correlation}) between any quantum circuit with quantum advice in $\Cc^{s,r}$ and $P$ is at most $\delta(\secpar)$, such that $1-\delta(\secpar)$ is at least $\frac{1}{p(\secpar)}$ where $p(\secpar)\in\poly$, then for every function $\epsilon:\NN\rightarrow \NN$, the correlation between any quantum circuit with quantum advice in $\Cc^{s',r'}$ and $P^t$ is at most $\delta^t(\secpar)$, where 
   \[s'(\secpar\cdot t(\secpar))\equiv \poly[\frac{\epsilon(\secpar)}{\secpar}]\cdot s(\secpar)-\poly[\secpar\cdot t(\secpar)],\]
   \[r'(\secpar\cdot t(\secpar))\equiv \poly[\frac{\epsilon(\secpar)}{\secpar}]\cdot r(\secpar)-\poly[\secpar\cdot t(\secpar)],\]
   and
   \[ \delta^t(\secpar)\equiv  \delta(\secpar)^{t(\secpar)}+\epsilon(\secpar). \]
\end{theorem}
\begin{proof}
The proof follows by the repeated use of the isolation lemma in the same manner as in~\cite{GNW11}.
% The proof follows by the same template as in~\cite{GNW11}, but we need the following generalization of the isolation lemma~\cite[Lemma 4]{GNW11}.
\qed
\end{proof}
\subsubsection{Next-Bit-Unpredictability and Its Equivalence to Pseudorandomness.}
We define the following property of distributions, called next bit unpredictability.
\begin{definition}[Next-bit-unpredictability]\label{def:next-bit-unpredictability}
    Let $\Xx\equiv\{X_\secpar\}_\secpar$ be a probabilistic ensemble on $\{0,1\}^{\ell(\secpar)}$. For every algorithm $A$ with a quantum advice $\rho$ and $i\in [\ell(\secpar)]$, the $i^{th}$ bit unpredictability of $A$ with respect to $\Xx$ is defined as
    $\gprob^{A,\rho}_{X_\secpar,i}\equiv \Pr[y\gets  D(1^\secpar), y'\gets A(y_{0,i-1}): y'=y_{i}  ].$ %The definition can be naturally extended to circuits.  
\end{definition}
It is easy to verify that the following holds.
\begin{proposition}\label{prop:correlation-unpredictability}
Let  $\Xx\equiv\{X_\secpar\}_\secpar$ be a probabilistic ensemble on $\{0,1\}^{\ell(\secpar)}$. For every $i\in [\ell(\secpar)]$, let $P_i$ be a randomized process defined on $i$-bits as follows. For every $x\in \{0,1\}^{i-1}$, $P_i$ is conditional distribution on the $i^{th}$ bit of $X_\secpar$ conditioned on the first $i$ bits being $x$. 
    Then the correlation (see \Cref{def:algorithmic-correlation}) between any quantum circuit $C$ with a quantum advice $\rho$, and $P_i$ is the same as $2\cdot \gprob^{C,\rho}_{X_\secpar,i}-1$.  %\[2\Pr[y\gets  D(1^\secpar), y'\gets A(y_{0,i-1}): y'=y_{i}  ]-1.\]
    Moreover, for every $t$, the correlation of a quantum circuit $C'$ with a quantum advice $\rho'$ with respect to the predicate $P^t$ where $P\equiv P_i$, and where $P^t$ is defined as in \Cref{thm:XOR-lemma-for-correlations}, is the same as $2\cdot \gprob^{C,\rho}_{X^t_\secpar,i}-1$.
\end{proposition}
Next, we observe that~\cite[Lemma A.3]{ALY23} can be generalized to arbitrary (not necessarily efficiently samplable) distributions as follows.
\begin{theorem}[Equivalence of Pseudorandomness and next bit unpredictability]\label{thm:ALY-equivalence-unpredictability-pseudorandomness}
   Let $\XX\equiv\{X_\secpar\}$ be a probabilistic ensemble on $\{0,1\}^{\ell(\secpar)}$ which is not necessarily efficiently samplable. Then the following holds.
   \begin{enumerate}
    \item Suppose for every QPT algorithm $A$ with quantum advice, \[\left|\Pr[y\gets X_\secpar: A(y)=1]-\Pr[y\xleftarrow{\$}\{0,1\}^{\ell(\secpar)}: A(y)=1]\right|\leq \delta ,\] then for every $i\in[\ell(\secpar)]$, and QPT algorithm $B$,  \[\Pr[y\gets  D(1^\secpar), y'\gets B(y_{0,i-1}): y'=y_{i}  ]\leq \frac{1}{2}+\delta.\]
    \label{it:pseudorandomness}
    \item Suppose for every $i\in[\ell(\secpar)]$, and QPT algorithm $A$,  \[\Pr[y\gets  D(1^\secpar), y'\gets A(y_{0,i-1}): y'=y_{i}  ]\leq \frac{1}{2}+\delta,\]
    then  for every QPT algorithm $B$, \[\left|\Pr[y\gets D(1^\secpar): B(y)=1]-\Pr[y\xleftarrow{\$}\{0,1\}^{\ell(\secpar)}: B(y)=1]\right|\leq \ell(\secpar)\cdot\delta .\]
    \label{it:next-bit-unpredictability}
   \end{enumerate}
\end{theorem}
\begin{proof}
    The proof is the same as that of~\cite[Lemma A.3]{ALY23} up to minor adaptation.
    \qed
\end{proof}

\subsubsection{Non-Uniform XOR Amplification for Next-Bit Unpredictability}
Next, we use \Cref{thm:XOR-lemma-for-correlations} on the predicates defined in \Cref{prop:correlation-unpredictability} to get the following XOR amplification result for next-bit unpredictability.
\begin{theorem}[Non-uniform XOR amplification for next-bit unpredictability]\label{thm:non-uniform-XOR-amplification-next-bit-unpredictabiity}

Let $\Xx$ be a probabilistic ensemble on $\{0,1\}^{\ell(\secpar)}$ which is not necessarily efficiently samplable. Suppose there exists a  polynomial function $p(\cdot)$ such that for every $i\in [\ell(\secpar)]$, and for every QPT algorithm $A$ with quantum advice,  \[\Pr[y\gets  X_\secpar, y'\gets A(y_{0,i-1}): y'=y_{i}  ]\leq \frac{1}{2}+\delta(\secpar),\]
where $\delta(\secpar) = \frac{1}{p(\secpar)}$.
Let $t(\secpar)\in \Theta(\secpar)$, and let $\Xx^t\equiv \{X^t_\secpar\}_\secpar$, 
where $X^t_\secpar$ is the $t(\secpar)$-fold independent copies of $X_\secpar$. %Then there exists a negligible function $\negl\equiv p(\secpar)\cdot\left(1-\frac{1}{p(\secpar)}\right)^{\secpar \cdot p(\secpar)}$, such that for every QPT algorithm $B$ with quantum advice,
Then, Then, for every QPT  algorithm $B$ with quantum advice, there exists a negligible function $\negl$ such that 
\[\Pr[y\gets  X^t_\secpar, y'\gets B(y_{0,i-1}): y'=y_{i}  ]\leq \frac{1}{2}+\negl.\]
\end{theorem}
\begin{proof}
Suppose for the sake of contradiction, there exists a QPT adversary $A$ with a polynomial-size quantum advice $\rho$ and $i\in [\ell(\secpar)]$, such that $\gprob^{A,\rho}_{X^t_\secpar,i}\geq \frac{1}{2}+\delta'(\secpar)$ for some inverse polynomial function $\delta'(\cdot)$. Note that since $A$ runs in polynomial time and the advice state is also on a polynomial number of qubits, there exists polynomial functions $s'(\secpar),r'(\secpar)\in \poly$ such that $A$ with the quantum advice has a circuit implementation $(C',\rho')$ where $C'$ is a circuit of size $s'(\secpar)$ associated with a quantum advice state $\rho'$ of size $r'(\secpar)$.\anote{Self: Update it in the preliminaries, what a circuit implementation of an algorithm with quantum advice means. Also note that we use algorithms and circuits interchangeably since any definition concerning Algorithms in our context, can be naturally extended to the circuit implementations of the algorithms.}
Let $\epsilon(\secpar)\equiv \frac{1}{2\delta'(\secpar)}$.

Let ${s}(\secpar),{r}(\secpar)\in \poly$ be large enough polynomial functions such that,
\[s'(\secpar\cdot t(\secpar))\leq \poly[\frac{\epsilon(\secpar)}{\secpar}]\cdot s(\secpar)-\poly[\secpar\cdot t(\secpar)],\]
   \[r'(\secpar\cdot t(\secpar))\leq \poly[\frac{\epsilon(\secpar)}{\secpar}]\cdot r(\secpar)-\poly[\secpar\cdot t(\secpar)].\]
Since $s'(\secpar),r'(\secpar),t(\secpar),\frac{1}{\epsilon(\secpar)}\in \poly$, such ${s}(\secpar),{r}(\secpar)\in \poly$ exists.

Next note that by the hypothesis in the theorem, every QPT algorithm with a quantum advice, and hence any polynomial size circuit implementation $C,\rho$, satisfies that 
\[\Pr[y\gets  X_\secpar, y'\gets C(\rho,y_{0,i-1}): y'=y_{i}  ]\leq \frac{1}{2}+\delta.\]
Hence, in particular, the above holds for every $(C,\rho)\in \Cc^{s,r}$. 
Therefore by \Cref{prop:correlation-unpredictability}, for every $(C,\rho)\in \Cc^{s,r}$, the correlation of $(C,\rho)$ with respect to the predicate $P\equiv P_i$ is
\[2\Pr[y\gets  X_\secpar, y'\gets C(\rho,y_{0,i-1}): y'=y_{i}  ]-1\leq \delta,\]
 where $P_i$ is as defined in \Cref{prop:correlation-unpredictability}.

 Hence by \Cref{thm:XOR-lemma-for-correlations}, it holds that for every circuit with quantum advice in $\Cc^{s',r'}$, and in particular for $(C,\rho)$, the correlation with respect to $P^t$, is at most 
 \begin{equation}\label{eq:correlation-bound-t-case}
 \delta(\secpar)^t +\epsilon(\secpar)= \delta(\secpar)^{t(\secpar)}+\frac{\delta'}{2}<\delta,
 \end{equation}
 since by the choice of $t(\secpar)$, $\delta(\secpar)^{t(\secpar)}$ is a negligible function.

However, note that by \Cref{prop:correlation-unpredictability}, the correlation of any quantum circuit with a quantum advice $(\tilde{C},\tilde{\rho})$ with respect to the predicate $P^t$ is the same as $2\cdot \gprob^{C,\rho}_{X^t_\secpar,i}$.
Therefore by \Cref{eq:correlation-bound-t-case}, we conclude that
\[2\cdot \gprob^{C,\rho}_{X^t_\secpar,i}-1<\delta,\]
which further implies
\[\gprob^{C,\rho}_{X^t_\secpar,i}< \frac{1}{2}+\delta,\]
which is in contradiction to the assumption made at the very beginning of the proof, that $\gprob^{C,\rho}_{X^t_\secpar,i}\geq \frac{1}{2}+\delta$. %with the supposition made at the beginning of the proof concerning $A$ with the quantum advice $\rho$.
%The proof follows by combining \Cref{thm:XOR-lemma-for-correlations,prop:correlation-unpredictability}.\anote{Self: Need to expand}.
\qed
\end{proof}

\end{document}